\newif\iflong
\longtrue

\documentclass[11pt]{article}

\usepackage{hyperref}
\usepackage[utf8]{inputenc}
\usepackage[T1]{fontenc}
\usepackage{amsmath}
\usepackage{amssymb}
\usepackage{amsthm}
\usepackage{comment}
\usepackage{graphicx}
\usepackage[margin=1in,letterpaper]{geometry}

\usepackage{titlesec}
\usepackage{cleveref}
\renewcommand{\subparagraph}[1]{%
\vspace{5pt}
\noindent\textbf{#1}\space
}

\usepackage{lineno}
\usepackage[utf8]{inputenc}
\usepackage[T1]{fontenc}
\usepackage{lmodern}
\usepackage{tikz}
\usepackage{xspace}
\usetikzlibrary{calc,shapes.geometric}

\newcommand{\foot}{underlying graph\xspace}

\newcommand{\smallgadget}{
\node[square,draw] (1) at (0,0) {1};
\node (xx) at (0,-.9) {\Huge $\vdots$};
\node (xx) at (0,-2.9) {\Huge $\vdots$};
\node[square,draw] (8) at (0,-2) {8};
\node[square,draw] (20) at (0,-4) {20};
\draw[draw=black] (-.5,.5) rectangle ++(1,-5.1);
}

\usepackage{lineno}
\newcommand{\Oh}{\ensuremath{\mathcal{O}}}

\newcommand{\W}[1]{\textsf{W[#1]}}

\newcommand\NP{\textsf{NP}\xspace}

\newcommand\TC{\ensuremath{\textsf{TC}}\xspace}

\newcommand{\NCTCC}{\textsc{Nontrivial TC Subgraph}\xspace}

\newcommand{\CTCC}{\textsc{Closed TCC}\xspace}

\newcommand{\G}{\ensuremath{\mg}\xspace}

\newcommand{\DirNT}{\NCTCC}

\newcommand{\pre}{\mathrm{pre}}
\newcommand{\post}{\mathrm{post}}
\newcommand{\iin}{\mathrm{in}}
\newcommand{\conn}{\mathrm{con}}
\newcommand{\oout}{\mathrm{out}}

\newcommand{\mP}[1][]{\mathcal{P}^{#1}}

\newcommand{\helpname}[3][]{\ensuremath{{#2}_{#1}^{#3}}}

\newcommand{\vin}[1][v]{\helpname{\iin}{#1}}
\newcommand{\postin}[1][v]{\helpname[\post]{\iin}{#1}}

\newcommand{\preout}[1][v]{\helpname[\pre]{\oout}{#1}}
\newcommand{\vout}[1][v]{\helpname{\oout}{#1}}

\newcommand{\vcon}[1][v]{\helpname{\conn}{#1}}

\newcommand{\Vin}{V_\iin}

\newcommand{\Vout}{V_\oout}

\newcommand{\Vconn}{V_\conn}

\newcommand{\prob}[3]{
\smallskip

\noindent \textsc{#1}:

\noindent \textbf{Input:} #2

\noindent \textbf{Question:} #3

\smallskip
}

\newcommand{\sel}{\mathrm{sel}}
\newcommand{\con}{\mathrm{con}}
\newcommand{\val}{\mathrm{val}}

\newcommand{\mg}{\mathcal{G}}

\pgfdeclarelayer{background}
\pgfdeclarelayer{foreground}
\pgfsetlayers{background,main,foreground}
\tikzstyle{filled}=[circle,fill=black,draw=black,minimum size=5pt,inner sep=0pt]

\tikzstyle{knoten}=[circle,fill=white,draw=black,minimum size=10pt,inner sep=2pt]

\title{On the Hardness of Finding Temporally Connected\linebreak Subgraphs of Any Size}

\usepackage{authblk}

\author[1]{Arnaud Casteigts\thanks{Supported by the Swiss NSF project RECAPT (200021-236640) and French ANR project TEMPOGRAL (ANR-22-CE48-0001).}}
\author[2]{Christian Komusiewicz}
\author[2,3]{Nils Morawietz\thanks{Supported by the French ANR project TEMPOGRAL (ANR-22-CE48-0001).}}
 
\affil[1]{Department of Computer Science, University of Geneva, Switzerland\\ \texttt{arnaud.casteigts@unige.ch}}
\affil[2]{Friedrich Schiller University Jena, Institute of Computer Science, Germany\\ \texttt{c.komusiewicz@uni-jena.de, nils.morawietz@uni-jena.de}}
\affil[3]{LaBRI, Université de Bordeaux, France}

\newcommand{\keywords}[1]{\noindent\textbf{Keywords:} \nobreak\ #1}

\usepackage{setspace}
\newtheorem{theorem}{Theorem}[section]
\newtheorem{lemma}[theorem]{Lemma}
\newtheorem{corollary}[theorem]{Corollary}
\newtheorem{definition}[theorem]{Definition}
\newtheorem{observation}[theorem]{Observation}
\newtheorem{claim}[theorem]{Claim}
\crefname{section}{Section}{Sections}

\newenvironment{claimproof}
  {\begin{proof}[Proof of Claim]}
  {\end{proof}}
  
\AtBeginEnvironment{claimproof}{}
  
  \date{}

\begin{document}

\maketitle
\begin{abstract}
  Temporal graphs are graphs whose edges are present only at certain points in time. Reachability in these graphs is defined via temporal paths, in which edges are traversed in increasing chronological order. A temporal graph is temporally connected (or \TC) if every ordered pair of vertices is connected by a temporal path. When the graph itself is not \TC, a natural question is whether it admits a \TC subgraph (a.k.a.~closed temporal component) of a given size $k$. This question was one of the earliest and most studied in the field, shown to be \NP-hard by Bhadra and Ferreira in 2003.
  We strengthen this result dramatically, showing that deciding if a temporal graph admits a \TC subgraph \emph{of any size} (beyond the trivial case of a single vertex in the directed case and a single edge in the undirected case) is already \NP-hard. Our result holds for all standard temporal graph settings, answering a series of open questions in the field and strengthening several existing results. 
This sharply separates closed components from open ones (where temporal paths can travel outside the component), for which the analogous problem is trivially solvable in polynomial time. 

  As a by-product, our reductions imply that the size of the largest \TC subgraph cannot even be approximated within a factor of~$(1-\epsilon)n$ in directed graphs, and within a factor of~$(1-\epsilon)\frac{n}{2}$ in undirected graphs. 
  They also complete the complexity landscape for \TC subgraphs of size exactly $k$ when parameterized by $k$ (answering the missing ``non-strict'' case). Our results also have structural implications. In particular, they imply that there exist arbitrarily large \TC graphs without nontrivial \TC subgraphs, and that there exist \TC graphs of arbitrary girth, both facts being of independent interest.
  \end{abstract}

\medskip

\noindent
\textbf{Non-usage of AI:} 
All the results in this paper were obtained without the use of AI (no exception). We used it to get feedback in the form of simulated reviews, which led us to modify (ourselves) the presentation of some results.

  \newpage

\section{Introduction}\label{sec intro}

Temporal graphs
have received significant attention for their ability to model dynamic networks, such as transportation, communication, social, and robotic networks. Formally, a temporal graph can be represented by a labeled graph~$\G=(G,\lambda)$ where~$G=(V,E)$ is a standard finite graph (directed or undirected) called the \emph{\foot}. The labeling function
$\lambda : E \to 2^{\mathbb{N}}$ assigns non-empty and finite number of
\emph{time labels} to each edge of $E$, interpreted as presence times. The largest time label $L$ is called the \emph{lifetime} of $\G$.
Reachability in these graphs is defined in terms of paths that
traverse the edges chronologically. Precisely, a \emph{temporal path}
is a sequence $\langle(e_i, t_i)\rangle$ such that
$\langle e_i \rangle$ is a path in $G$, $\langle t_i \rangle$ is
nondecreasing, and $t_i \in \lambda(e_i)$ for all $i$ in the
sequence. The graph \G is \emph{temporally connected} (\TC, for short) if it
has at least one temporal path between all ordered pairs of vertices.

Temporal paths differ crucially from standard paths in several aspects. In an early work, Kempe, Kleinberg, and Kumar~\cite{KKK00} observed that the
temporal analog of Menger's theorem does not hold: the size of a
maximum separator may differ from the maximum number of internally
node-disjoint temporal paths, making the problem of finding $k$
node-disjoint temporal paths between two vertices \textsf{NP}-hard.
Observe that temporal paths are not composable when the first path ends after the second path starts, which makes the reachability relation not transitive, with structural and algorithmic consequences. In another seminal article, Bhadra and Ferreira~\cite{BF03} study the existence of temporal components---subsets of vertices that can reach
each other by temporal paths. They distinguish between
\emph{open components}, allowing paths to contain intermediate vertices outside
the component, and \emph{closed components}, requiring that reachability can be achieved using only paths inside the component.
In both cases, they show that the absence of transitivity implies that (inclusion-wise) maximal temporal components may overlap, and the problem of finding such a component of given size~$k$ is also \textsf{NP}-hard. Interestingly, closed components
correspond to vertex-induced subgraphs that are themselves \TC, a quite
fundamental concept.

The literature on temporal graphs is sensitive to definitional details, whose mention is unavoidable.
Some works consider temporal paths whose time labels are not only required to be nondecreasing, but strictly increasing, giving rise to so-called \emph{strict} temporal paths, connectivity, reachability, and components. Another common distinction is that of simple versus nonsimple temporal graphs, where \emph{simple} refers to the fact that every edge has a single time label. Both distinctions matter, as some of the resulting combinations are incomparable in terms of expressivity (see e.g.~\cite{CCS24} and~\cite{doering2026}). For example, it can be observed that the open and closed versions of the component problem are already \textsf{NP}-hard in the simple and strict setting, in fact even for lifetime $L=1$: for any graph $G$, the strict temporal components of the graph $\G=(G,\lambda)$ where $\lambda$ assigns label $1$ to every edge coincide with the cliques of $G$. In contrast, the component problem is trivial to solve when $L=1$ in the non-strict setting, where temporal components coincide with static components.
The above-mentioned reduction from~\cite{BF03} is more subtle and implies that both component problems are \textsf{NP}-hard also in the non-strict setting (see also Balev, Sanlaville, and Schoeters~\cite{BSS24}).

As the reader may feel, the comparison of existing results in temporal graphs is quite delicate. A convenient way to unify negative results between the strict and the non-strict settings is to consider \emph{proper} temporal graphs, where adjacent edges never share a common label (in directed temporal graphs, the requirement is that no in-arc shares a common label with an out-arc at the same node). De facto, showing that a problem is \textsf{NP}-hard on proper temporal graphs establishes at once that it is \textsf{NP}-hard for both the strict and the non-strict settings.
Going further, showing that a problem is \textsf{NP}-hard in \emph{proper and simple} (a.k.a.~happy) temporal graphs establishes at once that it is \textsf{NP}-hard for any combination of the above parameters, making negative results unified and as general as possible. As it happens, deciding the existence of open and closed components of given size $k$ is indeed \textsf{NP}-hard in happy temporal graphs~\cite{CCS24}.

Our main result is to show that deciding if a closed component (i.e. \TC subgraph) of \emph{any nontrivial size} exists in happy temporal graphs is already \textsf{NP}-hard. Before presenting this result, let us discuss a few more related works to better explain how our results relate to what is known so far and what it implies beyond the question itself. 

The systematic study of temporal components from the perspective of parameterized complexity was initiated in~\cite{CLMS24}, in the non-strict setting.
 In that work, Lopes Costa et al. show that both the open and closed versions of the component problem are \textsf{W[1]}-hard in the solution size~$k$, as well as in the lifetime $L$, both in directed and undirected temporal graphs (the above works considered only undirected graphs).
In contrast, they show that both versions of the problem are FPT when parameterized by the combination $k+L$ in undirected graphs. They also consider the question of deciding whether a given component is maximal for inclusion, obtaining an interesting separation: the open version of the component problem is polynomial-time solvable, while the closed version is hard (in this case, \textsf{coNP}-hard).

Casteigts, Morawietz, and Wolf~\cite{CMW24} also considered the parameterized complexity of finding temporal components, this time in both the strict and the non-strict setting, in terms of a new parameter called \emph{distance to transitivity}, which measures how far the reachability of the temporal graph is from being transitive.
For both settings, they show that deciding if an open component of size $k$ exists in (un)directed temporal graphs is FPT for this parameter, whereas the closed version remains \textsf{NP}-hard even when the reachability relation is only a single pair away from being transitive. This negative result is obtained for happy graphs and therefore applies equally to the strict and the non-strict settings.

Recently, Deligkas et al.~\cite{DDE+25} show that both the open and the closed versions of the problem are \textsf{NP}-hard even when the \foot has bounded treewidth, unlike the case of trees, which is polynomial-time solvable~\cite{JL04}.
 They also consider a parameter called \emph{temporal path number} (namely, the minimum number of temporal paths whose union is the input graph), for which the open version of the problem is in \textsf{XP}, whereas the closed version remains \textsf{NP}-hard even for bounded values. 
Finally, they present FPT algorithms for several combinations of such parameters. 
All the results in~\cite{DDE+25} hold in the strict and the non-strict settings, due to counting arguments for positive results, and happy reductions for negative results.

In contrast to the hardness of finding closed components, in the random setting Becker~et~al.~\cite{BCCKRRZ26} showed that TC subgraphs of size~$n- o(n)$ emerge at a sharp density threshold.

Further concepts of temporal components have been studied, which are not directly related to the present work. For example, \cite{BSS24} investigates round-trip components (providing reachability plus subsequent return) and recurrent components (realized within all time windows of a certain duration). Finally, \cite{RML20} and~\cite{baudin2024} enumerate components of a slightly different type, which are either maximal in size or in persistence.

\subsection{Our contributions}

In the above works, the main question is whether a given temporal graph contains a temporal component of size $k$ or at least~$k$, for \emph{given} $k$. Considering the strong connection between temporal components on the one hand, and cliques in static graphs on the other hand, it is not surprising that these questions turn out to be hard for large values of $k$. In this paper, we ask the more radical question of whether a \TC subgraph \emph{of any size} exists in the input graph, apart from the trivial case of a single vertex in directed graphs and two adjacent vertices in undirected graphs.

\prob{\NCTCC}{A temporal graph~$\G$.}{Does~$\G$ admit any \TC subgraph (closed component) of nontrivial size, where nontrivial means more than $1$ in the directed case, and more than $2$ in the undirected case?}

Intuitively, this question is much simpler. Indeed, for open components, it is trivially solvable in polynomial time, because open components are \emph{hereditary}: any subset of an open component remains an open component (a fact also observed in~\cite{CLMS24}). This implies, in particular, that the answer in undirected graphs is yes if and only if there exist three vertices that form an open component, which is easy to check.
In contrast, \emph{closed} components are not hereditary (not even quasi-hereditary): the existence of a closed component of size $k$ does not imply the existence of one of size $k-1$, making the closed version of the question relevant.

Our main result is the striking fact that \NCTCC is indeed \textsf{NP}-hard, for both directed and undirected temporal graphs, even in the happy setting. (Recall that in directed temporal graphs, the happy condition is that no in-arc shares a common label with an out-arc at the same node, so that no two arcs with a same label can be composed even in the non-strict case.) Precisely, we obtain the following main theorems.

\begin{theorem}
\NCTCC is \textsf{NP}-hard on happy directed temporal graphs, even if the lifetime is bounded (namely, for each~$L\ge 11$).
\end{theorem}  
\begin{theorem}
\NCTCC is \textsf{NP}-hard on happy undirected temporal graphs.
\end{theorem}
For the undirected case, this leaves open the setting with constant lifetime. 
Indeed, the problem cannot be hard for constant lifetime in the non-strict (and thus happy) setting, as an algorithm by Lopes Costa et al.~\cite{CLMS24} shows that the existence of a closed component of size at least~$k$ can be decided in FPT-time for~$k+L$, which implies that \NCTCC is FPT when parameterized by $L$ in this setting, as \NCTCC asks for a closed component of size at least~$k=3$. 
We show that all remaining settings are \textsf{NP}-hard for constant lifetime. 
\begin{theorem}
\NCTCC is \textsf{NP}-hard on simple undirected temporal graphs in the strict setting, even if the lifetime is bounded (namely, for each~$L\ge 55$).
\end{theorem}  
Thus, these results are tight in terms of the considered settings. As a by-product, our results strengthen the \textsf{W[1]}-hardness for~$k+L$ by Lopes Costa et al.~\cite{CLMS24} to a \textsf{para-NP}-hardness for the problem of determining whether a closed component of size at least~$k$ exists (except in the undirected non-strict setting, where the problem is FPT). 
Technically, the presented reductions have further strong implications for several questions related to temporal components. Namely:

\begin{enumerate}
\item\label{imp hardness question} The reduction for happy undirected temporal graphs answers the open question~4 in~\cite{CLMS24}, regarding the parameterized complexity of finding closed temporal components of size exactly~$k$ in the non-strict setting, thereby completing the complexity landscape for open and closed components when parameterized by $k$.
For all the other settings, \textsf{W[1]}-hardness for~$k$ is already known, and the fact that our hardness result considers happy undirected graphs implies in particular that it holds in the non-strict setting.
\item\label{imp maximality} For each setting of the three theorems, determining whether a vertex (in the directed case) or an edge (in the undirected case) is an inclusion-maximal closed component is already \textsf{coNP}-hard.
This strengthens the result by Lopes Costa et al.~\cite{CLMS24} that testing whether a given closed component (of unbounded size) is inclusion-maximal is \textsf{coNP}-hard. 
\item\label{imp ETH} Under the Exponential Time Hypothesis (ETH)~\cite{IPZ01}, \NCTCC on happy directed temporal graphs of constant lifetime or undirected (nonhappy) temporal graphs of constant lifetime in the strict model cannot be solved in $2^{o(n)}$~time. In other words, a simple brute-force algorithm that checks all vertex subsets is essentially optimal. For happy undirected temporal graphs, our results also imply that a $2^{o(\sqrt{n})}$-time algorithm cannot exist under ETH.
\item\label{imp Counting} 
Implication~\ref{imp ETH} holds even when promised that each yes-instance contains $2^{\Theta(n)}$~maximal closed components ($2^{\Theta(\sqrt{n})}$ in the undirected non-strict case).
This implies also that there is some fixed~$c>0$ for which one cannot approximate the number of maximal closed components within a factor of $2^{c \cdot n}$ or~$2^{c \cdot \sqrt{n}}$, respectively.
\item\label{imp approx} In each of the settings and for all fixed~$0 < \epsilon < 1$, it is \textsf{NP}-hard to approximate the size of the largest closed component (i)~within a factor of $(1-\epsilon)n$ in directed graphs and (ii)~within a factor of $(1-\epsilon)\frac{n}{2}$ in undirected graphs. (Observe that approximation for $\epsilon=0$ is trivial in both cases, by outputting a vertex in the directed case, or an edge in the undirected case.)
\end{enumerate}

On the structural side, our reductions also imply that there exist happy directed temporal graphs of constant lifetime and undirected temporal graphs of constant lifetime in the strict setting that are \TC but for which no proper \TC subgraph of nontrivial size exists. Of independent interest, we also prove the existence of happy \TC graphs of arbitrary girth.

\subsection{Technical overview}

Here, we explain the general ideas behind our reductions. In a sense, all three reductions follow the same general principle, although the level of technicality increases as the settings become less expressive. For this reason, the directed case is presented first (in Section~\ref{sec directed}) and it can serve as a basis for understanding the more complex undirected cases (described at a high level in Sections~\ref{sec gagdet} and~\ref{sec sketches}, and with more details in the appendix in Sections~\ref{sec hard happy} and~\ref{sec hard strict}). The following explanations are common to the three reductions.

Each reduction is performed from the~\NP-hard \textsc{Multicolored Clique} problem, which is also~\W1-hard for the solution size~$k$~\cite{C+15}.

\prob{Multicolored Clique (MCC)}{An undirected graph~$G$, an integer~$k$, and a $k$-partition~$(V_1 ,\dots, V_k)$ of~$V(G)$ such that~$V_i$ (called a \emph{color class}) is an independent set in~$G$ for each~$i\in [1,k]$.}{Is there a clique of size~$k$ in~$G$?}

The vertices~$V$ of the MCC instance~$I$ become a subset of vertices of the constructed temporal graph~$\mg$, which additionally contains two dedicated vertices~$\alpha$ and~$\omega$ and additional gadget vertices to ensure temporal reachability between specific vertex pairs of~$V$.

In each reduction, we show that the two vertices~$\alpha$ and~$\omega$ are part of every set of vertices~$S$ that induce a nontrivial \TC subgraph.
This then implies that~$\mg[S]$ (the subgraph of $\G$ induced by $S$) has to contain a temporal path from~$\alpha$ to~$\omega$.
We also establish a selection mechanism that, starting from such a temporal path and proceeding inductively, ensures that~$S$ contains at least one vertex of each of the color classes~$V_i$ of the MCC instance~$I$.
By an additional validation mechanism, we guarantee that vertices of different color classes can be part of the same \TC subgraph only if they are adjacent in the MCC instance.
This validation mechanism is achieved by introducing additional vertices~$\Vconn$ (two for each original vertex).
The main difficulty now is that adjacent vertices~$x$ and~$y$ from the MCC instance must reach all vertices of the connector gadget~$\Vconn$ associated with them and conversely, these vertices of the connector gadget must reach~$x$ and~$y$ (and each other). 
However, this must be done in a way that avoids that the vertices~$x$ and~$y$ plus their associated vertices of the connector gadget already form a \TC subgraph.
This construction guarantees that a nontrivial \TC subgraph in \G implies the existence of a clique of size~$k$ in~$I$.

For the other direction, we show that a clique~$X$ of size~$k$ in~$I$ implies a nontrivial \TC subgraph in~$\mg$.
To this end, we show that~$X \cup \Vconn \cup \{\alpha,\omega\}$ induces a nontrivial \TC subgraph.
Intuitively, the temporal connectivity is ensured as follows: 
all vertices of~$\Vconn$ can reach~$\alpha$ early, then follow paths from~$\alpha$ to all vertices of~$X \cup \{\omega\}$, and finally reach each vertex of~$\Vconn$ via paths starting late at~$\omega$.
Similarly, for each vertex~$x$ of~$X$, there is a temporal path from~$x$ to~$\omega$ in~$\mg[X\cup \{\omega\}]$ that can subsequently be extended by temporal paths to all vertices of~$\Vconn$.
Finally, the connectivity between distinct vertices of~$\{\alpha,\omega\} \cup X$ is ensured via paths that visit vertices of~$\Vconn$.
By construction, such paths only exist since the respective endpoints are adjacent in~$I$.

\section{Hardness on Happy Directed Temporal Graphs}\label{sec directed}

We start by presenting our reduction for the case of happy directed temporal graphs of constant lifetime.
To distinguish between a (nontrivial) \TC subgraph~$\mg[S]$ and the set of vertices~$S$ that induce this subgraph, we refer to~$S$ as a~\emph{(nontrivial) closed tcc}.

\begin{theorem}\label{hardness happy directed}
\DirNT is \NP-hard in happy directed temporal graphs of lifetime~$11$. 
\end{theorem}

\newcommand{\tableDirected}{
\begin{table}
\caption{An overview for the arcs of the temporal graph grouped by their label.}
\label{tab label dir}
\centering
\begin{tabular}{l|l}
label &  set of arcs with that label \\\hline
1 & the arc~$(\vin[\alpha],\postin[\alpha])$\\
2 & the arcs~$\{(\vout[x],\alpha),(\vin[x],\alpha)\mid x\in V\} \cup \{(\preout[\omega],\alpha),(\vout[\omega],\alpha),(\postin[\alpha],\alpha)\}$\\
3 & the arcs~$\{(x,\vout[x])\mid x\in V \cup \{\omega\}\}$\\
4 & the arcs~$(\{\alpha\} \times V_1) \cup \bigcup_{\text{even~}i\in [1,k]} (V_i \times V_{i+1})$ and each arc of the form~$(\vout[x],\vin[y])$ in~$\mg$\\
5 & the arcs~$\bigcup_{\text{even~}i\in [1,k]} (\{\omega\} \times V_i)$\\
6 & the arcs~$\bigcup_{\text{odd~}i\in [1,k]} (V_i \times \{\omega\})$\\
7 &  the arcs~$\bigcup_{\text{even~}i\in [1,k]} (V_i \times V_{i-1})$\\
8 & the arcs~$\bigcup_{i\in [2,k]} (\{\alpha\} \times V_{i})$\\
9  & the arcs~$\{(\vin[x],x)\mid x\in V \cup \{\alpha\}\}$\\
10  & the arcs~$\{(\omega,\vout[x]),(\omega,\vin[x])\mid x\in V\} \cup \{(\omega,\preout[\omega]),(\omega,\vin[\alpha]),(\omega,\postin[\alpha])\}$\\
11  & the arc~$(\preout[\omega],\vout[\omega])$
\end{tabular}
\end{table}
}

\newcommand{\proofDirHappyClaim}{
\begin{claimproof}
Note that the arc~$(\vout[x],\vin[y])$ is in~$\mg$ if and only if~$\{x,y\}$ is an edge of~$G$.
Let~$\mg' := \mg$ if~$\{x,y\}$ is not an edge of~$G$, and let~$\mg'$ be the temporal graph obtained from~$\mg$ by removing the arc~$(\vout[x],\vin[y])$, otherwise.
Hence, $\mg'$ does not contain the arc~$(\vout[x],\vin[y])$.
We show that there is no temporal path from~$x$ to~$y$ in~$\mg'$.
This then proves both parts of the statement.

Assume towards a contradiction that there is a temporal path~$P$ from~$x$ to~$y$ in~$\mg'$.
We distinguish between the choices of the first arc of~$P$.

Firstly, consider the case that~$P$ starts with the arc~$(x,\vout[x])$ with time label~$3$.
Each outgoing arc of~$\vout[x]$ with a time label larger than~$3$ ends in a vertex~$\vin[z]$ for some~$z\in V\cup \{\alpha\}$.
Note that~$z\neq y$, as~$\mg'$ does not contain the arc~$(\vout[x],\vin[y])$.
The vertex~$\vin[z]$ has only one outgoing arc with time label larger than~$3$, namely the arc~$(\vin[z],z)$ with time label~$9$.
Furthermore, $z$ has no outgoing arc with a time label larger than~$9$.
This contradicts the assumption that~$P$ can reach~$y$.

Secondly, consider the case that~$P$ starts with any arc distinct from~$(x,\vout[x])$.
We show that this implies that~$P$ consists of a single arc (contradicting that~$(x,y)$ is not an arc of~$\G$) or that~$P$ traverses an arc towards~$\omega$ at time~$6$.
To this end, let~$i$ be the index of the color class that contains~$x$.
If~$i$ is odd, then~$(x,\omega)$ is the only outgoing arc of~$x$ distinct from~$(x,\vout[x])$.
Moreover, this arc has time label~$6$.
Otherwise, that is, if~$i$ is even, then the only outgoing arcs of~$x$ distinct from~$(x,\vout[x])$ end in a vertex~$z$ of~$V_{i-1}\cup V_{i+1}$ and have a time label of~$7$ if~$z\in V_{i-1}$ and a time label of~$4$, otherwise.
As the index of the color class of~$z$ is odd, the only outgoing arc of~$z$ with a time label larger than~$3$ is the arc~$(z,\omega)$ with time label~$6$.
Hence, in the case of~$z\in V_{i-1}$, $P$ only has length 1.
In the other case, $P$~traverses an arc towards~$\omega$ at time~$6$.
By construction, the only outgoing arcs of~$\omega$ with a time label larger than~$6$ are arcs with time label~$10$ towards vertices of~$V_\con$.
As the largest label of any incoming arc of~$y$ is~$9$, this contradicts the assumption that~$P$ reaches~$y$.

As both cases lead to a contradiction, we conclude that there is no temporal path from~$x$ to~$y$ in~$\mg'$, which proves the claim.
\end{claimproof}
}

\newcommand{\proofDirectedForward}{
Let~$C$ be a clique of size~$k$ in~$G$.
Since~$G$ is~$k$-partite, for each~$i\in [1,k]$, $C$ contains exactly one vertex~$v_i$ of~$V_i$.
We show that there is a tcc of size at least two in~$\mg$.
We set~$S := V_\con \cup C$ and show that~$\mg[S]$ is a temporally connected graph.
Recall that~$S$ contains~$\alpha$ and~$\omega$, as both are from~$V_\con$.

First, we argue that all vertices of~$S\cap V = C$ can reach each other in~$\mg[S]$.
Let~$x$ and~$y$ be distinct vertices of~$C$.
By definition of the arcs in~$A_\val$, $(x,\vout[x],\vin[y],y)$ is a temporal path in~$\mg[S]$ with label sequence~$(3,4,9)$.
Here, observe that the arc~$(\vout[x],\vin[y])$ exists by the fact that~$\{x,y\}$ is an edge in~$G$.
Hence, all vertices of~$C$ can reach each other in~$\mg[S]$.

Now, consider the vertex~$\omega$.
By definition, $(\omega,\vout[\omega],\vin[\alpha],\alpha)$ and~$(\omega,\vout[\omega],\vin[v_k],v_k)$ are temporal paths in~$\mg$ (and thus in~$\mg[S]$) with label sequence~$(3,4,9)$, which allows~$\omega$ to reach both~$\alpha$ and~$v_k$.
For each~$i\in [1,k-1]$, $\omega$ can also reach the vertex~$v_i$ in~$\mg[S]$ as follows:
If~$i$ is even, then~$v_i$ is an out-neighbor of~$\omega$.
If~$i$ is odd, then~$(\omega,v_{i+1})$ is an arc of~$\mg$ with label~$5$ and~$(v_{i+1},v_i)$ is an arc with label~$7$, thus implying that~$(\omega, v_{i+1}, v_i)$ is a temporal path in~$\mg[S]$.
Hence, $\omega$ can reach all vertices of~$C\cup \{\alpha\}$ in~$\mg[S]$.
For paths towards~$\omega$, we show that each vertex of~$C$ can arrive at~$\omega$ with an arc of label~$6$.
This directly holds for~$v_i\in C$ with odd~$i$, as arc~$(v_i,\omega)$ has label~$6$.
If~$i$ is even, then~$(v_{i},v_{i+1})$ and~$(v_{i+1},\omega)$ are arcs with label~$4$ and~$6$ respectively, implying that~$v_i$ can arrive at~$\omega$ in~$\mg[S]$  with an arc of label~$6$.
Such a path also exists from~$\alpha$ to~$\omega$, namely the path~$(\alpha,v_1,\omega)$ with label sequence~$(4,6)$.
Hence, $\omega$ can reach all vertices of~$C\cup \{\alpha\}$ in~$\mg[S]$, and all vertices of~$C\cup \{\alpha\}$ can arrive at~$\omega$ with an arc of time label~$6$ in~$\mg[S]$.

Now, consider the vertex~$\alpha$. 
Each vertex~$v_i\in C$ can reach~$\alpha$ via the temporal path~$(v_i,\vout[v_i],\vin[\alpha],\alpha)$ with label sequence~$(3,4,9)$ in~$\mg[S]$.
Moreover, $\alpha$ can reach each vertex of~$C$ via the direct arc with a time label larger than~$2$.
Hence, the vertices of~$C\cup \{\alpha,\omega\}$ can reach each other in~$\mg[S]$.

It thus remains to show that all vertices of~$V_\con \setminus \{\alpha,\omega\}$ can reach all vertices of~$S$ in~$\mg[S]$, and vice versa.
Due to the arcs of~$A_\con$, for each vertex~$u\in V_\con \setminus \{\alpha,\omega\}$, 
\begin{itemize}
\item $\mg[S]$ contains a temporal path from~$u$ that reaches~$\alpha$ at time step~$2$ and 
\item there is a temporal path from~$\omega$ to~$u$ that starts at time step~$10$.\footnote{The arcs of these paths are highlighted in~\Cref{fig connectivity gadget happy} as red and blue, respectively.}
\end{itemize}
By the temporal paths from~$\alpha$ to each vertex of~$C\cup \{\omega\}$ that start with an edge of label~$4$, the first item implies that all vertices of~$V_\con \setminus \{\alpha,\omega\}$ can reach all vertices of~$C\cup \{\alpha,\omega\}$.
In particular, the vertices of~$V_\con \setminus \{\alpha,\omega\}$ reach~$\omega$ at time~$6$ via the subpath~$(\alpha,v_1,\omega)$.
Thus, by the second item, all vertices of~$V_\con$ can reach all vertices of~$V_\con\setminus \{\alpha,\omega\}$ in~$\mg[S]$.
It remains to show that all vertices of~$C$ can also reach all vertices of~$V_\con \setminus \{\alpha,\omega\}$.
This also follows by the second item (the fact that there is a temporal path from~$\omega$ to~$u$ that starts at time step~$10$) and the already argued temporal paths from each vertex of~$C$ to~$\mg[S]$ that arrive at~$\omega$ at time~$6$.
Concluding, all vertices of~$S$ can reach each other via temporal paths in~$\mg[S]$.
This implies that~$S$ is a closed tcc which clearly has size more than~$1$.
Hence, $I$ is a yes-instance of~\DirNT.
}

\newcommand{\proofDirectedBack}{
Let~$S\subseteq V(\mg)$ be a vertex set of size at least~$2$, such that~$\mg[S]$ is temporally connected.
Moreover, assume that~$S$ has minimal size among all such vertex sets.
We show that there is a clique of size~$k$ in~$G$.

Since there are no two vertices~$x$ and~$y$ for which~$\mg$ contains the arcs~$(x,y)$ and~$(y,x)$, $S$ has size at least 3. 
In the following, we will show that~$S$ contains at least one vertex of~$V_i$ for each~$i\in [1,k]$.
We do this in two steps.
Firstly, we show that~$S$ contains~$\alpha$ and~$\omega$.
Secondly, we show that this implies that~$S$ contains at least one vertex of~$V_i$ for each~$i\in [1,k]$.
To show the first step, we observe the following.

\begin{observation}\label{obs dag}
$\mg[V_\con]$ is a DAG. 
\end{observation}
Since~$\mg[S]$ is temporally connected, $\mg[S]$ is not a DAG.
Hence, $S$ contains at least one vertex of~$V(\mg)\setminus V_\con = V$.
Based on this observation, we now show that~$S$ contains~$\alpha$ and~$\omega$.

\begin{claim}\label{claim dir alpha and omega}
Both~$\alpha$ and~$\omega$ are contained in~$S$.
\end{claim}
\begin{claimproof}
By~\Cref{obs dag}, $S$ contains vertices of at least one color class~$V_i$.
We distinguish several cases.
In the first three cases, we use the fact that each vertex~$v\in S$ needs to be an intermediate vertex of at least one temporal path in~$\mg[S]$, as otherwise, all temporal paths between the vertices of~$S\setminus \{v\}$ are in~$\mg[S\setminus \{v\}]$.
This would then contradict the minimality of~$S$, as~$S$ has size at least~$3$. 

\textbf{Case 1:} $S \subseteq V_1 \cup V_\con$\textbf{.}
Recall that there are no arcs between the vertices of~$V_1$.
Consider any vertex~$v_1\in V_1\cap S$.
By construction, the only arcs incident with~$v_1$ and any vertex of~$V_\con$ are (i)~arc~$(v_1,\vout[v_1])$ with label~$3$, (ii)~arc~$(\alpha,v_1)$ with label~$4$, (iii)~arc~$(v_1,\omega)$ with label~$6$, and (iv)~arc~$(\vin[v_1],v_1)$ with label~$9$.
This implies that each temporal path in~$\mg[S]$ that has~$v_1$ as an intermediate vertex uses both arcs~$(\alpha,v_1)$ and~$(v_1,\omega)$.
As~$v_1$ needs to be an intermediate vertex of at least one temporal path in~$\mg[S]$, $S$ contains~$\alpha$ and~$\omega$.

\textbf{Case 2:} $S \subseteq V_i \cup V_\con$ for some~$i\in [2,k]$\textbf{.}
We show that this is impossible.
Recall that there are no arcs between the vertices of~$V_i$.
Consider any vertex~$v_i\in V_i\cap S$.
By construction, the only arcs incident with~$v_i$ and any vertex of~$V_\con$ are (i)~arc~$(v_i,\vout[v_i])$ with label~$3$, (ii)~arc~$(\alpha,v_i)$ with label~$8$, (iii)~arc~$(\vin[v_i],v_i)$ with label~$9$, and (iv)~either arc~$(v_i,\omega)$ with label~$6$ (if~$i$ is odd) or arc~$(\omega,v_i)$ with label~$5$.
In both cases, there is no temporal path in~$\mg[S]$ that has~$v_i$ as an intermediate vertex.
Hence, $S$ cannot be a subset of~$V_i \cup V_\con$.

\textbf{Case 3:} $S$ contains vertices of at least two distinct color classes\textbf{.}
Let~$x$ and~$y$ be vertices of~$S$ from distinct color classes.
By construction, we can assume without loss of generality that~$(x,y)$ is not an arc of~$\mg$.
However, since~$\mg[S]$ is temporally connected, there is a temporal path from~$x$ to~$y$ in~$\mg[S]$.
Due to~\Cref{dir happy paths between V}, this implies that~$\vout[x]$ and~$\vin[y]$ are in~$S$.
Thus, there needs to be a temporal path~$P$ from~$\vin[y]$ to~$\vout[x]$ in~$\mg[S]$.
Note that~$\vin[y]$ has only two outgoing arcs, namely~$(\vin[y],\alpha)$ with time label~$2$ and~$(\vin[y],y)$ with time label~$9$.
The second arc cannot be traversed by~$P$, as~$y$ has no outgoing arc with a time label larger than~$9$, which would prevent~$P$ from reaching~$\vout[x]$.
Similarly, note that~$\vout[x]$ has only two incoming arcs, namely~$(\omega,\vout[x])$ with time label~$10$ and~$(x,\vout[x])$ with time label~$3$.
Also here, the second arc cannot be traversed by~$P$, as~$x$ has no incoming arc with a time label smaller than~$3$.
Thus, $P$ traverses both arcs~$(\vin[y],\alpha)$ and~$(\omega,\vout[x])$, which implies that~$S$ contains~$\alpha$ and~$\omega$.
\end{claimproof}

We are now ready to show that $S$ contains at least one vertex of~$V_i$ for each~$i\in[1,k]$.

\begin{claim}\label{claim for each vert some color}
For each~$i\in[1,k]$, $S$ contains at least one vertex of~$V_i$.
\end{claim}
\begin{claimproof}
We show via induction over~$j$ that for each~$i\in[1,j]$, $S$ contains at least one vertex of~$V_i$.
To this end, recall that~\Cref{claim dir alpha and omega} implies that~$\alpha$ and~$\omega$ are in~$S$.

For the base case of~$j= 1$, we thus need to show that~$S$ contains at least one vertex of~$V_1$.
As~$\mg[S]$ is temporally connected, there is a temporal path~$P$ from~$\alpha$ to~$\omega$ in~$\mg[S]$.
By construction, all outgoing arcs of~$\alpha$ end in vertices of~$V$.
Hence, the first arc of~$P$ is of the form~$(\alpha,v)$ for some~$v\in V$.
We get that~$v$ is from~$V_1$, as the arc~$(\alpha,w)$ receives time label~$8$ for each vertex~$w\in V\setminus V_1$ and~$w$ has no outgoing arc with a time label larger than~$8$.
Hence, $S$ contains a vertex of~$V_1$, which establishes the base case.

For the inductive step assume that there is some~$j\in [1,k-1]$, such that for each~$i\in[1,j]$, $S$ contains at least one vertex of~$V_i$.
Let~$v_j$ be an arbitrary vertex of~$V_j\cap S$.
We show that~$S$ also contains a vertex of~$V_{j+1}$.
We distinguish between the parity of~$j$.

If~$j$ is even, then~$(v_j,\omega)$ is not an arc of~$\mg$.
Still, there has to be a temporal path~$P$ from~$v_j$ to~$\omega$ in~$\mg[S]$.
By construction, each incoming arc of~$\omega$ in~$\mg$ is of the form~$(v_\ell,\omega)$ with time label~$6$  from some vertex~$v_\ell\in V_\ell$ with odd~$\ell$.
Thus, $P$ traverses the arc~$(v_\ell,\omega)$ for such a vertex~$v_\ell$.
Now, consider the outgoing arcs of~$v_j$.
These are (i)~the arc~$(v_j,\vout[v_j])$ with time label~$3$, (ii)~the arcs towards vertices of~$V_{j-1}$ with time label~$7$, and (iii)~the arcs towards vertices of~$V_{j+1}$.
As~$P$ traverses the arc~$(v_\ell,\omega)$ at time~$6$, $P$ cannot traverse any arc of the second type previously.
Assume towards a contradiction the arc~$(v_j,\vout[v_j])$ as first arc of~$P$.
The only outgoing arcs of~$\vout[v_j]$ with a time label larger than~$3$ end in a vertex~$\vin[x]$ for some~$x\in V \cup \{\alpha\}$.
That vertex~$\vin[x]$ only has a single outgoing arc with a time label larger than~$3$, namely the arc~$(\vin[x],x)$ with time label~$9$.
This contradicts the fact that~$P$ traverses the arc~$(v_\ell,\omega)$ at time~$6$.
Thus, $(v_j,\vout[v_j])$ is also not the first arc of~$P$.
Hence, the first arc of~$P$ is an arc towards a vertex of~$V_{j+1}$.
In particular, this implies that~$S$ contains at least one vertex of~$V_{j+1}$.

If~$j$ is odd, then~$(\omega,v_j)$ is not an arc of~$\mg$.
Still, there has to be a temporal path~$P$ from~$\omega$ to~$v_j$ in~$\mg[S]$.
There are three types of outgoing arcs of~$\omega$:
(i)~the arc~$(\omega,\vout[\omega])$ with time label~$3$, (ii)~the arcs towards the vertices of~$V_\con\setminus \{\alpha,\omega,\vout[\omega]\}$ with time label~$10$, and (iii)~the arcs towards vertices~$v_\ell\in V_\ell$ with time label~$5$ for some even~$\ell$.
We show that the first arc of~$P$ needs to be of type~(iii).
Clearly, it cannot be of type~(ii), as each incoming arc of~$v_j$ has a time label of at most~$9$.
Moreover, the first arc of~$P$ cannot be~$(\omega,\vout[\omega])$ by the following facts:
By construction, each outgoing arc of~$\vout[\omega]$ with a time label larger than~$3$ ends in a vertex~$\vin[x]$ for some~$x\in V_k \cup \{\alpha\}$, and each outgoing arc of~$\vin[x]$ with a time label larger than~$3$ has time label~$9$ and ends in~$x$.
As~$j\neq k$ and $v_j\neq \alpha$, we have~$v_j \neq x$.
Moreover, no vertex of~$V_k\cup \{\alpha\}$ has an outgoing arc with a label larger than~$8$, so~$P$ could end at~$x\neq v_j$ after traversing the arc~$(\vin[x],x)$ at time~$9$.
Consequently, the first arc of~$P$ is of the third type.
That is, the first arc of~$P$ has time label~$5$ and ends in a vertex~$v_\ell\in V_\ell$ for some even~$\ell$. 
By construction, the only outgoing arc of~$v_\ell$ with a label larger than~$5$ is an arc towards a vertex of~$V_{\ell-1}$ with time label~$7$.
This implies that~$\ell = j+1$, as no vertex of~$V$ has an outgoing arc with a time label larger than~$7$.
Hence, $P$ visits at least one vertex of~$V_{j+1}$, which implies that~$S$ contains at least one vertex of~$V_{j+1}$.
\end{claimproof}

We are now ready to show that there is a clique~$C$ of size~$k$ in~$G$.
We set~$C := S \cap V$.
Due to~\Cref{claim for each vert some color}, for each~$i\in [1,k]$, $C$ contains a vertex~$v_i$ of~$V_i$.
Let~$x$ and~$y$ be distinct vertices of~$C$.
We show that~$\{x,y\}$ is an edge of~$G$.
By construction of the selection gadget, $\mg$ contains at most one of the arcs~$(x,y)$ and~$(y,x)$.
Assume without loss of generality that~$(x,y)$ is not an arc of~$\mg$.
As~$\mg[S]$ is temporally connected, there is a temporal path from~$x$ to~$y$ in~$\mg[S]$.
By \Cref{dir happy paths between V}, $\{x,y\}$ is an edge in~$G$.
This implies that~$C$ is a clique of size~$k$ in~$G$.}

\begin{proof}
As already discussed, we reduce from \textsc{Multicolored Clique}.
Let~$I=(G=(V,E),k)$ be an instance of MCC with
$k$-partition~$(V_1 ,\dots, V_k)$ of~$V$.
We assume without loss of generality that~$k$ is odd.\footnote{If~$k$ is even, simply add a new universal vertex to the graph and increase~$k$ by~$1$.}
The directed temporal graph we construct consists of a \emph{selection gadget} which enforces the selection of one vertex of each~$V_i$, a \emph{connector gadget} that establishes some temporal paths, and \emph{validation arcs} between selection gadget and connector gadget that ensure that the selected vertices of~$V$ form a clique.
The connector gadget and the selection gadget overlap on two vertices~$\alpha$ and~$\omega$.
The selection gadget is shown in~\Cref{dir sel gadget} and the connector gadget  together with some validation arcs is shown in \Cref{fig connectivity gadget happy}. 
A general description of the labels of the temporal graph is given in~\Cref{tab label dir}.

\begin{figure}
\centering
\scalebox{.8}{
\begin{tikzpicture}[xscale=1.5,square/.style={regular polygon,regular polygon sides=4,inner sep=1pt}]

\node[knoten] (a) at (-.5,0) {$\alpha$};
\node[knoten] (o) at (8.5,0) {$\omega$};

\foreach \x in {1,...,7}
\node[square,draw] (\x) at (\x,0) {$V_\x$};

\foreach \x/\y/\z in {2/1/3,4/3/5,6/5/7}{

\draw[-stealth,very thick,bend left=0] (\x) edge node[fill=white] {$7$} (\y);
\draw[-stealth,very thick,bend left=0] (\x) edge node[fill=white] {$4$} (\z);
\draw[-stealth,very thick,bend left=50] (o) edge node[fill=white] {$5$} (\x);
\draw[-stealth,very thick,bend right=50] (a) edge node[fill=white] {$8$} (\x);
}

\foreach \x in {3,5,7}{
\draw[-stealth,very thick,bend left=50] (\x) edge node[fill=white] {$6$} (o);
\draw[-stealth,very thick,bend right=50] (a) edge node[fill=white] {$8$} (\x);
}

\draw[-stealth,very thick,bend left=50] (1) edge node[fill=white] {$6$} (o);
\draw[-stealth,very thick,bend right=0] (a) edge node[fill=white] {$4$} (1);

\end{tikzpicture}
}
\caption{The selection gadget. Arcs and time labels between the vertices of~$V\cup \{\alpha,\omega\}$ in the reduction for~\Cref{hardness happy directed}, for an example instance of MCC with~$k=7$.
Each vertex~$v\in V$ has two other incident arcs not depicted here.
Namely, the arc~$(v,\vout[v])$ with time label~$3$ and the arc~$(\vin[v],v)$ with time label~$9$.}
\label{dir sel gadget}
\end{figure}

\subparagraph{Selection Gadget.}
We define
$V_\sel := \{\alpha,\omega\} \cup V$ and
$A_\sel := \{(\alpha,v) \mid v\in V\} \cup \{(v,\omega)\mid v\in V_i, i \text{ odd}\} \cup \{(\omega,v)\mid v\in V_i, i \text{ even}\} \cup  \bigcup_{\text{even}~i\in [1,k]} (V_{i} \times (V_{i-1}\cup V_{i+1}))$.

The labels of these arcs are defined as follows:
The arcs of~$\{\alpha\} \times V_{1}$ receive time label~$4$ and the arcs of~$\{\alpha\} \times (V\setminus V_1)$ receive time label~$8$.
The arcs of~$\{(v,\omega)\mid v\in V_i, i \text{ odd}\}$ receive time label~$6$ and the arcs of~$\{(\omega,v)\mid v\in V_i, i \text{ even}\}$ receive time label~$5$.
The arcs of~$\bigcup_{\text{even}~i\in [1,k]} (V_{i} \times V_{i-1})$ receive time label~$7$ and the arcs of~$\bigcup_{\text{even}~i\in [1,k]} (V_{i} \times V_{i+1})$ receive time label~$4$.

\subparagraph{Connector Gadget.}
We initialize~$V_\con := \{\alpha,\omega\}$ and~$A_\con := \emptyset$.
For each vertex~$v\in V$, we add two vertices~$\vout[v]$ and~$\vin[v]$ to~$V_\con$, two arcs~$(\omega,\vout[v])$ and~$(\omega,\vin[v])$ with time label~$10$ to~$A_\con$, and two arcs~$(\vout[v],\alpha)$ and~$(\vin[v],\alpha)$ with time label~$2$ to~$A_\con$.
Next, we add two vertices~$\preout[\omega]$ and~$\vout[\omega]$ to~$V_\con$, the arc~$(\omega,\preout[\omega])$ with time label~$10$, the arc~$(\preout[\omega],\vout[\omega])$ with time label~$11$, and the arcs~$(\preout[\omega],\alpha)$ and~$(\vout[\omega],\alpha)$ with time label~$2$.
Finally, we add the two vertices~$\vin[\alpha]$ and~$\postin[\alpha]$ to~$V_\con$, and to~$A_\con$ the two arcs~$(\omega,\vin[\alpha])$ and~$(\omega,\postin[\alpha])$ with time label~$10$, the arc~$(\vin[\alpha],\postin[\alpha])$ with time label~$1$, and the arc~$(\postin[\alpha],\alpha)$ with time label~$2$.

\begin{figure}
\centering
\scalebox{.55}{
\begin{tikzpicture}[square/.style={regular polygon,regular polygon sides=4},yscale = .7]

\normalfont

\node[knoten] (ome) at (0,-2) {$\omega$};

\node[knoten] (alp) at (0,-12) {$\alpha$};

\foreach \i/\oi/\nammme [count=\j] in {outalp/out/\omega}
{

\node[knoten] (pre\i) at (-10+\j*4,-4) {$\preout[\nammme]$};
\node[knoten] (\i) at ($(pre\i) + (0,-3)$) {$\vout[\nammme]$};
\node (post\i) at ($(\i) + (0,-3)$) {$ $};

}

\foreach \i/\oi/\nammme [count=\j] in {out/out/V}
{

\node (pre\i) at (-6+\j*4,-4) {};
\node[knoten,square,inner sep = -1] (\i) at ($(pre\i) + (0,-3)$) {$\vout[\nammme]$};
\node[] (post\i) at ($(\i) + (0,-3)$) {
};

}

\foreach \i/\oi/\nammme [count=\j] in {in/in/V}
{

\node[] (pre\i) at (-2+\j*4,-4) {
};
\node[knoten,square,inner sep = 1] (\i) at ($(pre\i) + (0,-3)$) {$\vin[\nammme]$};
\node (post\i) at ($(\i) + (0,-3)$) {};

}

\foreach \i/\oi/\nammme [count=\j] in {inome/in/\alpha}
{

\node (pre\i) at (2+\j*4,-4) {$ $};
\node[knoten] (\i) at ($(pre\i) + (0,-3)$) {$\vin[\nammme]$};
\node[knoten] (post\i) at ($(\i) + (0,-3)$) {$\postin[\nammme]$};

}

\large

\draw[-stealth,very thick,blue] (preoutalp) edge node[fill=white] {$11$} (outalp);

\draw[-stealth,very thick,red] (inome) edge node[fill=white] {$1$} (postinome);

\tikzstyle{every path}=[blue]

\draw[-stealth,very thick] (ome) edge node[fill=white] {$10$} (preoutalp);

\draw[-stealth,very thick] (ome) edge node[fill=white] {$10$} (out);

\draw[-stealth,very thick] (ome) edge node[fill=white] {$10$} (in);

\draw[-stealth,very thick] (ome) edge node[fill=white] {$10$} (inome);
\draw[-stealth,very thick] (ome) edge node[fill=white] {$10$} (postinome);

\tikzstyle{every path}=[red]

\draw[stealth-,very thick] (alp) edge node[fill=white] {$2$} (preoutalp);
\draw[stealth-,very thick] (alp) edge node[fill=white] {$2$} (outalp);

\draw[stealth-,very thick] (alp) edge node[fill=white] {$2$} (out);

\draw[stealth-,very thick,bend left=0] (alp) edge node[fill=white] {$2$} (in);

\draw[stealth-,very thick] (alp) edge node[fill=white] {$2$} (postinome);

\tikzstyle{every path}=[black!60]

\draw[-stealth,very thick] (ome) edge node[fill=white] {$3$} (outalp);
\draw[stealth-,very thick] (alp) edge node[fill=white] {$9$} (inome);

\draw[-stealth,very thick,bend right=20] (out) edge node[fill=white,near end] {$4$} (inome);
\draw[-stealth,very thick,bend right=20,dashed] (outalp) edge node[fill=white,near start] {$4$} (in);
\draw[-stealth,very thick,bend left=20] (outalp) edge node[fill=white] {$4$} (inome);
\draw[-stealth,very thick,bend left=0,dashed] (out) edge node[fill=white] {$4$} (in);

\end{tikzpicture}
}
\caption{The connector gadget plus the arcs of~$A_\val$ in gray.
Here, $\vin[V]$ and~$\vout[V]$ denote the vertex sets~$\{\vin\mid v\in V\}$ and~$\{\vout\mid v\in V\}$, respectively. The dashed arcs indicate that some but not all the potential arcs between the two respective (sets of) vertices exist.
For example, $\vout[\omega]$ only has outgoing gray arcs exactly to the vertices of~$\{\vin[x]\mid x\in V_k \cup \{\alpha\}\}$.}
\label{fig connectivity gadget happy}
\end{figure}

\subparagraph{Validation Arcs.}
We add further arcs~$A_\val$ as follows:
We initialize~$A_\val := \emptyset$ and add for each vertex~$x\in V \cup \{\omega\}$ the arc~$(x,\vout[x])$ with time label~$3$ to~$A_\val$.
Moreover, for each vertex~$y\in V \cup \{\alpha\}$, we add the arc~$(\vin[y],y)$ with time label~$9$ to~$A_\val$.
For each vertex~$v\in V$, we additionally add the arc~$(\vout[v],\vin[\alpha])$ with time label~$4$ to~$A_\val$.
Moreover, for each~$v_k\in V_k$, we add the arc~$(\vout[\omega],\vin[v_k])$ with time label~$4$ to~$A_\val$.
Finally, we add arc~$(\vout[\omega], \vin[\alpha])$ with time label~$4$ to~$A_\val$.

We now describe the remaining arcs that we add between the in- and out-vertices of the vertices of~$V$.
Let~$x$ and~$y$ be distinct vertices of~$V$.
Then, we add the arc~$(\vout[y],\vin[x])$ with time label~$4$ to~$A_\val$ if and only if~$\{x,y\}$ is an edge of~$G$.
In this way, we ensure that there is a temporal path from~$x$ to~$y$, namely~$(x,\vout[x],\vin[y],y)$ with label sequence~$(3,4,9)$.

This completes the construction.
Note that~$\mg$ is happy, that is, each arc has only a single time label and no vertex has both an in- and an out-neighbor at the same time step (see~\Cref{tab label dir}).
\tableDirected
Before we come to the correctness proof, we observe some properties about temporal paths in~$\mg$ between vertices of~$V$.

\begin{claim}\label{dir happy paths between V}
Let~$x$ and~$y$ be distinct vertices of~$V$ where~$(x,y)$ is not an arc of~$\mg$.
If~$\{x,y\}$ is not an edge of~$G$, then there is no temporal path from~$x$ to~$y$ in~$\mg$.
If~$\{x,y\}$ is an edge of~$G$, then each temporal path from~$x$ to~$y$ in~$\mg$ traverses the arc~$(\vout[x],\vin[y])$.
\end{claim}

\proofDirHappyClaim
Note that~\Cref{dir happy paths between V} implies that no tcc can contain more than one vertex per color class.

We now show the correctness of the reduction.

\subparagraph{A size-$k$ clique in~$G$ implies a nontrivial closed tcc in~$\G$.}~

\proofDirectedForward

\subparagraph{A nontrivial closed tcc in~$\G$ implies a size-$k$ clique in~$G$.}~

\proofDirectedBack

\end{proof}

\subparagraph{The implications of this reduction.}
We now discuss how the implications stated in~\Cref{sec intro} follow from this reduction.
Implication~\ref{imp maximality} for this setting follows because in the family of constructed instances by these reductions, $\alpha$ is part of every closed tcc of size at least~$2$.
Hence, it is \textsf{coNP}-hard to decide whether the set~$\{\alpha\}$ is an inclusion-maximal closed tcc.
Next, Implication~\ref{imp ETH} for this setting follows because in the reduction, the number of vertices~$n$ in~$\mg$ is~$\Oh(|V|)$.
An algorithm that solves~\NCTCC on happy directed temporal graphs in $2^{o(n)}$~time would imply an algorithm for~MCC that runs in~$2^{o(|V|)}$~time, which is impossible if the ETH is true~\cite{C+15}.

Further, Implication~\ref{imp Counting} follows by assuming that the initial instance of~MCC contains at least $2^{\Theta(|V|)}$ multicolored cliques in case of a yes-instance.
That can be ensured by adding to an instance of MCC $\frac{|V|}{2}$~additional color classes of size~$2$, for which each newly added vertex is adjacent to all other vertices of the graph besides the other vertex of its color class, and by increasing~$k$ by~$\frac{|V|}{2}$.
These new color classes contain $2^{\frac{|V|}{2}}$~maximal cliques and each such clique can be extended to a multicolored clique for the whole instance, if the initial instance contained at least one multicolored clique.
Consequently, there is some fixed~$c>0$ such that it is \NP-hard to approximate the number of (maximal) \TC subgraphs within a factor of~$2^{c\cdot n}$, as a no-instance only contains~$\Oh(n)$ \TC subgraphs.
That each unique size-$k$ clique~$C$ implies a unique maximal closed tcc follows from the proof of the forwards direction, which shows that~$S:= V_\con \cup C = C \cup (V(\G)\setminus V)$ is a closed tcc (the maximality follows from~\Cref{dir happy paths between V}).

Finally, Implication~\ref{imp approx} for this setting can be obtained by slightly extending the reduction.
For any fixed~$0<\epsilon<1$, we can add vertices~$q$ to the connector gadget with arc~$(\omega,q)$ with label~$10$ and arc~$(q,\alpha)$ with label~$2$, until the size of the connector gadget is at least~$(1-\epsilon)|V^*|$ where~$V^*$ is the vertex set of the current graph.
The resulting temporal graph then has a nontrivial closed tcc if and only if the original graph had one.
As we showed, the connector gadget is part of every maximal nontrivial closed tcc, just due to the arcs incident with~$\alpha$ and~$\omega$.
Hence, the resulting temporal graph has a \TC subgraph of size at least~$(1-\epsilon)n$ if and only if it has a nontrivial \TC subgraph.
This implies that it is \NP-hard to decide whether the largest \TC subgraph has size (i) at most~$1$ or (ii)~at least~$(1-\epsilon)n$.

\section{Structural Analysis and Gadgets for the Undirected Case}\label{sec gagdet}
In this section, we prepare the gadget for the undirected case.
A simple adaptation of the directed reduction is not possible.
In fact, if we would make each arc in the previous reduction undirected, we would immediately get triangles\footnote{For example $\{\omega,\vout[\omega],\vin[\alpha]\}$ would induce a triangle (see~\Cref{fig connectivity gadget happy}), if each arc was replaced by an edge.} in the \foot, which are (no matter the labels of the  edges) closed tccs of nontrivial size. 
Hence, to obtain a gadget for the undirected case that mimics the behavior of the directed connectivity gadget, we aim to replace each vertex~$\vin[x]$ and~$\vout[x]$ by a constant-size temporal graph~$\mP$ that contains no nontrivial closed tcc.
This in particular implies that~$\mP$ must be triangle-free. Moreover, we would also like to avoid cycles of length 4: these become temporally connected as soon as the two earliest edges are on opposite sides of the cycle. This complicates proofs that undesired subgraphs are not \TC. Thus, we want to construct some suitable~$\mP$ of higher girth.    
Finally, as the vertices of such a gadget~$\mP$ need to reach~$\alpha$ early and~$\omega$ needs to reach all vertices of~$\mP$ late, $\mP$ also needs two disjoint spanning trees. One of these trees will be an in-tree with a root that can be temporally reached from all vertices via edges of that tree and is adjacent to~$\alpha$, the other one will be an out-tree with a root that is adjacent to~$\omega$ and temporally reaches all other vertices via edges of that tree.

In the first subsection, we show that \TC graphs with these properties exist.
In fact, we show the existence of \TC graphs of arbitrary girth and where additionally both spanning trees are Hamiltonian paths. 
For one such graph, we will then show that reversing the chronological order of one Hamiltonian path gives us the desired gadget.

\subsection{Happy TC graphs of arbitrary girth}\label{sec tc girth}

\begin{theorem}[$\star$]\label{arb girth}
Let~$g\geq 4$.
There are undirected happy temporal graphs consisting of two temporal Hamiltonian paths that are \TC and whose \foot has girth at least~$g$.
\end{theorem}

\newcommand{\textGirth}{
First, observe that if a graph~$G$ contains two edge-disjoint spanning trees, then there is a happy labeling~$\lambda$ of the edges of~$G$, such that~$\mg = (G,\lambda)$ is temporally connected.
So, to prove~\Cref{arb girth}, it suffices to show that there are graphs of girth~$g$ that have two edge-disjoint spanning trees.
In particular, we show that this is true even for edge-disjoint Hamiltonian paths.
We do this by combining two results about random graphs.

\begin{lemma}[{\cite[Corollary 1]{MW04}}]\label{prob girth}
Let~$g\geq 4$.
The probability that a random $4$-regular graph on~$n$ vertices has girth (at least)~$g$ tends towards
$$\mathsf{exp}\left({-\sum_{\ell=3}^{g-1} \frac{3^\ell}{2\ell} + o(1)}\right)$$
as~$n$ goes towards~$\infty$.
\end{lemma}

\begin{lemma}[{\cite[Corollary 4.9]{W99}}]\label{prob ham cycles}
The probability that a random $4$-regular graph on~$n$ vertices decomposes in two edge-disjoint Hamiltonian cycles tends towards~$1$ as~$n$ goes towards~$\infty$.
\end{lemma}

In combination, we conclude~\Cref{arb girth}.

\begin{proof}[Proof of~\Cref{arb girth}]
Due to~\Cref{prob girth}, there is some finite~$n_g > 0$, such that for each~$n_g' \geq n_g$, the probability that a random $4$-regular graph on~$n_g'$ vertices has girth (at least)~$g$ is at least $$c_g := \frac{1}{2}\cdot \mathsf{exp}\left({-\sum_{\ell=3}^{g-1} \frac{3^\ell}{2\ell}}\right).$$
Note that~$c_g$ is larger than~$0$ and independent from~$n_g$.
Now, by \Cref{prob ham cycles}, there is some finite~$n_c > 0$, such that for each~$n_c' \geq n_c$, the probability that a random $4$-regular graph on~$n_c'$ vertices contains two edge-disjoint Hamiltonian cycles is larger than~$1-c_g$.
Hence, for all~$n \geq \max(n_c,n_g)$, the probabilities that a random $4$-regular graph on~$n$ vertices (i)~has girth (at least)~$g$ and (ii)~contains two edge-disjoint Hamiltonian cycles sum up to more than~$1$.
As a consequence, these two random events cannot be disjoint, which implies that there are $4$-regular graphs~$G$ on~$n$ vertices that have girth (at least)~$g$ and contain two edge-disjoint Hamiltonian cycles.
This completes the proof, as~$G$ contains two edge-disjoint spanning trees.
\end{proof}
}

For girth~$5$, an example of a happy temporally connected graph can be obtained from~\Cref{fig gadget}, if the labeling of the red path is reversed (that is, it goes from vertex~$20$ to vertex~$1$ in chronological order starting with label~$1$ and ending with label~$19$). 

\subsection{The connector path gadget}

We now present the gadget~$\mP$ for the undirected case.
As mentioned above, our gadget~$\mP$ needs to be triangle-free. 
To simplify the argument that~$\mP$ and specific subgraphs of the temporal graph that we construct in our reductions do not contain closed tccs of size more than~$2$, we combine properties of  extremal graph theory with properties of \TC graphs. 

\newcommand{\proofGirthImplication}{
\begin{lemma}[{\cite[Table 1]{GKL93}}]\label{atleats17}
Let~$G$ be a graph on~$n>2$ vertices that has girth at least~$5$.
If~$G$ has at least~$2n-3$ edges, then~$n\geq 17$.
If~$G$ has at least~$2n-2$ edges, then~$n\geq 18$.
\end{lemma}

Note that this implies that for each~$S\subseteq V(\G)$ of size at least~$3$, for which~$\G[S]$ is \TC under strict reachability and has one label per edge, $\G[S]$ has at least~$17$ vertices, as a \TC graph without a cycle of length~$4$ needs at least~$2n-3$ time-edges\footnote{Here, the set of \emph{time-edges} is defined as~$\{(e,t)\mid e\in E(\mg), t\in \lambda(e)\}$.
Note that a simple temporal graph has as many edges as it has time-edges.}
\cite{Bumby81,KMMS24}.
\begin{lemma}[\cite{Bumby81,KMMS24}]\label{tccminedges}
A happy undirected temporally connected graph with~$n$ vertices and no cycle of length~$4$ has at least~$2n-3$ edges.
\end{lemma}

In combination, this implies~\Cref{tccimplication}.
}

\proofGirthImplication

\begin{corollary}
\label{tccimplication}
Let~$\G$ be a happy undirected temporal graph where the \foot has girth at least~$5$.
Then, each nontrivial closed tcc of~$\G$ has size at least~$17$.
\end{corollary}

With~\Cref{tccimplication}, we can show that the temporal graph depicted in~\Cref{fig gadget} contains no nontrivial tcc, even though it contains two (edge-disjoint) temporal Hamiltonian~$(1,20)$-paths.
For integers~$i$ and~$j$ with~$i\leq j$, we denote with~$[i,j]$ the set of all integers~$k$ with~$i\leq k \leq j$.

\begin{definition}[Connector path]\label{def con path}
Let~$\mP$ be the undirected happy temporal graph with vertex set~$[1,20]$ where the edges set is partitioned into the two edge-disjoint Hamiltonian~$(1,20)$-paths
\begin{itemize}
\item the path~$(1,2,\dots,20)$ (called~\emph{red path}) and
\item the path~$(1,5,9,13,2,7,11,19,14,6,17,10,3,15,8,18,4,12,16,20)$ (called~\emph{blue path}).
\end{itemize}
We say that $\mP$ is a~\emph{connector path}, if the labels of the edges of the red (blue, respectively) path are consecutively increasing and all labels assigned to edges of the blue path are strictly larger than all labels assigned to edges of the red path.
\end{definition}

Note that a connector path contains no cycle of length 3 or 4.

\begin{figure}
\centering
\scalebox{.6}{
\begin{tikzpicture}[yscale=.9]

\foreach \x in {1,...,20}
\node[knoten] (\x) at (\x,0) {\x};

\foreach \x [count=\xi] in {2,...,20}
\draw[thick,-,red] (\xi) to (\x);

\xdef\y{1}
\foreach \x [count = \xi from 41]in {5,9,13,2,7,11,19,14,6,17,10,3,15,8,18,4,12,16,20}{

\draw[blue,thick,-, bend left = 39] (\y) edge node[black,fill=white] {\xi} (\x);
\xdef\y{\x}
}

\end{tikzpicture}
}
\caption{A connector path.
The labels on the red path (not depicted) are defined as the smaller index of the endpoints of that edge.
A blue edge is above the vertices, if the left endpoint (with smaller index) precedes the right endpoint in the blue path; otherwise it is below the vertices.}
\label{fig gadget}
\end{figure}

\begin{lemma}\label{gadget no tcc}
A connector path contains no nontrivial closed tcc. 
\end{lemma}
\newcommand{\proofGadget}{
\begin{proof}
Let~$\G$ be a connector path (see~\Cref{fig gadget} for an illustration).
Assume towards a contradiction that there is some vertex set~$S\subseteq [1,20]$ of size at least~$3$, such that~$\G[S]$ is temporally connected.

We show that for each~$(s,t) \in \{17,18,19,20\} \times \{1,2,5,7\}$, there is no temporal path from~$s$ to~$t$ in~$\G$ and thus also not in~$\G[S]$.
This is due to the fact that in the red path, no vertex of~$\{17,18,19,20\}$ can reach any vertex of~$[1,15]$.
Moreover, in the blue path, each vertex of~$\{1,2,5,7\}$ precedes~$11$ and~$11$ precedes each of~$\{16,17,18,19,20\}$.
Thus, for each~$(s,t) \in \{17,18,19,20\} \times \{1,2,5,7\}$, there is no temporal path from~$s$ to~$t$ in~$\G$.
As~$\G[S]$ is temporally connected, this implies that if~$S$ contains at least one vertex of~$\{1,2,5,7\}$, then~$S$ contains no vertex of~$\{17,18,19,20\}$, and vice versa.
Thus, $S$ has size at most~$16$.

However, since~$\G$ is happy and contains no cycle of length~$3$ or~$4$,  
\Cref{tccimplication} implies~$|S| \geq 17$, leading to a contradiction. 
\end{proof}
}
\proofGadget

\newcommand{\proofExtensionGadgets}{
We also show that some specific temporal graphs obtained from connector paths contain no nontrivial closed tccs.
Essentially, these extensions will occur as subgraphs in our reductions.

\begin{lemma}\label{gadget extensions no tcc}
Let~$\mP$ be a connector path, where~$r$ is the largest label of any edge of the red path and let~$b$ be the smallest label of any edge of the blue path.
The following temporal graphs contain no nontrivial closed tcc:
\begin{enumerate}
\item \label{gadget plus alpha} the temporal graph~$\G_1$ obtained from $\mP$ by adding a vertex~$\alpha$ and edges~$\{20,\alpha\}$ and~$\{8,\alpha\}$ with labels fulfilling~$r < \lambda(\{20,\alpha\}) < \lambda(\{8,\alpha\})<b$,
\item\label{gadget plus alpha plus x1} the temporal graph~$\G_2$ obtained from $\mP$ by adding a vertex~$\alpha$, a vertex~$\alpha'$, and edges~$\{20,\alpha\}$, $\{\alpha,\alpha'\}$, and~$\{\alpha',8\}$ with labels fulfilling~$r < \lambda(\{20,\alpha\}) < \lambda(\{\alpha,\alpha'\})<\lambda(\{\alpha',8\})<b$,
\item\label{gadget plus omega} the temporal graph~$\G_3$ obtained from $\mP$ by adding a vertex~$\omega$ and edges~$\{\omega,1\}$ and~$\{\omega,8\}$ with labels fulfilling~$r < \lambda(\{\omega,8\}) < \lambda(\{\omega,1\})<b$, and
\item\label{gadget plus omega plus xk} the temporal graph~$\G_4$ obtained from $\mP$ by adding a vertex~$\omega$, a vertex~$\omega'$, and edges~$\{\omega,1\}$, $\{\omega',\omega\}$, and~$\{\omega',8\}$ with labels fulfilling~$r < \lambda(\{\omega',8\}) < \lambda(\{\omega',\omega\})<\lambda(\{\omega,1\})<b$.
\end{enumerate}
\end{lemma}
The four temporal graphs defined in \Cref{gadget extensions no tcc} are sketched in~\Cref{fig gadget extensions}.
\begin{proof}
Let~$i\in [1,4]$.
Recall that a connector path has a girth of~$5$.
So do the temporal graphs~$\G_i$, as the distance between~$1$ and~$8$ and the distance between~$8$ and~$20$ is at least 3 each, so the additionally added vertices and edges do not create a cycle of length less than~$5$.
Assume towards a contradiction that~$\G_i$ contains a nontrivial closed tcc~$S$.
By~\Cref{tccimplication}, the above implies that~$S$ has size at least~$17$.
Moreover, $S$ has to contain at least one vertex that is not in~$[1,20]$, as we did not add any edges between these vertices, and \Cref{gadget no tcc} implies that the connector path contains no nontrivial tcc.
We now distinguish the different cases for~$i$.

\textbf{Case:} $i = 1$\textbf{.}
As~$S$ contains at least one vertex outside of~$[1,20]$, $\alpha$ is in~$S$.
However, $\alpha$ can reach only~$7$ vertices of~$[1,20]$ in~$\G_1$.
This is due to the fact that all edges incident with~$\alpha$ have a label that is larger than~$r$, which is the largest label of any edge of the red path.
Hence, a temporal path starting at~$\alpha$ can only use the incident edges to a vertex~$v\in \{8,20\}$ and then either follow the blue path from vertex~$v$ or use the edge towards the unique predecessor of~$v$ in the blue path and end there.
By definition of the blue path, this implies that~$S \subseteq \{\alpha\} \cup \{15,8,18,4,12,16,20\}$, which contradicts the fact that~$S$ has size at least~$17$.

\textbf{Case:} $i = 2$\textbf{.}
This case follows the same arguments as the previous one:
$S$ has to contain at least one of~$\{\alpha,\alpha'\}$, but both vertices can only reach vertices from~$\{\alpha,\alpha'\} \cup \{15,8,18,4,12,16,20\}$, as all the incident edges of these two vertices have labels larger than~$r$ and end only in vertices~$8$ and~$20$.
This thus contradicts the fact that~$S$ has size at least~$17$.

\textbf{Case:} $i = 3$\textbf{.}
As~$S$ contains at least one vertex outside of~$[1,20]$, $\omega$ is in~$S$.
The idea is mostly similar to the one for~$i=1$, besides that in this case, only few vertices can reach~$\omega$.
That is, $\omega$ can only be reached by~$9$ vertices of~$[1,20]$ in~$\G_3$.
This is due to the fact that all edges incident with~$\omega$ have a label that is smaller than~$b$, which is the smallest label of any edge of the blue path.
Hence, a temporal path ending at~$\omega$ can only use the incident edges from a vertex~$v\in \{1,8\}$ and previously reach the vertex~$v$ via edges of only the red path.
The only vertices that can reach~$1$ or~$8$ via edges of the red path are the vertices of~$[1,9]$.
This implies that~$S \subseteq \{\omega\} \cup [1,9]$, which contradicts the fact that~$S$ has size at least~$17$.

\textbf{Case:} $i = 4$\textbf{.}
This case follows the same arguments as the previous one:
$S$ has to contain at least one of~$\{\omega,\omega'\}$, but both vertices can only be reached by vertices from~$\{\omega,\omega'\} \cup [1,9]$, as all the incident edges of these two vertices have labels smaller than~$b$ and reach only the vertices~$1$ and~$8$.
This thus contradicts the fact that~$S$ has size at least~$17$.

As a consequence, none of the four described temporal graphs contains a nontrivial closed tcc.
\end{proof}

\begin{figure}
\centering
\scalebox{.75}{
\begin{tikzpicture}[square/.style={regular polygon,regular polygon sides=4},yscale = .8]

\begin{scope}
\smallgadget
\node[knoten] (a) at (0,-7) {$\alpha$};
\draw[thick] (20) edge node[black,fill=white] {$30$} (a);
\draw[thick,-, bend left = 40] (8) edge node[black,fill=white] {31} (a);

\end{scope}

\begin{scope}[xshift = 4cm]
\smallgadget
\node[knoten] (a) at (0,-6) {$\alpha$};
\node[knoten] (x1) at (0,-8) {$\alpha'$};
\draw[thick] (20) edge node[black,fill=white] {$30$} (a);
\draw[thick] (x1) edge node[black,fill=white] {$31$} (a);
\draw[thick,-, bend left = 40] (8) edge node[black,fill=white] {32} (x1);

\end{scope}

\begin{scope}[xshift = 8cm]
\smallgadget
\node[knoten] (o) at (0,2) {$\omega$};
\draw[thick] (1) edge node[black,fill=white] {$31$} (o);
\draw[thick,-, bend left = 40] (8) edge node[black,fill=white] {30} (o);
\end{scope}

\begin{scope}[xshift = 12cm]
\smallgadget
\node[knoten] (o) at (0,2) {$\omega$};
\node[knoten] (xk) at (0,4) {$\omega'$};
\draw[thick] (xk) edge node[black,fill=white] {$31$} (o);
\draw[thick] (o) edge node[black,fill=white] {$32$} (1);
\draw[thick,-, bend left = 40] (8) edge node[black,fill=white] {30} (xk);
\end{scope}

\end{tikzpicture}
}
\caption{Examples for the four temporal graphs described in~\ref{gadget extensions no tcc}.
The rectangular areas with the vertices~$1$, $8$, and~$20$ represent connector paths for which the largest red label~$r$ is smaller than~$30$ and the smallest blue label~$b$ is larger than~$32$.}
\label{fig gadget extensions}
\end{figure}

We also observe the following property about nontrivial closed tccs of minimal size.

\begin{lemma}\label{connectivity}
Let~$S$ be an inclusion-minimal nontrivial closed tcc in a simple undirected temporal graph~$\G$ with strict reachability.
Then, the \foot~$G[S]$ is 2-vertex-connected.  
\end{lemma}
\begin{proof}
Clearly, $G[S]$ is connected.
Moreover, since~$S$ is a nontrivial closed tcc, $S$ has size at least~$3$. 
For~$|S|=3$, the statement follows directly, as each temporally connected simple graph on~$3$ vertices is always a triangle and thus 2-vertex-connected.
So we assume for the remainder, that~$S$ has size at least~$4$. 
Assume towards a contradiction that~$\G[S]$ is not 2-vertex-connected.
Then, there is a vertex~$v$, such that the \foot of~$\G[S\setminus \{v\}]$ is disconnected.
Let~$S'$ be a largest connected component in~$G[S\setminus \{v\}]$.
If~$|S'| = 1$, then $G[S]$ is a star.
This implies that~$G[S]$ contains only~$n-1$ edges, which is less than~$2n-3$ by~$|S|\geq 4$.
This contradicts the assumption that~$\G[S]$ is temporally connected (\Cref{tccminedges} also holds for simple undirected temporal graphs with strict reachability~\cite{Bumby81,KMMS24}).
Thus, assume that~$S'$ contains at least two vertices.
We show that for~$S^* := S'\cup \{v\}$, $\G[S^*]$ is temporally connected, which then contradicts the minimality of~$S$.
As~$S'$ is a connected component in~$G[S\setminus \{v\}]$, each path in~$G[S]$ (and thus also each temporal path in~$\G[S]$) between distinct vertices of~$S^*$ only visits vertices of~$S^*$.
By the fact that~$\G[S]$ is temporally connected, $\G[S^*]$ is also temporally connected.
Moreover, we get that~$3 \leq |S^*| < |S|$ by~$|S'|\geq 2$ and the fact that~$G[S\setminus \{v\}]$ contains at least two connected components.
This then implies that~$S^*$ is a nontrivial closed tcc in~$\G$; a contradiction to the minimality of~$S$. 
\end{proof}

\begin{observation}\label{intermediate}
Let~$\G$ be a directed or undirected temporally connected graph with at least two vertices.
Then, for each vertex~$v$ of~$\mg$ for which there is no temporal path in~$\mg$ that contains~$v$ as an intermediate vertex, $\mg$ is still temporally connected after removing vertex~$v$.
\end{observation}
}

In the appendix, we also show that this gadget is minimal in the sense that there is no graph on less than 20 vertices with (i)~girth at least~$5$ and (ii)~two edge-disjoint spanning trees both sharing at least two leaves (vertices~1 and~20 in our gadget).
So the only way to possibly obtain a smaller gadget would be to reduce the girth to~$4$ which we deliberately avoided as discussed above. 

\newcommand{\proofOptimal}{
\begin{lemma}
Let~$G$ be a graph of girth at least~$5$ that contains two edge-disjoint spanning trees that share at least two leaves.
Then, $G$ contains at least~$20$ vertices.
\end{lemma}
\begin{proof}
Let~$n$ be the number of vertices of~$G$.
Let~$T_1$ and~$T_2$ be the two edge-disjoint spanning trees of~$G$ and let~$s$ and~$t$ be distinct common leaves of both trees.
By assumption, two such vertices exist.
Clearly, $n\geq 2$.
Moreover, $n\geq 5$, as a graph with at least two vertices and two edge-disjoint spanning trees contains a cycle, which has length at least~$5$ based on the assumption on the girth.
Now, $G-\{s,t\}$ still contains the two edge-disjoint spanning trees~$T_1 - \{s,t\}$ and~$T_2 - \{s,t\}$ and has a girth of at least~$5$.
This implies that~$G$ contains~$n-2 \geq 3$ vertices and at least~$2(n-2)-2$ edges.
By~\Cref{atleats17}, this implies that~$n-2\geq 18$, which implies that~$G$ contains at least~$n\geq 20$ vertices.
\end{proof}
}

\section{Overview of  the Undirected Reductions (and Implications)}\label{sec sketches}
\begin{figure}
\centering
\scalebox{.55}{
\begin{tikzpicture}[square/.style={regular polygon,regular polygon sides=4},yscale = .8]

    \begin{scope}[xshift=9cm]

\begin{scope}[xshift = 2.5cm]
\smallgadget

\node[knoten] (a) at (0,-6) {$\alpha$};

\node[knoten] (o) at (0,2) {$\omega$};

\draw[thick] (20) edge (a);
\draw[thick] (1) edge  (o);

\node[square,draw] (cc) at (8) {8};

\node (ll) at ($(1) + (-1.1,0)$) {$\mP[{\vcon[u,v]}]$};

\end{scope}

\begin{scope}
\smallgadget

\draw[thick] (20) edge (a);
\draw[thick] (1) edge  (o);
\node[square,draw] (cl) at (8) {8};

\node (ll) at ($(1) + (-1,0)$) {$\mP[{\vout[u]}]$};
\end{scope}

\begin{scope}[xshift = 5cm]
\smallgadget

\draw[thick] (20) edge  (a);
\draw[thick] (1) edge (o);
\node[square,draw] (cr) at (8) {8};

\node (ll) at ($(1) + (-1,0)$) {$\mP[{\vin[v]}]$};
\end{scope}

\draw[thick] (cl) edge (cc);
\draw[thick] (cr) edge  (cc);

\node[knoten] (u) at ($(cl) + (-1.5,0)$) {$u$};
\node[knoten] (v) at ($(cr) + (1.5,0)$) {$v$};

\draw[thick] (u) edge (cl);
\draw[thick] (v) edge (cr);
    \end{scope}

\begin{scope}[xshift = 1.25cm]
\node[knoten] (a) at (0,-6) {$\alpha$};

\node[knoten] (o) at (0,2) {$\omega$};
\end{scope}

\begin{scope}
\smallgadget

\draw[thick] (20) edge (a);
\draw[thick] (1) edge (o);
\node[square,draw] (cl) at (8) {8};

\node (ll) at ($(1) + (-1,0)$) {$\mP[{\vout[u]}]$};
\end{scope}

\begin{scope}[xshift = 2.5cm]
\smallgadget

\draw[thick] (20) edge (a);
\draw[thick] (1) edge (o);
\node[square,draw] (cr) at (8) {8};

\node (ll) at ($(1) + (-1,0)$) {$\mP[{\vin[v]}]$};
\end{scope}

\draw[thick] (cl) edge (cr);

\node[knoten] (u) at ($(cl) + (-1.5,0)$) {$u$};
\node[knoten] (v) at ($(cr) + (1.5,0)$) {$v$};

\draw[thick] (u) edge  (cl);
\draw[thick] (v) edge  (cr);

\end{tikzpicture}
}
\caption{An illustration of parts of the connector gadgets for the undirected reductions and paths between vertices of~$V$ going through the gadget. 
On the left, the parts of the gadget for the strict reduction and on the right, parts of the gadget for the happy reduction. 
The rectangular areas represent the connector paths~$\mP[{\vout[u]}]$, $\mP[{\vcon[u,v]}]$, and~$\mP[{\vin[v]}]$, where the respective vertex~$8$ of each such path denotes the docking point of that connector path, that is, the vertices~$\vout[u]$, $\vcon[u,v]$, and~$\vin[v]$, respectively.
The time labels are not depicted.
}
\label{fig sketch}
\end{figure}

In this section, we present an overview of how the connector paths are implemented as subgraphs of the connector gadget in both undirected reductions (see~\Cref{fig sketch} for a sketch or~\Cref{fig connection strict} on Page~\pageref{fig connection strict} and~\Cref{fig connection happy} on Page~\pageref{fig connection happy} for more detailed illustrations). We focus here on the most important part, some details of the labeling being deferred to the appendix.
Consider the connector gadget for the directed case depicted in~\Cref{fig connectivity gadget happy}.
In a first step, we remove the vertices~$\postin[\alpha]$ and~$\preout[\omega]$, as the purpose of these vertices is to prevent the need for multiple labels on the arcs from~$\omega$ to~$\vout[\omega]$ or from~$\vin[\alpha]$ to~$\alpha$ which we now achieve by using multiple vertices in~$\mP$.
Next, for each vertex~$c$ of type~$\vin[x]$ or~$\vout[y]$, we replace vertex~$c$ by a connector path~$\mP[c]$ and identify vertex~$8$ of that connector path with~$c$, which we call the~\emph{docking point of~$\mP[c]$}. 
Additionally, we replace the arc with label~$2$ from~$c$ to~$\alpha$ by an edge between vertex~$20$ of~$\mP[c]$ and~$\alpha$, and we replace the arc with label~$10$ from~$\omega$ to~$c$ by an edge between vertex~$1$ of~$\mP[c]$ and~$\omega$.
Finally, arcs of the form~$(x,\vout[x])$, $(\vout[x],\vin[y])$, or~$(\vin[y],y)$ are simply made undirected.
The labels of the connector paths are assigned in a way that the labels on the red paths and the blue paths are the first and last $19$~labels of the temporal graph, respectively.

As for the directed case, this construction allows all vertices of the connector gadget (except~$\omega$) to reach $\alpha$ early by using the red paths of the connector paths, and it allows all vertices of the connector gadget (except~$\alpha$) to be reached late by $\omega$ by using the blue paths of the connector paths. 

To prevent the creation of triangles in the \foot (which would immediately imply the existence of a closed tcc of size~$3$), we redefine the edges of the selection gadget. To this end, we make use of several additional properties (that can be assumed without loss of generality) on the structure of the MCC instance that we reduce from.
Finally, to prevent unwanted small closed tccs, the edges of type~$\{\vout[x],\vin[y]\}$ do not exist for all edges~$\{x,y\}$ of~$G$, as some vertex pairs will already have a temporal path that uses the vertices and edges of the selection gadget (at least one-sided).

For the reduction in happy undirected graphs, to ensure that vertices of~$V$ can reach each other via temporal paths only if they are adjacent in the MCC instance, we subdivide each edge~$\{\vout[x],\vin[y]\}$ by a vertex~$d:=\vcon[x,y]$.
As this vertex must be able to reach all other vertices of the desired closed tcc, we add a connector path~$\mP[d]$ for that vertex, for which we again identify the docking point (that is, vertex~$8$) with~$d$, and for which we make (again) vertices~$1$ and~$20$ adjacent to~$\omega$ and~$\alpha$, respectively.

\subparagraph{Sketch for the implications.}
We now briefly describe how to see the implications stated in~\Cref{sec intro}.

\begin{theorem}[Implication~\ref{imp hardness question}]\label{hardness exact tcc}
Deciding whether a happy undirected temporal graph contains a closed tcc of size exactly~$k$ is \textsf{W[1]}-hard for parameter~$k$.
\end{theorem}
We prove this statement by showing that if there is a closed tcc of nontrivial size in~$\mg$, then there is one of size exactly~$10 k^2 + 11 k - 60$ in the undirected happy reduction, where~$k$ is the size of the sought clique in the MCC instance.

Implication~\ref{imp maximality} follows because both constructions contain the edge~$e^* := \{\alpha,\vin[\alpha]\}$ which is part of every nontrivial closed tcc.
Hence, the endpoints of~$e^*$ are \emph{not} an inclusion-maximal closed tcc if and only if there is a nontrivial closed tcc in the respective temporal graph.

Implication~\ref{imp ETH} holds since the number of vertices in the reductions is~$\Oh(|V|)$ and~$\Oh(|V|^2)$ respectively, and MCC cannot be solved in $2^{o(|V|)}$~time, unless the ETH fails~\cite{C+15}.

Similarly to the directed case, Implication~\ref{imp Counting} follows by assuming that the MCC instance contained either no or~$2^{\Theta(|V|)}$ multicolored cliques.
As the undirected reductions construct temporal graphs with~$\Oh(|V|)$ and~$\Oh(|V|^2)$ vertices, respectively, all constructed yes-instances then contain $2^{\Theta(n)}$ and $2^{\Theta(\sqrt{n})}$ (maximal) \TC subgraphs, respectively.
This then also implies that there is some fixed~$c> 0$ such that approximating the number of (maximal) \TC subgraphs cannot be done within a factor of $2^{c\cdot n}$ and $2^{c\cdot \sqrt{n}}$, respectively, as a no-instance only contains~$\Oh(|V|^2)$ \TC subgraphs.

Finally, for Implication~\ref{imp approx}, we essentially proceed like in the directed case: we extend the connector gadget by additional connector paths~$\mP$ for which we do not add incident edges to the docking point (vertex~$8$), until the size of the connector gadget is at least~$(1-\epsilon)|V^*|$, for~$|V^*|$ being the number of vertices of the current temporal graph.
All the vertices of the connector gadget can always be added to a nontrivial closed tcc and the newly added connector paths do not create new nontrivial closed tccs.
Hence, similar to the directed case, it is \NP-hard to decide whether the largest \TC subgraph has trivial size or whether the largest \TC subgraph has size~$(1-\epsilon)n$.
However, since in the undirected case tccs of size~$2$ are trivial, we need to divide the factor by~$2$.
This yields the stated hardness to approximate the size of largest \TC subgraph within a factor of~$(1-\epsilon)\frac{n}{2}$.

\section{Conclusion}
In this paper, we significantly refined the lower bounds for finding \TC subgraphs (closed temporally connected components) in a given temporal graph.
We showed that even when asking for components of size at least 2 in the directed case and size at least 3 in the undirected case, the task can presumably not be performed in polynomial time.
Restricting the temporal graphs to have a constant lifetime does not change the complexity here, except in the case of undirected temporal graphs with non-strict reachability.
Our reductions imply a number of significant side results; in particular, they answer the open question~\cite{CLMS24} of whether finding a \TC subgraph of size \emph{exactly}~$k$ in the non-strict model on undirected temporal graphs can be done in FPT time for~$k$.

Following this work, there are a number of natural directions for future work:
\begin{itemize}
\item We show that determining the existence of a nontrivial \TC subgraph in the non-strict model on undirected temporal graphs can not be solved in $2^{o(\sqrt{n})}$~time. 
The trivial upper bound is at $2^{\Oh(n)}$~time.
Is it possible to improve the upper-bound to $2^{\Oh(\sqrt{n})}$~time?
\item Our work essentially only considers lower bounds. It would be interesting to see for which parameters the problem can be solved in FPT time. For example, is it possible to solve the problem in FPT time for the maximum degree or the treewidth of the \foot? 
In recent work~\cite{DDE+25}, it was shown that finding a large \TC subgraph is \NP-hard even on graphs of bounded treewidth, but the presented reduction is guaranteed to contain small nontrivial \TC subgraphs.
\item It would also be interesting to consider further models of components, for example the closed version of the round-trip components in~\cite{BSS24}, or closed components that are not completely \TC, but still highly connected, for example where half of the vertices can reach each other by temporal paths. 
These properties seem to share with \TC subgraphs the feature that they are neither hereditary nor quasi-hereditary, making their study an interesting target for future work.
  
\end{itemize}

\bibliographystyle{plainurl}
\bibliography{bib}

\longtrue
\newpage
\appendix

\section{Missing Parts of Section~\ref{sec gagdet}}
\subparagraph{Proof of~\Cref{sec tc girth}.}
\textGirth

\subparagraph{Some further observations and extensions of connector paths.}
\proofExtensionGadgets

\subparagraph{Proof of the optimality of the gadget for girth 5.}
\proofOptimal

\section{The Hardness on Undirected Happy Temporal Graphs}\label{sec hard happy}
With the previous observations and the definition of the connector paths, we are now ready to prove the second main result.

\begin{theorem}\label{hardness undir happy}
\NCTCC is \NP-hard on happy undirected temporal graphs.
\end{theorem}

We again reduce from MCC.
Let~$I=(G=(V,E),k)$ be an instance of~MCC and let~$V_1 \cup \dots \cup V_k$ be a~$k$-partition of~$V$.
We can assume that~$k \mod 3 = 0$, $V_i := \{v_i^1\}$ for each~$i\in \{1,k\} \cup \{2 + 3r \mid r \in [0,\frac{k}{3}-1]\}$, $V_i := \{v_i^1, \dots, v_i^q\}$ with~$q> 1$ for each other~$i\in [1,k]$.
Moreover, we can assume that~$G[V_i \cup V_{i+1}]$ is a star or a matching for each~$i\in [1,k-1]$.\footnote{To achieve this property, we can simply duplicate each color class (while keeping the adjacencies towards vertices outside of that class) of a normal instance of MCC, add a matching between the two copies of each vertex, and add sufficiently many color classes to the instance that all contain only a single vertex, which is adjacent to each other vertex of the graph.
This only increases the size of the sought clique and the number of vertices by a constant factor.}

\begin{figure}
\centering
\begin{tikzpicture}[xscale=1.2,yscale=.8]

\tikzstyle{k}=[circle,fill=white,draw=black,minimum size=8pt,inner sep=2pt]

\foreach \y in {3,4,6,7,9,10}{

\draw[rectangle] (.3 + \y,-2) -- (-.3 + \y,-2) -- (-.3 + \y,2) -- (.3 + \y,2) -- (.3 + \y,-2);
\node (v1) at (0 + \y,-2.4) {$V_{\y}$};

\foreach \x in {1,2,3,4,5} 
\node[k] (\y\x) at ($(0 + \y,-2.25) +  (0,\x*.75)$) {};
}
\foreach \y/\lab in {1/\alpha,2/\alpha',5/ ,8/ ,11/\omega',12/\omega}{

\draw[rectangle] (.3 + \y,-2) -- (-.3 + \y,-2) -- (-.3 + \y,2) -- (.3 + \y,2) -- (.3 + \y,-2);
\node (v1) at (0 + \y,-2.4) {$V_{\y}$};

\node[k] (\y1) at ($(0 + \y,0) +  (0,0)$) {$\lab$};
}

\draw[thick] (11) to (21);
\draw[thick] (111) to (121);

\foreach \y/\yy in {2/3, 5/4, 5/6, 8/7, 8/9, 11/10}{

\foreach \x in {1,2,3,4,5} 
\draw[thick] (\y1) to (\yy\x);

}

\foreach \y/\yy in {3/4, 6/7, 9/10}{

\foreach \x in {1,2,3,4,5} 
\draw[thick] (\y\x) to (\yy\x);

}

\end{tikzpicture}
\caption{An example for the \foot~$G'[V]$ in the reduction of~\Cref{hardness undir happy} for an input instance of MCC, where~$k=12$ and all color classes of size more than one contain exactly 5 vertices.}
\label{fig selection edges}
\end{figure}

\begin{figure}
\centering
\scalebox{.75}{
\begin{tikzpicture}[square/.style={regular polygon,regular polygon sides=4},yscale = .8]

\begin{scope}[xshift = 5cm]
\smallgadget

\node[knoten] (a) at (0,-8) {$\alpha$};

\node[knoten] (o) at (0,4) {$\omega$};

\draw[thick] (20) edge node[black,fill=white] {$\in$ $\alpha$-times} (a);
\draw[thick] (1) edge node[black,fill=white] {$\in$ $\omega$-times} (o);

\node[square,draw] (cc) at (8) {8};

\node (ll) at ($(1) + (-1.1,0)$) {$\mP[{\vcon[u,v]}]$};

\end{scope}

\begin{scope}
\smallgadget

\draw[thick] (20) edge node[black,fill=white] {$\in$ $\alpha$-times} (a);
\draw[thick] (1) edge node[black,fill=white] {$\in$ $\omega$-times} (o);
\node[square,draw] (cl) at (8) {8};

\node (ll) at ($(1) + (-1,0)$) {$\mP[{\vout[u]}]$};
\end{scope}

\begin{scope}[xshift = 10cm]
\smallgadget

\draw[thick] (20) edge node[black,fill=white] {$\in$ $\alpha$-times} (a);
\draw[thick] (1) edge node[black,fill=white] {$\in$ $\omega$-times} (o);
\node[square,draw] (cr) at (8) {8};

\node (ll) at ($(1) + (-1,0)$) {$\mP[{\vin[v]}]$};
\end{scope}

\draw[thick] (cl) edge node[black,fill=white] {$\in$ out2edge-times} (cc);
\draw[thick] (cr) edge node[black,fill=white] {$\in$ edge2in-times} (cc);

\node[knoten] (u) at ($(cl) + (-3,0)$) {$u$};
\node[knoten] (v) at ($(cr) + (3,0)$) {$v$};

\draw[thick] (u) edge node[black,fill=white] {out-time} (cl);
\draw[thick] (v) edge node[black,fill=white] {in-time} (cr);

\end{tikzpicture}
}
\caption{An illustration of the paths between vertices of~MCC instance in the reduction behind~\Cref{hardness undir happy}.
For vertices~$v\in V_i$ and~$u\in V_j$ with~$\{u,v\}\in E$ and~$j > i +1$, we have a temporal path from~$u$ to~$v$ in our temporal graph, namely, the path $(u,\vout[u],\vcon[u,v],\vin[v],v)$.
Similar to~\Cref{fig gadget extensions}, the rectangular areas indicate the connector paths~$\mP[{\vout[u]}]$, $\mP[{\vcon[u,v]}]$, and~$\mP[{\vin[v]}]$, where the respective vertex~$8$ of each such path denotes the docking point of that connector path, that is, the vertices~$\vout[u]$, $\vcon[u,v]$, and~$\vin[v]$, respectively.}
\label{fig connection happy}
\end{figure}

\subsection{Construction}
We initialize the \foot~$G' := (V',E')$ of~$\G$ as the graph~$(V,E' := \bigcup_{i=1}^{k-1} E_G(V_i,V_{i+1}))$.
See~\Cref{fig selection edges} for an example.
First we argue that~$G'$ has a girth of at least~$6$ so far. 
As~$G[V_i \cup V_{i+1}]$ is a star or a matching for each~$i\in [1,k-1]$, every cycle in~$G'$ contains at least two vertices for which the color class has size~$1$.
Since there is no~$j\in [1,k-2]$ with~$|V_j| = |V_{j+2}| = 1$, this implies that each cycle in~$G'$ has length at least~$6$.

Recall that~$|V_1| = |V_2| = |V_{k-1}| = |V_{k}| = 1$. 
We use the shorthands~$\alpha := v_1^1$, $\alpha' := v_2^1$, $\omega:= v_k^1$, and~$\omega' := v_{k-1}^1$.
Observe, that in contrast to the reduction for the directed temporal graphs, $\alpha$ and~$\omega$ are vertices of~$V$.
Note that we can assume that each of the vertices of~$\{\alpha,\alpha',\omega,\omega'\}$ is adjacent to every other vertex of~$V$, as each of these four vertices is part of every clique of size~$k$ in~$G$, and removing a nonneighbor of any of these four vertices would result in an equivalent instance of~MCC.

We now define several in-vertices and out-vertices similar to the reduction behind~\Cref{hardness happy directed}, as well as some additional connector-vertices that essentially act as subdivisions of the arcs from out- to in-vertices from the reduction behind~\Cref{hardness happy directed}.
We initialize three sets~$\Vin$, $\Vout$, and~$\Vconn$ as empty sets.
For each~$v\in V \setminus \{\omega,\omega'\}$, we add a vertex~$\vin[v]$ to~$\Vin$, and for each~$v\in V \setminus \{\alpha,\alpha'\}$, we add a vertex~$\vout[v]$ to~$\Vout$.
Now, we define the connector-vertices.
For each edge~$\{u,v\}\in E$ with~$v\in V_i$, $u\in V_j$, and~$j \geq i+2$, we add a vertex~$\vcon[u,v]$ to~$\Vconn$.  

For each~$c\in \Vin \cup \Vout \cup \Vconn$, we now add a copy of a connector path~$\mP[c]$ to~$G'$ (see~\Cref{def con path}), where we identify the vertex~$8$ of~$\mP[c]$ with~$c$, to which we may refer to as the~\emph{docking point} of~$\mP[c]$.
This makes~$\Vin \cup \Vout \cup \Vconn$ a subset of the vertices of~$G'$.
We call vertex~$1$ and~$20$ of~$\mP[c]$ the~\emph{source} and~\emph{sink} of~$\mP[c]$ respectively.

We complete the construction of the \foot by adding edges as follows:
\begin{enumerate}
\item For each connector path~$\mP[c]$, we add the edge~$\{t^c,\alpha\}$ to~$G'$, where~$t^c$ is the sink of~$\mP[c]$.
\item For each vertex~$v\in V \setminus \{\alpha,\alpha'\}$, we add an edge~$\{v,\vout[v]\}$ to~$G'$ (called the~\emph{out-edge of~$v$}).
\item For each~$\vcon[u,v]\in \Vconn$, we add the edge~$\{\vout[u],\vcon[u,v]\}$ to~$G'$.
\item For each~$\vcon[u,v]\in \Vconn$, we add the edge~$\{\vcon[u,v],\vin[v]\}$ to~$G'$.
\item For each vertex~$v\in V \setminus \{\omega,\omega'\}$, we add an edge~$\{\vin[v],v\}$ to~$G'$ (called the~\emph{in-edge of~$v$}).
\item For each connector path~$\mP[c]$, we add the edge $\{\omega,s^c\}$ to~$G'$, where~$s^c$ is the source of~$\mP[c]$.
\end{enumerate}

This completes the definition of the \foot~$G'$.
Some of the edges incident with vertices of~$\Vin\cup \Vout\cup \Vconn$ are depicted in~\Cref{fig connection happy}.

\begin{table}
\caption{The different intervals of time labels and the edges that receive labels from these intervals.
The intervals 'out-time' and 'in-time' have length one.
If an interval is below another interval in this table, then each label of that interval is strictly larger than the largest label of the interval above.
For example, all labels of the selection-times are strictly smaller than all labels of the~$\omega$-times.}
\centering
\begin{tabular}{l|l}
name &  receiving edges \\\hline
red-times & edges of the red paths of each connector path\\
$\alpha$-times & edges between~$\alpha$ and the sink vertex of each connector path\\
out-time & edges~$\{\{v,\vout[v]\}\mid v\in V \setminus \{\alpha,\alpha'\}\}$\\
out2edge-times &  edges~$\{\{\vout[u],\vcon[u,v]\}\mid \vcon[u,v]\in \Vconn\}$\\
edge2in-times  &  edges~$\{\{\vcon[u,v], \vin[v]\}\mid \vcon[u,v]\in \Vconn\}$\\
selection-times & edges between the vertices of~$V$ in~$\G$: $\bigcup_{i=1}^{k-1}E_G(V_i,V_{i+1})$\\
in-time &  edges~$\{\{v,\vin[v]\}\mid v\in V \setminus \{\omega,\omega'\}\}$\\
$\omega$-times & edges between~$\omega$ and the source vertex of each connector path\\
blue-times  & edges of the blue paths of each connector path
\end{tabular}
\label{tab times}
\end{table}

\begin{observation}
$G'$ has a girth of~$5$.
\end{observation}
We complete the construction by defining the labels of the edges.
To this end, we define several intervals that are pairwise disjoint and assign labels of these intervals to specific edge sets. 
See~\Cref{tab times} for a general overview.

\begin{itemize}
\item The first 19 time steps are the~\emph{red-times}. 
These are assigned to the red paths of all connector paths.
\item Afterwards, there are the~\emph{$\alpha$-times}.
These are assigned to the edges $\{\{t^c,\alpha\}\mid c\in \Vin\cup\Vout\cup\Vconn\}$.
Here, recall that~$t^c$ denotes the sink of the connector path~$\mP[c]$.
The labeling is done in an arbitrary proper way to these edges.
\item Afterwards, there is the~\emph{out-time}.
This time interval has length only~$1$ and thus consists of a single label. 
This label is assigned to the edges $\{\{v,\vout[v]\}\mid v\in V \setminus \{\alpha,\alpha'\}\}$.
\item Afterwards, there are the~\emph{out2edge-times}.
These are assigned to the edges between the vertices of~$\Vout$ and~$\Vconn$ in an arbitrary proper way.
\item Afterwards, there are the~\emph{edge2in-times}.
These are assigned to the edges between the vertices of~$\Vconn$ and~$\Vin$ in an arbitrary proper way.
\item Afterwards, there are the~\emph{selection-times}.
These are assigned to the edges between the vertices of~$V$ in an arbitrary proper way, such that all labels assigned to edges of~$E(V_i,V_{i+1})$ are strictly smaller than all labels assigned to edges of~$E(V_j,V_{j+1})$ for~$1\leq i < j < k$. 
\item Afterwards, there is the~\emph{in-time}.
This time interval has length only~$1$ and thus consists of a single label. 
This label is assigned to the edges $\{\{\vin[v],v\}\mid v\in V \setminus \{\omega,\omega'\}\}$.
\item Afterwards, there are the~\emph{$\omega$-times}.
These are assigned to the edges $\{\{\omega,s^c\}\mid c\in \Vin\cup\Vout\cup\Vconn\}$.
Here, recall that~$s^c$ denotes the source of the connector path~$\mP[c]$.
The labeling is done in an arbitrary proper way to these edges.
\item The last 19 time steps are the~\emph{blue-times}. 
These are assigned to the blue paths of all connector paths.
\end{itemize}

This completes the construction. 
Let~$\G$ be the resulting temporal graph.

We show that there is a nontrivial closed tcc in~$\G$ if and only if~$G$ has a clique of size~$k$.
In the proof, we will highly rely on the relative order of the time intervals (see~\Cref{tab times}).

First, we observe the following property of temporal paths in~$\G$ that is relevant for statements in both directions of the correctness.
\begin{lemma}\label{claim dont go in}
Let~$v\in V\setminus \{\omega,\omega'\}$ and let~$c= \vin[v]$.
Each temporal path that traverses the edge~$\{v,c\}$ from~$v$ to~$c$ in~$\G$ ends in a vertex of the connector paths~$\mP[c]$ and each temporal path that traverses the edge~$\{v,c\}$ from~$c$ to~$v$ in~$\G$ ends in~$v$.
\end{lemma}
\begin{proof}
The first point follows by the fact that~$\{v,c\}$ is the last edge incident with~$c$ that has a label prior to the blue-times.
Thus, after entering the connector path through this edge, the only way to continue is by traversing with edges of the blue path, which is entirely in~$\mP[c]$.

The second point follows by the fact that~$\{v,c\}$ is the unique edge of largest label incident with~$v$, so the temporal path ends by necessity at~$v$.
\end{proof}

\subsection{A size-$k$ clique in~$G$ implies a nontrivial closed tcc in~$\G$}
In this subsection, we show one direction of the correctness of the reduction.

\begin{lemma}\label{undirected direction forward}
If there is a clique of size~$k$ in~$G$, then there is a nontrivial closed tcc of size~$k + 20\cdot \left(2\cdot (k-2) +  \binom{k}{2} - (k-1) \right)\in \Oh(k^2)$ in~$\G$.
\end{lemma}
For the remainder of the subsection assume that~$G$ contains a clique~$X$ of size~$k$.
By construction, for each~$i\in [1,k]$, there is a unique vertex~$x_i$ in~$X\cap V_i$.
Note that~$x_1 = \alpha$, $x_2 = \alpha'$, $x_{k-1} = \omega'$, and~$x_k = \omega$.

Recall that for each~$x_i\in X \setminus \{\omega,\omega'\}$, there is a vertex~$\vin[x_i]$ in~$\Vin$.
Similarly, for each~$x_i\in X \setminus \{\alpha,\alpha'\}$, there is a vertex~$\vout[x_i]$ in~$\Vout$.
Furthermore, for each~$1 \leq i < i+1 < j \leq k$, there is a vertex~$\vcon[x_i,x_j]$ in~$\Vconn$.
The latter is the case since~$\{x_i,x_j\}$ is an edge of~$G$, as~$X$ is a clique in~$G$.

Let~$C_X$ denote these vertices from~$\Vin\cup\Vout\cup\Vconn$.
That is, $$C_X:= \{\vin[x_i]\mid 1 \leq i \leq k-2\}\cup \{\vout[x_i]\mid 3\leq i \leq k\} \cup \{\vcon[x_j,x_i] \mid 1 \leq i < i+1 < j \leq k\}.$$

We show that~$S:= X \cup \bigcup_{c\in C_X} V(\mP[c])$ is a closed tcc.
Here, recall that~$\mP[c]$ denotes the connector path that contains~$c$ and that we denote the source of~$\mP[c]$ by~$s^c$ and the sink of~$\mP[c]$ by~$t^c$.

First, we show that each vertex of~$S\setminus X$ can reach each vertex of~$S$ in~$\G[S]$ and vice versa.
\begin{lemma}\label{reachability outside V is trivial}
In~$\G[S]$, each vertex of~$S\setminus X$ can reach each vertex of~$S$, and each vertex of~$S$ can reach each vertex of~$S\setminus X$. 
\end{lemma}
\begin{proof}
For each~$c\in C_X$, each vertex of~$V(\mP[c])$ can reach the sink~$t^c$ by following the red path in~$\mP[c]$, thus arriving at~$t^c$ prior to the~$\alpha$-times.
Afterwards, the edge~$\{t^c,\alpha\}$ can be traversed during the~$\alpha$-times, thus reaching~$\alpha$ prior to the selection-times.
Afterwards, there is the temporal path~$P_X = (\alpha = x_1, \dots, x_k = \omega)$ during the selection-times.
This path exists due to the fact that~$X$ is a clique in~$G$ and~$G'[V]$ contains all edges between consecutive color classes.
Moreover, the path is a temporal path, that is, its labels are strictly increasing, since by definition of the labels, the largest label assigned to any edge of~$E(V_i,V_{i+1})$ is strictly smaller than the smallest label assigned to~$E(V_j,V_{j+1})$ for~$1\leq i < j < k$.
Thus, this path reaches~$x_k = \omega$ (and all other vertices of~$X$) prior to the~$\omega$-times.
So far, this implies that each vertex of~$S\setminus X$ can reach each vertex of~$X$ prior to the~$\omega$-times and each vertex of~$X$ can reach~$\omega$ prior to the~$\omega$-times.
To show that each vertex of~$S\setminus X$ can reach each vertex of~$S$ in~$\G[S]$ and vice versa, it suffices to show that~$\omega$ can reach every vertex of~$S \setminus X$ when starting with an edge that receives a label from the~$\omega$-times.
As for each~$c\in C_X$, the edge~$\{\omega, s^c\}$ can be traversed during the~$\omega$-times and afterwards following the blue path of~$\mP[c]$ at the blue-times, $\omega$ can reach each vertex of~$\mP[c]$ when starting with an edge that receives a label from the~$\omega$-times.
As~$S\setminus X = \bigcup_{c\in C_X} V(\mP[c])$, this then implies that each vertex of~$S\setminus X$ can reach each vertex of~$S$ in~$\G[S]$ and vice versa.
\end{proof}

It thus remains to show that the vertices of~$X$ can reach each other in~$\G[S]$.
\begin{lemma}
In~$\G[S]$, each vertex of~$X$ can reach each vertex of~$X$. 
\end{lemma}
\begin{proof}
Let~$1\leq i < j \leq k$.
We show that~$x_i$ and~$x_j$ can reach each other in~$\G[S]$.
If~$j = i+1$, then the two vertices are adjacent in~$\G$ and thus also in~$\G[S]$.
This implies that both can reach each other.
Otherwise, consider~$j \geq i+2$.
As already discussed in the proof of~\Cref{reachability outside V is trivial}, $P_X = (\alpha = x_1, \dots, x_k = \omega)$ is a temporal path in~$\G[S]$.
This immediately implies that~$x_i$ can reach~$x_j$ in~$\G[S]$ by simply following the respective subpath of~$P_X$.
It thus remains to show that~$x_j$ can reach~$x_i$ in~$\G[S]$.
As~$j \geq i+2$, we get that~$x_i \notin \{\omega,\omega'\}$ and~$x_j \notin \{\alpha,\alpha'\}$.
This implies that the vertices~$\vin[x_i]$ and~$\vout[x_j]$ exist and are part of~$C_X\subseteq S$.
Moreover, since~$\{x_i,x_j\}$ is an edge of~$G$, $\vcon[x_j,x_i]$ exists and is part of~$C_X \subseteq S$.
By definition of the edges in~$\G$, $(x_j,\vout[x_j],\vcon[x_j,x_i],\vin[x_i],x_i)$ is a path in~$\G$ and thus also in~$\G[S]$.
Moreover, this path is a temporal path, as the first edge of the path receives a label from the out-times, the second edge receives a label from the out2edge-times, the third edge receives a label from the edge2in-times, and the last edge receives a label from the in-times.
This implies that~$x_j$ can also reach~$x_i$ in~$\G[S]$.
\end{proof}

As a consequence, each vertex of~$S$ can reach every other vertex of~$S$ in~$\G[S]$.
That is, $\G[S]$ is temporally connected and thus~$S$ is a closed tcc in~$\G$.
As~$S$ is a closed tcc in~$\G$ of size~$k + 20\cdot \left(2\cdot (k-2) +  \binom{k}{2} - (k-1) \right)$ and~$k$ is larger than~$3$, this implies that~$\G$ contains a nontrivial closed tcc, which proves~\Cref{undirected direction forward}.
 
With the same arguments above, one can also show that~$S \cup (V'\setminus V)$ is a closed tcc, that is, we can add the vertices of the whole connector gadget (the vertices of all connector paths) to each nontrivial closed tcc of~$\mg$ without violating the property of being a closed tcc.

\subsection{A nontrivial closed tcc in~$\G$ implies a size-$k$ clique in~$G$}
In this subsection, we show the other direction of the correctness of the reduction.
\begin{lemma}\label{undirected direction backward}
If there is a nontrivial closed tcc in~$\G$, then~$G$ contains a clique of size~$k$.
\end{lemma}

For the remainder of this subsection, assume that~$\G$ contains a nontrivial closed tcc.
Let~$S$ be a nontrivial closed tcc in~$\G$ of minimal size.
We will carefully analyze the structure of~$S$ and eventually show that~$S\cap V$ is a clique of size~$k$ in~$G$.
As the \foot~$G'$ of~$\G$ has a girth of~$5$, we immediately get that~$S$ has size at least~$17$ due to~\Cref{tccimplication}.

\begin{lemma}\label{undir not only from V}
$S$ contains at least one vertex of~$V'\setminus V$.
\end{lemma}
\begin{proof}
We show that~$\G[V]$ contains no nontrivial closed tcc.
This then implies that~$S$ contains at least one vertex of~$V'\setminus V$.
Assume towards a contradiction that~$S$ contains only vertices of~$V$.
As~$S$ is a nontrivial closed tcc, \Cref{tccminedges} implies that~$\G[S]$ contains at least one cycle.
Based on the fact that the edges in~$\G$ between consecutive color classes are either a star or a matching, this implies that~$S$ contains vertices of two nonconsecutive color classes. 

By definition of the labels, all edges in~$\G[V]$ are from the selection-times and defined in a way, such that for each~$1\leq i < j \leq k$, all labels assigned to the edges of~$E(V_i,V_{i+1})$ are strictly smaller than all labels assigned to the edges of~$E(V_j,V_{j+1})$.
This implies that there is no temporal path from any vertex of~$V_p$ to~$V_q$ in~$\G[V]$ if~$p > q+1$.
However, as~$S$ contains vertices of two nonconsecutive color classes, $S$ contains a vertex~$v_p$ and a vertex~$v_q$, such that there is no temporal path from~$v_p$ to~$v_q$ in~$\G[V]$ and thus also not in~$\G[S]$.
This contradicts the assumption that~$\G[S]$ is temporally connected.
\end{proof}

\begin{lemma}
\label{undir both al and om}
$S$ contains~$\alpha$ and~$\omega$.
\end{lemma}
\begin{proof}
Assume towards a contradiction that~$S$ does not contain both~$\alpha$ and~$\omega$.
Due to~\Cref{undir not only from V}, $S$ contains at least one vertex outside of~$V$.
Since~$V'\setminus V$ only consists of the union of the vertex sets of the connector paths, this implies that there is some~$c\in \Vin\cup \Vout\cup \Vconn$, such that~$S$ contains at least one vertex of~$\mP[c]$. We first observe the following, where for some connector path~$\mP[c]$, we call a vertex~$q\in V' \setminus V(\mP[c])$ a~\emph{neighbor} of~$\mP[c]$, if~$q$ is adjacent to at least one vertex of~$\mP[c]$.

\begin{claim}\label{if s contains neither al nor om}
Assume that~$S$ contains a vertex of some connector path~$\mP[c']$.
Then~$S$ contains at least two neighbors of~$\mP[c']$.
Moreover, if~$S$ does not contain both~$\alpha$ and~$\omega$, then~$S$ contains~$c'$.
\end{claim}
\begin{claimproof}
Since~$S$ is a nontrivial closed tcc of minimal size, $S$ is 2-vertex-connected (see \Cref{connectivity}). 
Let~$Q = V(\mP[c'])$.
We first show that~$S$ contains at least two neighbors of~$\mP[c']$.
Assume towards a contradiction that~$S$ does not contain two neighbors of~$\mP[c']$.
First of all, $S$ contains at least one neighbor of~$\mP[c']$, as~$\mP[c']$ is a connector path and does not contain a nontrivial closed tcc by~\Cref{gadget no tcc}.
That is, $S$ contains at least one vertex of~$V'\setminus Q$, which implies that there needs to be a neighbor~$r$ of~$\mP[c']$ in~$S$.
Since we assume that there is no other neighbor of~$\mP[c']$ in~$S$, there are two cases: there is a vertex~$r'\in S\setminus (Q\cup \{r\})$, or~$S\subseteq Q\cup \{r\}$.
In the first case, as~$r'$ is not a neighbor of~$\mP[c']$, this contradicts the assumption that~$S$ is 2-vertex connected, since each path between~$r'$ and any vertex of~$Q$ in~$\G[S]$ goes over~$r$.
In the second case, $r$ needs to have at least two neighbors in~$S\cap Q$, as~$S$ is 2-vertex-connected.
By construction of the connector paths and their incident edges, this implies that (i)~$r=\alpha$ and~$c'= \vin[\alpha]$ or (ii)~$r=\omega$ and~$c'= \vout[\omega]$.
That is, $S \subseteq S_1 := \{\alpha\} \cup V(\mP[{\vin[\alpha]}])$ or $S \subseteq S_3 := \{\omega\} \cup V(\mP[{\vout[\omega]}])$.
Note that~$\G[S_1]$ and $\G[S_3]$ fulfill the requirement of the temporal graph~$\G_1$ and~$\G_3$ from Item~\ref{gadget plus alpha} and Item~\ref{gadget plus omega} of~\Cref{gadget extensions no tcc}, respectively.
Thus, neither of these temporal subgraphs contains any nontrivial closed tcc due to~\Cref{gadget extensions no tcc}.
Consequently, $S$ contains at least two neighbors of~$\mP[c']$.

It remains to show that~$S$ contains~$c'$ if~$S$ does not contain both~$\alpha$ and~$\omega$.
To this end, assume towards a contradiction that~$S \cap \{\alpha,\omega,c'\} = \{\alpha\}$ or~$S \cap \{\alpha,\omega,c'\} = \{\omega\}$.
In the first case, this implies that the sink~$t^{c'}$ is the only vertex of~$Q$ that is adjacent to a neighbor (namely~$\alpha$) of~$\mP[c']$ in~$S$.
That is, $\{\alpha,t^{c'}\}$ is the only edge between vertices of~$Q$ and vertices outside of~$Q$ in~$\G[S]$.
This however contradicts the assumption that~$S$ is 2-vertex-connected.
The second case leads to a contradiction in a similar way:
the source~$s^{c'}$ is the only vertex of~$Q$ that is adjacent to a neighbor (namely~$\omega$) of~$\mP[c']$ in~$S$.
That is, $\{\omega,s^{c'}\}$ is the only edge between vertices of~$Q$ and vertices outside of~$Q$ in~$\G[S]$.
This contradicts the assumption that~$S$ is 2-vertex-connected.
Consequently, $S$ contains~$c'$ if~$S$ contains not both~$\alpha$ and~$\omega$.
\end{claimproof}

As we assume (towards a contradiction) that~$S$ does not contain both~$\alpha$ and~$\omega$, this claim implies that $S$ contains some vertex~$c\in \Vin\cup \Vout\cup \Vconn$.
In the following, let~$c$ be a vertex of~$\Vconn\cap S$, if one exists.
Otherwise, let~$c$ be an arbitrary vertex of~$\Vin\cup \Vout \cap S$.
By the above, $S$ also contains at least two neighbors of~$\mP[c]$.
We now distinguish 6 different cases which two neighbors of~$\mP[c]$ are contained in~$S$.
We show that all of them lead to a contradiction, since we assume that~$S$ does not contain both~$\alpha$ and~$\omega$.

\begin{itemize}
\item \textbf{Case 1:} $S$ contains none of the vertices of~$\Vconn \cup \{\alpha,\omega\}$\textbf{.} 
As~$S$ contains~$c$, this implies that~$c\in \Vin\cup \Vout$.
Moreover, \Cref{if s contains neither al nor om} implies that~$S$ contains at least two neighbors of~$\mP[c]$.
As~$c\in \Vin\cup \Vout$, there is at most one vertex of~$V\setminus \{\alpha,\omega\}$ which is a neighbor of~$\mP[c]$.
All other neighbors of~$\mP[c]$ are from~$\Vconn \cup \{\alpha,\omega\}$.
This contradicts the assumption that~$S$ contains no vertex of~$\Vconn \cup \{\alpha,\omega\}$.
\item \textbf{Case 2:} $c\in \Vconn$ and~$S$ contains neither~$\alpha$ nor~$\omega$\textbf{.} 
By construction, $\mP[c]$ has only four neighbors, namely~$\alpha$, $\omega$, a vertex~$c_1\in \Vin$, and a vertex~$c_2\in \Vout$.
As~$S$ contains neither~$\alpha$ nor~$\omega$, this implies that~$c_1$ and~$c_2$ are in~$S$.
Moreover, for each~$i\in \{1,2\}$, $c_i$ is the only vertex of~$\mP[c_i]$ in~$S$, as~$S$ is nontrivial closed tcc of minimal size, which is~$2$-vertex connected by~\Cref{connectivity}.
This is due to the fact that since~$S$ contains neither~$\alpha$ nor~$\omega$, $G'[S\setminus \{c_i\}]$ would have at least two connected components if~$S$ contains at least a second vertex of~$\mP[c_i]$.
Hence, in~$\G[S]$, $c_i$ is incident with edges towards vertices of~$\Vconn$ and at most one vertex of~$V$.
Thus, in~$\G[S]$, all labels incident with~$c_1$ are from the edge2in-times or the in-time, and all labels incident with~$c_2$ are from the out2edge-times or the out-time.
That is, each label incident with~$c_2$ is strictly smaller than each label incident with~$c_1$ in~$\G[S]$.
Thus, there is no temporal path from~$c_1$ to~$c_2$ in~$\G[S]$, contradicting that~$\G[S]$ is temporally connected. 
\item \textbf{Case 3:} $c\in \Vconn$ and~$S$ contains~$\alpha$ but not~$\omega$\textbf{.}
Consider any temporal path~$P$ from~$\alpha$ to~$c$ in~$\G[S]$.
There are only three types of edges incident with~$\alpha$: (1)~the edges towards the sink vertices of the connector paths, (2)~the edge towards~$\vin[\alpha]$, and (3)~the edge towards~$\alpha'$.
The path~$P$ cannot use an edge towards the sink~$t^{c'}$ of a connector path~$\mP[c']$ (which is labeled with an~$\alpha$-time), as~$t^{c'}$ has only one incident edge with a later label: the edge towards the vertex~$16$ of~$\mP[c']$, which exists only at the very last time step of~$\G$.
Moreover, $P$ cannot use the edge towards~$\vin[\alpha]$, as~\Cref{claim dont go in} implies that then~$P$ can only end in~$\mP[{\vin[\alpha]}]$ which is distinct from~$\mP[c]$ and thus does not contain~$c$.
In fact, for the same reason, $P$ cannot traverse any in-edge, as~$c\notin \Vin$.
Thus, assume that~$P$ traverses the edge towards~$\alpha'$.
This happens at a selection-time, which is later than the out-time.
Consequently, $P$ traverses none of the in-edges or the out-edges.
As~$S$ does not contain~$\omega$, $P$ does not visit~$\omega$ and thus can only traverse edges between vertices of~$V$, which contradicts the assumption that~$P$ reaches~$c$. 
\item \textbf{Case 4:} $c\in \Vconn$ and~$S$ contains~$\omega$ but not~$\alpha$\textbf{.}
Consider any temporal path~$P$ from~$c$ to~$\omega$ in~$\G[S]$.
There are only two types of edges incident with~$c$: (1)~the edges towards other vertices of~$\mP[c]$ and (2)~the edges towards the unique neighbors in~$\Vin$ and~$\Vout$.
The path~$P$ cannot use an edge towards another vertex of~$\mP[c]$, as~$\alpha$ is not in~$S$, and the only other vertex of~$\mP[c]$ that is adjacent to a neighbor of~$\mP[c]$ is the source~$s^c$, which cannot be reached from~$c$ by traversing only edges in the connector path~$\mP[c]$ (see~\Cref{fig gadget} and recall that~$c$ is associated with vertex~$8$).
Hence, assume that~$P$ uses an edge of the second type.
Then, $P$ arrives at a vertex of~$\Vin\cup\Vout$ at a time strictly larger than the out-time.
The only edges between vertices of~$V$ and vertices of~$V'\setminus V$ that exists later than the out-time are the edges between the~$\omega$ and the source vertices of the connector paths, and the in-edges.
Since~$\omega$ has no incident in-edge, $P$ cannot use any in-edge, as~$P$ would then end in a vertex distinct from~$\omega$ due to~\Cref{claim dont go in}. 
Moreover, as already discussed, there is no temporal path inside a connector path from the docking point to the source.
Thus, $P$ can also not traverse any edge between a source vertex~$s^c$ and~$\omega$.
This contradicts the assumption that $P$ reaches~$\omega$.
\end{itemize}

Note that the Cases 2--4 cover all possibilities in which~$S$ contains a vertex of~$\Vconn$.
In the remaining two cases, we will thus assume that~$S$ contains no vertex of~$\Vconn$, and by Case 1, exactly one of~$\alpha$ and~$\omega$.
Before we come to these cases, we observe the following.

\begin{claim}\label{claim dont go out}
If~$S$ contains no vertex of~$\Vconn$, then each temporal path in~$\G[S]$ that traverses an out-edge~$\{v,\vout[v]\}$ from~$v$ to~$\vout[v]$ ends in a vertex of~$\mP[{\vout[v]}]$.
\end{claim}
\begin{claimproof}
Recall that each out-edge receives the out-time as label.
Thus, if a temporal path~$P$ in~$\G[S]$ traverses an out-edge from~$v$ to~$\vout[v]$, then the only edges with labels larger than the out-time that are incident with~$\vout[v]$ are (i)~edges towards vertices of~$\Vconn$ and (ii)~edges of the blue path of~$\mP[{\vout[v]}]$.
As~$S$ does not contain any vertex of~$\Vconn$, $P$ thus has to continue with an edge of the blue path of~$\mP[{\vout[v]}]$.
Since these edges receive labels from the blue-times, which are strictly later than any edge between a vertex of~$\mP[{\vout[v]}]$ and a neighbor of~$\mP[{\vout[v]}]$, this implies that~$P$ ends in a vertex of~$\mP[{\vout[v]}]$.
\end{claimproof}
\begin{claim}\label{guaranteed neighbors}
If~$S$ contains no vertex of~$\Vconn$ and at most one vertex of~$\{\alpha,\omega\}$, then for each vertex~$\vin\in S\cap \Vin$, $v$ is in~$S$, and for each vertex~$\vout[w]\in S\cap \Vout$, $w$ is in~$S$.
\end{claim}
\begin{claimproof}
Recall that if~$\vin\in S$ ($\vout[w]\in S$), then by~\Cref{if s contains neither al nor om}, $S$ contains at least two neighbors of~$\mP[{\vin}]$ ($\mP[{\vout[w]}]$).
Now, as~$v$ ($w$) is the only neighbor of~$\vin$ ($\vout[w]$) outside of~$\Vconn \cup \{\alpha,\omega\}$, we conclude that~$v$ ($w$) is in~$S$.
\end{claimproof}

\begin{itemize}
\item \textbf{Case 5:} $c\in \Vin\cup \Vout$ and~$S$ contains~$\alpha$ but no vertex of~$\Vconn \cup \{\omega\}$\textbf{.}
We first show that~$S$ contains at least one vertex~$v^*$ of~$V\setminus \{\alpha,\alpha'\}$.

Assume that~$S\cap V \subseteq \{\alpha,\alpha'\}$.
Since~$\omega\notin S$, \Cref{if s contains neither al nor om} implies that for each~$d\in \Vin\cup\Vout$ for which~$S$ contains at least one vertex of~$\mP[d]$, $S$ also contains~$d$.
Furthermore, \Cref{guaranteed neighbors} thus implies that~$S \setminus V$ can only contain vertices from the two connector paths~$\mP[{\vin[\alpha]}]$ and~$\mP[{\vin[\alpha']}]$.
Moreover, $S$ can contain vertices from at most one of the two connector paths, as~$S$ is 2-vertex-connected (see \Cref{connectivity}) and~$G'[S\setminus \{\alpha\}]$ would otherwise have two connected components.
Thus, $S \subseteq Q_1 := \{\alpha,\alpha'\} \cup V(\mP[{\vin[\alpha]}])$ or~$S \subseteq Q_2 := \{\alpha,\alpha'\} \cup V(\mP[{\vin[\alpha']}])$.
In the first case, $S$ would contain only one neighbor of~$\mP[{\vin[\alpha]}]$, which would contradict that~$S$ is a closed tcc by~\Cref{if s contains neither al nor om}.
Hence, $S\subseteq Q_2$.
However, $\G[Q_2]$ fulfills the requirement of the temporal graph~$\G_2$ described in Item~\ref
{gadget plus alpha plus x1} of~\Cref{gadget extensions no tcc}.
As~$\G_2$ does not contain a nontrivial tcc, this contradicts the assumption that~$S$ is a closed tcc in~$\G$. 

Thus, $S$ contains a vertex~$v^*\in V \setminus \{\alpha,\alpha'\}$.
This implies that~$v^*\in V_j$ for some~$j>2$, as~$V_1 = \{\alpha\}$ and~$V_2 = \{\alpha'\}$.
Consider any temporal path~$P$ from~$v^*$ to~$\alpha$ in~$\G[S]$.
By definition of the labels between the vertices of~$V$, $P$ has to visit at least one vertex outside of~$V$.
To see this, recall that all labels assigned to~$E(V_i,V_{i+1})$ are strictly smaller than all labels assigned to~$E(V_j,V_{j+1})$ for~$1\leq i < j < k$.
However, all edges that are incident with both a vertex of~$V$ and a vertex outside of~$V$ are (i)~the in-edges, (ii)~the out-edges, and (iii)~edges incident with~$\alpha$ or~$\omega$.
Due to~\Cref{claim dont go in} and~\Cref{claim dont go out}, $P$ cannot traverse any in-edge or any out-edge, as otherwise~$P$ would end in a vertex of a connector path and not in~$\alpha$.
Furthermore, $P$ cannot traverse an edge incident with~$\omega$, as~$\omega$ is not in~$S$.
Finally, $P$ can also not traverse an edge incident with~$\alpha$ to leave~$V$, as~$P$ is a path (and thus visits each vertex at most once) and ends in~$\alpha$. 
This contradicts the assumption that $P$ reaches~$\alpha$.

\item \textbf{Case 6:} $c\in \Vin\cup \Vout$ and~$S$ contains~$\omega$ but no vertex of~$\Vconn \cup \{\alpha\}$\textbf{.}
This case is similar to the previous one.
We first show that~$S$ contains at least one vertex~$v^*$ of~$V\setminus \{\omega,\omega'\}$.

Assume that~$S\cap V \subseteq \{\omega,\omega'\}$.
Since~$\alpha\notin S$, \Cref{if s contains neither al nor om} implies that for each~$d\in \Vin\cup\Vout$ for which~$S$ contains at least one vertex of~$\mP[d]$, $S$ also contains~$d$.
Furthermore, \Cref{guaranteed neighbors} thus implies that~$S \setminus V$ can only contain vertices from the two connector paths~$\mP[{\vout[\omega]}]$ and~$\mP[{\vout[\omega']}]$.
Moreover, $S$ can contain vertices from at most one of the two connector paths, as~$S$ is 2-vertex-connected (see \Cref{connectivity}) and~$G'[S\setminus \{\omega\}]$ would otherwise have two connected components.
Thus, $S \subseteq Q_3 := \{\omega,\omega'\} \cup V(\mP[{\vout[\omega]}])$ or~$S \subseteq Q_4 := \{\omega,\omega'\} \cup V(\mP[{\vout[\omega']}])$.
In the first case, $S$ would contain only one neighbor of~$\mP[{\vout[\omega]}]$, which would contradict that~$S$ is a closed tcc by~\Cref{if s contains neither al nor om}.
Hence, $S\subseteq Q_4$.
However, $\G[Q_4]$ fulfills the requirement of the temporal graph~$\G_4$ described in Item~\ref
{gadget plus omega plus xk} of~\Cref{gadget extensions no tcc}.
As~$\G_4$ does not contain a nontrivial tcc, this contradicts the assumption that~$S$ is a closed tcc in~$\G$.

Thus, $S$ contains a vertex~$v^*\in V \setminus \{\omega,\omega'\}$.
This implies that~$v^*\in V_j$ for some~$j<k-1$, as~$V_{k-1} = \{\omega'\}$ and~$V_k = \{\omega\}$.
Consider any temporal path~$P$ from~$\omega$ to~$v^*$ in~$\G[S]$.
By definition of the labels between the vertices of~$V$, $P$ has to visit at least one vertex outside of~$V$.
To see this, recall that all labels assigned to~$E(V_i,V_{i+1})$ are strictly smaller than all labels assigned to~$E(V_\ell,V_{\ell+1})$ for~$1\leq i < \ell < k$.
However, all edges that are incident with both a vertex of~$V$ and a vertex outside of~$V$ are (i)~the in-edges, (ii)~the out-edges, (iii)~edges incident with~$\alpha$, and (iv)~edges between~$\omega$ and the source vertices of the connector paths.
Due to~\Cref{claim dont go out} and \Cref{claim dont go in}, $P$ cannot traverse any in-edge or any out-edge, as otherwise~$P$ would end in a vertex of a connector path and not in~$v^*$.
Furthermore, $P$ cannot traverse an edge incident with~$\alpha$, as~$\alpha$ is not in~$S$.
Finally, $P$ cannot traverse an edge from~$\omega$ to a source vertex of a connector path, as this edge receives a label from the~$\omega$-times, which are strictly larger than the in-time, which is the largest label incident with~$v^*$. 
This contradicts the assumption that $P$ reaches~$v^*$.
\end{itemize}
As the case distinction is exhaustive and each case leads to a contradiction, we conclude that~$S$ contains~$\alpha$ and~$\omega$.
\end{proof}

Next, we show that the containment of~$\alpha$ and~$\omega$ also implies that existence of at least one vertex per color class in~$S$.

\begin{lemma}
\label{undir path from al to om}
$S$ contains for each~$i\in [1,k]$ a vertex~$x_i\in V_i$, such that~$(\alpha = x_1,\dots, x_k=\omega)$ is a temporal path in~$\G$.
\end{lemma}
\begin{proof}
Due to~\Cref{undir both al and om}, we know that~$S$ contains~$\alpha$ and~$\omega$.
As~$\G[S]$ is temporally connected, this implies that there is a temporal path~$P$ from~$\alpha$ to~$\omega$ in~$\G[S]$.
We show that~$P$ has the desired structure, that is, $P=(\alpha = x_1,\dots, x_k=\omega)$, where~$x_i\in V_i$ for each~$i\in [1,k]$.
To this end, we show that~$P$ only visits vertices of~$V$.
Based on the structure of~$G'[V]$, this then immediately implies that~$P=(\alpha = x_1,\dots, x_k=\omega)$, where~$x_i\in V_i$ for each~$i\in [1,k]$, since each vertex of degree larger than 2 in~$G'[V]$ is part of a color class of size one and is thus contained in every path from~$\alpha$ to~$\omega$ in~$G'[V]$ (see~\Cref{fig selection edges}). 

Consider the edges incident with~$\alpha$.
There are three types of such edges:
\begin{itemize}
\item the edges towards the sinks of the connector paths,
\item the edge~$\{\alpha,\vin[\alpha]\}$, and
\item the edge~$\{\alpha,\alpha'\}$.
\end{itemize}
Recall that there is no out-edge for~$\alpha$, that is, there is no vertex~$\vout[\alpha]$.
We show that starting~$P$ with any of the first two types of edges would violate the property that~$P$ reaches~$\omega$.

If starting~$P$ with an edge of the first type, we end in the sink vertex~$t^c$ of some connector path~$\mP[c]$ while using an edge that receives a label from the~$\alpha$-times.
As~$\alpha$ is the only neighbor of~$t^c$ outside of~$\mP[c]$, the path~$P$ would need to continue inside of~$\mP[c]$ and at a time that appears after the~$\alpha$-times.
By definition of the connector paths, there is only one such edge incident with~$t^c$, namely the unique edge of the blue path of~$\mP[c]$.
However, this edge receives a time from the blue-times which are strictly later than all times assigned to edges incident with~$\omega$.
Thus, $P$ cannot use an edge of the first type as its first edge.

Similarly, if~$P$ traverses an edge~$\{v,\vin[v]\}$ from~$v$ to~$\vin[v]$ for some~$v\in V\setminus \{\omega,\omega'\}$, then the path can only continue at the blue-times, as the edge~$\{v,\vin[v]\}$ receives the largest time of all edges incident with~$\vin[v]$ that is not part of the blue path of~$\mP[{\vin[v]}]$.
Thus, $P$ cannot traverse an edge~$\{v,\vin[v]\}$ from~$v$ to~$\vin[v]$ for some~$v\in V\setminus \{\omega,\omega'\}$, since each edge of the blue path of~$\mP[{\vin[v]}]$ receives a label from the blue-times which are strictly later than all times assigned to edges incident with~$\omega$.

In particular, $P$ cannot use the edge~$\{\alpha,\vin[\alpha]\}$ as its first edge.

This implies that the first edge of~$P$ is~$\{\alpha,\alpha'\}$.
As this edge receives a label from the selection-times, no edge of~$P$ can traverse any edge at the out-time.
By construction, each vertex of~$V\setminus \{\alpha,\omega\}$ has at most two incident edges that are not incident with other vertices of~$V$.
These are the edges at the in-time and the out-time, both of which~$P$ cannot traverse as discussed above.
Hence, the underlying path of~$P$ in~$G'$ contains only vertices of~$V$, which implies that~$P$ has the desired property due to the structure of~$G'[V]$ (see~\Cref{fig selection edges}).
\end{proof}

Let~$X:= \{x_i\mid 1\leq i \leq k\}$.
In the remainder of the proof we show that~$X$ is a clique in~$G$.
This then implies that~$G$ contains a clique of size~$k$, which proves~\Cref{undirected direction backward}.

\begin{lemma}\label{x is clique undirected happy}
$X$ is a clique of size~$k$ in~$G$.
\end{lemma}
\begin{proof}
Let~$1\leq i < j \leq k$.
We show that~$\{x_i,x_j\}$ is an edge of~$G$.
If~$j = i+1$, this trivially holds, as the edges between the vertices of~$V$ in~$\G$ are a subset of the edges of~$G$, and~$\{x_i,x_{i+1}= x_j\}$ is an edge of the path from~$\alpha$ to~$\omega$ stated in~\Cref{undir path from al to om}.
Now consider~$j \geq i+2$.
If~$i\in \{1,2\}$ (that is, if~$x_i \in \{\alpha,\alpha'\}$), then the statement trivially holds, since~$x_i$ is by assumption adjacent to each other vertex of~$V$ in~$G$.
Similarly, if~$j\in \{k-1,k\}$ (that is, if~$x_j \in \{\omega,\omega'\}$), then the statement also trivially holds, since~$x_j$ is by assumption adjacent to each other vertex of~$V$ in~$G$.
So, assume that~$x_i$ and~$x_j$ are not from~$\{\alpha,\alpha',\omega,\omega'\}$.
By the fact that~$S$ is a closed tcc in~$\G$ and~$X\subseteq S$, there is a temporal path~$P$ from~$x_j$ to~$x_i$ in~$\G$.
We show that~$P= (x_j,\vout[x_j],\vcon[x_j,x_i],\vin[x_i],x_i)$.
This then implies that the vertex~$\vcon[x_j,x_i]$ exists, which by construction then implies that~$\{x_i,x_j\}$ is an edge of~$G$.

Consider the edges incident with the vertices~$x_i$ and~$x_j$.
For each~$\ell\in \{i,j\}$, $x_\ell$ has only edges towards (i)~$\vout[x_\ell]$ (labeled with the out-time), (ii)~$\vin[x_\ell]$ (labeled with the in-time), and (iii)~each neighbor of~$x_\ell$ of~$V_{\ell-1}\cup V_{\ell+1}$ in~$G$ (labeled with a label from the selection-times).
Hence, $P$ starts later than the~$\alpha$-times and ends prior to the~$\omega$-times.

As already argued in the proof of~\Cref{undir not only from V}, $P$ has to use at least one vertex outside of~$V$, as all labels incident with~$x_j$ towards other vertices of~$V$ are strictly larger than all labels incident with~$x_i$ towards other vertices of~$V$. 
This implies that~$P$ starts by traversing the edge~$\{x_j,\vout[x_j]\}$.
To see this, let~$q$ be the first vertex of~$P$ that is outside of~$V$.
The only possible choices for~$q$ are from~$\Vin\cup \Vout$.
Due to~\Cref{claim dont go in}, $q\notin \Vin$, as otherwise, $P$ would end in a vertex distinct from~$x_i$.
Thus~$q\in \Vout$.
However, as the unique edge between~$q$ and a vertex of~$V$ receives the out-label, which is the smallest label incident with~$x_j$, $P$ cannot traverse any edge prior to the edge that leads to~$q$, which gives us that~$q$ is the unique neighbor of~$\Vout$ of~$x_j$.
That is, $q = \vout[x_j]$.
Now consider the remainder of the path.
As~$P$ reaches a vertex of~$\Vout$ at the out-time and there are no edges between any vertex of~$V$ and any vertex of~$V'\setminus V$ in the times between the out-time and the in-time, $P$ has to contain the edge~$\{\vin[x_i],x_i\}$.
This is due to the fact that this edge receives the in-time, which is the largest label incident with~$x_i$.

So far, we thus know that~$P$ starts by traversing the edge~$\{x_j,\vout[x_j]\}$, ends by traversing the edge~$\{\vin[x_i],x_i\}$, and avoids vertices of~$V$ in between.
Note that by construction, each vertex of~$\Vconn$ has exactly one neighbor in~$\Vin$ and exactly one neighbor in~$\Vout$ and the edge towards the neighbor in~$\Vout$ receives a smaller label than the edge towards the neighbor in~$\Vin$.
Moreover, in the times strictly between the out-time and the in-time, each vertex of~$\Vout$ has only edges towards vertices of~$\Vconn$, each vertex of~$\Vin$ has only edges towards vertices of~$\Vconn$, and each vertex of~$\Vconn$ has only edges towards vertices of~$\Vin\cup \Vout$.
This implies that there is no path from any vertex of~$\Vin$ to any vertex of~$\Vout$ that uses only labels strictly between the out-time and the in-time.
As a consequence, $P$ has exactly one vertex~$q\in \Vconn$ between~$\vout[x_j]$ and~$\vin[x_i]$.
That is, $P=(x_j, \vout[x_j],q,\vin[x_i],x_i)$.
By construction, the only vertex of~$\Vconn$ that can be adjacent to both~$\vout[x_j]$ and~$\vin[x_i]$ is~$\vcon[x_j,x_i]$.
As this vertex only exists if~$\{x_i,x_j\}$ is an edge of~$G$, we conclude that~$X$ is a clique in~$G$.
\end{proof}

The correctness of the reduction now follows from~\Cref{undirected direction forward} and~\Cref{undirected direction backward}.

\subsection{Implications}

Recall that if we are dealing with a yes-instance of MCC in the reduction, then there is a closed tcc of size exactly~$k + 20\cdot \left(2\cdot (k-2) +  \binom{k}{2} - (k-1) \right) \in \Oh(k^2)$ in~$\G$, where~$k$ is the size of the sought clique in the MCC instance.
Otherwise, $\G$ contains no closed tcc of size larger than~$2$.
As MCC is~\W1-hard when parameterized by~$k$, this implies also parameterized intractability for the following famous problem even on happy temporal graphs and thus in particular in the non-strict model.

\prob{\CTCC}{A temporal graph~$\G$ and an integer~$k$.}{Does~$\G$ have a closed temporally connected component of size exactly~$k$?}

\begin{theorem}[\Cref{hardness exact tcc} reformulated]
\CTCC is \W1-hard when parameterized by~$k$ even on happy undirected temporal graphs.
\end{theorem}

That is, we get Implication~\ref{imp hardness question}.

Next, we argue the following, which is important for Implication~\ref{imp Counting}.

\begin{lemma}\label{imp 4 happy}
Let~$X$ be a clique of size~$k$ in~$G$.
Then $S^* = X \cup (V' \setminus V)$ is a maximal closed tcc in~$\G$.
\end{lemma}
\begin{proof}
First, we argue that~$S^*$ is a closed tcc.
This essentially follows by the same arguments that were used to prove~\Cref{undirected direction forward}.
Let~$S$ be the closed tcc for which the existence is shown in~\Cref{undirected direction forward} based on the existence of~$X$.
By definition of~$S$, $S^* = S \cup (V'\setminus V)$, that is, we can add the vertices of the whole connector gadget (the vertices of all connector paths) to~$S$ without violating the property of being a closed tcc.
This is due to the following: The reachability between vertices of~$S$ is done inside of~$\G[S]$, and each other vertex is part of some connector path and can reach all other vertices (and can be reached by them) via the same arguments used in the proof of~\Cref{reachability outside V is trivial}.

Next, we show that~$S^*$ is maximal.
Assume towards a contradiction that there is a closed tcc~$S' \supsetneq S^*$.
This set contains at least one vertex~$y_i$ of~$V\setminus X$, as $S^*$ contains all of~$V'\setminus V$.
Let~$V_i$ be the color class of~$y_i$.
We show that $x_i$ and $y_i$ cannot reach each other in both directions.
This then contradicts the assumption that~$\G[S']$ is temporally connected.
By construction and the fact that $V_i$ has size more than~$1$, for each~$z\in \{x,y\}$, $z_i$ has two neighbors in the selection gadget: some vertex~$z_{i-1}$ in~$V_{i-1}$ and some vertex~$z_{i+1}$ in~$V_{i+1}$ (see~\Cref{fig selection edges}).
Moreover, exactly one of the following holds: $x_{i-1} = y_{i-1}$ or~$x_{i+1} = y_{i+1}$.
Let~$q$ be that common neighbor and note that one temporal path of length~$2$ between~$x_i$ and~$y_i$ exists based on the existence of~$q$. 
Without loss of generality, assume that~$x_i$ can reach~$y_i$ over the path~$(x_i,q,y_i)$.
This implies that~$\lambda(\{x_i,q\}) < \lambda(\{y_i,q\})$.
We show that there is no temporal path from~$y_i$ to~$x_i$ in~$\G$. 
As we assume that~$\G[S']$ is temporally connected, there has to be a temporal path~$P$ from~$y_i$ to~$x_i$.
Similar to the arguments used to prove~\Cref{x is clique undirected happy}, $P$ cannot use any vertex outside of~$V$, since~$\{x_i,y_i\}$ is not an edge, which implies that there is no vertex~$\vcon[y_i,x_i]$ in~$\G$.
This in particular implies that~$P$ arrives at~$x_i$ not later than the time assigned to edge~$\{x_i, x_{i+1}\}$, which is the largest label incident with~$x_i$ in~$\G[V]$.

Moreover, $P$ leaves~$y_i$ not earlier than the time assigned to edge~$\{y_i, y_{i-1}\}$, which is the smallest label incident with~$y_i$ in~$\G[V]$.
As the labels along~$P$ strictly increase, and for each~$j\in [2,k-1]$, all labels assigned to~$E(V_{j-1}, V_j)$ are strictly smaller than all labels assigned to~$E(V_j, V_{j+1})$, all edges of~$P$ belong to~$E(V_{i-1}, V_i) \cup E(V_i, V_{i+1})$.
In the subgraph of~$\G$ formed by these edges, $q$ is adjacent to all vertices of~$V_i$, each vertex of~$V_i$ has exactly one further neighbor, and each such neighbor has degree~$1$ in this subgraph.
Hence, the only path from~$y_i$ to~$x_i$ in this subgraph is~$(y_i, q, x_i)$, which requires $\lambda(\{y_i, q\}) < \lambda(\{q, x_i\})$.
This contradicts $\lambda(\{x_i, q\}) < \lambda(\{q, y_i\})$.
Consequently, there is no temporal path from~$y_i$ to~$x_i$ in~$\G$, a contradiction to $S'$ being a closed tcc.
\end{proof}

\section{Hardness for Strict Undirected Temporal Graphs of Constant Lifetime}\label{sec hard strict}
We now show hardness also for undirected temporal graphs of constant lifetime in the strict setting.

\begin{theorem}\label{hardness undir strict}
\NCTCC is \NP-hard on simple undirected temporal graphs of lifetime~$55$ in the strict setting.
\end{theorem}
We mostly follow the reduction behind~\Cref{hardness undir happy}.
The main differences are 
\begin{itemize}
\item we assign only~$\Oh(1)$ distinct time labels per time interval,
\item the vertices of~$\Vconn$ (and their respective connector paths) are not added, but there are now edges between vertices of~$\Vin$ and~$\Vout$, and
\item there are edges between vertices of~$V$ that are not necessarily in consecutive color classes.
\end{itemize}

We also assume a slightly different restriction on the structure of the graph from the MCC instance.
Let~$I= (G=(V,E),k)$ be an instance of~MCC and let~$V_1 \cup \dots \cup V_k$ be a~$k$-partition of~$V$.
We can assume that~$k > 20$, $k \mod 6 = 4$, and that $V_i$ has size~$1$ for each~$i\in [1,k]$ with~$i < 11$ or~$i\mod 6 \neq 5$.
Moreover, we can assume that~$G[V_i \cup V_{i+1}]$ is a complete biclique, and that for each color class of size one, the unique vertex of that color class is adjacent to each other vertex of the graph, as otherwise, we are dealing with a trivial no-instance of MCC.\footnote{To achieve this property, we can simply add sufficiently many color classes to the instance that all contain only a single vertex which is adjacent to every other vertex of the graph.}
Note that this implies that for example the color classes~$V_1,V_2,V_{k-1}$, and~$V_k$ have size~$1$.

\subsection{Construction}
We start the construction:
We initialize the \foot~$G' := (V',E')$ as the graph~$(V,E' := \bigcup_{i=1}^{k-1} E_G(V_i,V_{i+1}))$.

Recall that~$|V_i| = 1$ for all~$i\in [1,k]$ with~$i< 11$ or~$i\mod 6 \neq 5$. 
For each such~$i$, we denote the unique vertex of~$V_i$ as~$v_i^*$.
We use the shorthands~$\alpha := v_1^*$, $\alpha' = v_2^*$, $\omega:= v_k^*$, and~$\omega' := v_{k-1}^*$.
We now add additional edges between some vertices of size-1 color classes.
See~\Cref{fig undir strict selection} for an illustration.
Namely, we add
\begin{itemize}
\item the edges~$\{v_3^*,v_8^*\}$ and~$\{v_8^*,\omega'\}$, and
\item for each~$i\in [8,k-8]$ with~$i \mod 6 = 2$, the edges~$\{v_5^*, v_{i+1}^*\}$, $\{v_{i+6}^*,\omega'\}$, and~$\{v_{i-2}^*,v_{i+6}^*\}$.
\end{itemize}
This completes the definition of the edges between the vertices of~$V$ in~$\mg$.

\begin{figure}
\centering
\scalebox{.8}{
\begin{tikzpicture}[scale=1.7,square/.style={regular polygon,regular polygon sides=4,inner sep=1pt}]

\tikzstyle{k}=[circle,fill=white,draw=black,minimum size=8pt,inner sep=2pt]

\def\iii{j}
\foreach \y in {3,...,8}{
\node (\iii\y) at (-4 + \y,0) {$v_{\y}^*$};
}

\node (a) at (-3,0) {$\alpha$};
\node (a2) at (-2,0) {$\alpha'$};

\node (o) at (8,0) {$\omega$};
\node (o2) at (7,0) {$\omega'$};

\draw[thick,-, bend left = 0] (a) edge node[black,fill=white] {1} (a2);
\draw[thick,-, bend left = 0] (a2) edge node[black,fill=white] {2} (\iii3);
\draw[thick,-, bend left = 0] (\iii3) edge node[black,fill=white] {3} (\iii4);
\draw[thick,-, bend left = 0] (\iii4) edge node[black,fill=white] {4} (\iii5);
\draw[thick,-, bend left = 0] (\iii5) edge node[black,fill=white] {5} (\iii6);
\draw[thick,-, bend left = 0] (\iii6) edge node[black,fill=white] {9} (\iii7);
\draw[thick,-, bend left = 0] (\iii7) edge node[black,fill=white] {10} (\iii8);
\draw[thick,-, bend left = 40] (\iii3) edge node[black,fill=white] {11} (\iii8);

\draw[thick,-, bend left = 0] (o2) edge node[black,fill=white] {11} (\iii8);
\draw[thick,-, bend left = 0] (o) edge node[black,fill=white] {12} (o2);

\begin{scope}[yshift=-3cm,xshift=20pt]
\def\iii{i}

\foreach \y in {-2,-1}{
\node (\iii\y) at (0 + \y,0) {$v_{i\y}^*$};
}

\foreach \y in {0}{
\node (\iii\y) at (0 + \y,0) {$v_{i}^*$};
}

\foreach \y in {3}{
\draw [draw=black] (.2 + \y,-1) rectangle (-.2 + \y,1);

\node (\iii\y) at (0 + \y,-1.2) {$V_{i+\y}$};
\node[knoten] (\iii\y1) at (0 + \y,-.75) {};
\node[knoten] (\iii\y2) at (0 + \y,0) {};
\node[knoten] (\iii\y3) at (0 + \y,.75) {};
}

\foreach \y in {1,2,4,5,6}{
\node (\iii\y) at (0 + \y,0) {$v_{i+\y}^*$};
}

\draw[thick,-, bend left = 0] (\iii-2) edge node[black,fill=white] {9} (\iii-1);
\draw[thick,-, bend left = 0] (\iii-1) edge node[black,fill=white] {10} (\iii0);
\draw[thick,-, bend left = 0] (\iii0) edge node[black,fill=white] {5} (\iii1);
\draw[thick,-, bend left = 0] (\iii1) edge node[black,fill=white] {6} (\iii2);
\draw[thick,-, bend left = 0] (\iii2) edge node[black,fill=white] {7} (\iii31);
\draw[thick,-, bend left = 0] (\iii2) edge node[black,fill=white] {7} (\iii32);
\draw[thick,-, bend left = 0] (\iii2) edge node[black,fill=white] {7} (\iii33);
\draw[thick,-, bend left = 0] (\iii4) edge node[black,fill=white] {8} (\iii31);
\draw[thick,-, bend left = 0] (\iii4) edge node[black,fill=white] {8} (\iii32);
\draw[thick,-, bend left = 0] (\iii4) edge node[black,fill=white] {8} (\iii33);
\draw[thick,-, bend left = 0] (\iii4) edge node[black,fill=white] {9} (\iii5);
\draw[thick,-, bend left = 0] (\iii5) edge node[black,fill=white] {10} (\iii6);
\draw[thick,-, bend left = 35] (\iii-2) edge node[black,fill=white] {11} (\iii6);

\end{scope}

\draw[thick,-, bend left = 0] (j5) edge node[black,fill=white] {5} (i1);
\draw[thick,-, bend left = 0] (o2) edge node[black,fill=white] {11} (i6);
\draw[thick,-, bend left = 0] (o2) edge node[black,fill=white] {11} (i0);

\node (h1) at ($.5*(j6)+ .5*(8,0)$) {};
\node (h11) at ($(h1) + (0,1)$) {};
\draw[] (j6) edge [dashed,thick,bend left = 20] node[black,fill=white, very near end] {11} (h11);

\node (h2) at ($(i0) - (2.3,0)$) {};
\node (h21) at ($(h2) + (0,1.2)$) {};
\draw[] (i0) edge [dashed,thick,bend right = 20] node[black,fill=white, very near end] {11} (h21);

\node (h3) at ($(i4)+ (2.3,0)$) {};
\node (h31) at ($(h3) + (0,1.2)$) {};
\draw[] (i4) edge [dashed,thick,bend left = 20] node[black,fill=white, very near end] {11} (h31);

\end{tikzpicture}
}
\caption{The selection gadget for the reduction with strict temporal paths and constant lifetime.
For each~$i\in[8,k-6]$ with~$i\mod 6 = 2$, the lower gadget exists for the color classes between~$V_{i}$ and~$V_{i+6}$ and overlap.
The dashed edges with label~$11$ are part of such a gadget for the previous or next~$i$.
Note that only the color classes~$V_{i+3}$ contain more than one vertex and all other color classes are assumed to have size~$1$.}
\label{fig undir strict selection}
\end{figure}

\begin{figure}
\centering
\scalebox{.75}{
\begin{tikzpicture}[square/.style={regular polygon,regular polygon sides=4},yscale = .8]

\begin{scope}[xshift = 2cm]
\node[knoten] (a) at (0,-6.5) {$\alpha$};

\node[knoten] (o) at (0,2.5) {$\omega$};
\end{scope}

\begin{scope}
\smallgadget

\draw[thick] (20) edge node[black,fill=white] {$\alpha$-time} (a);
\draw[thick] (1) edge node[black,fill=white] {$\omega$-time} (o);
\node[square,draw] (cl) at (8) {8};

\node (ll) at ($(1) + (-1,0)$) {$\mP[{\vout[u]}]$};
\end{scope}

\begin{scope}[xshift = 4cm]
\smallgadget

\draw[thick] (20) edge node[black,fill=white] {$\alpha$-time} (a);
\draw[thick] (1) edge node[black,fill=white] {$\omega$-time} (o);
\node[square,draw] (cr) at (8) {8};

\node (ll) at ($(1) + (-1,0)$) {$\mP[{\vin[v]}]$};
\end{scope}

\draw[thick] (cl) edge node[black,fill=white] {out2in-time} (cr);

\node[knoten] (u) at ($(cl) + (-3,0)$) {$u$};
\node[knoten] (v) at ($(cr) + (3,0)$) {$v$};

\draw[thick] (u) edge node[black,fill=white] {out-time} (cl);
\draw[thick] (v) edge node[black,fill=white] {in-time} (cr);

\end{tikzpicture}
}
\caption{An illustration of the paths between vertices of~MCC instance in the reduction behind~\Cref{hardness undir strict}.
For each~$(u,v)\in A_\conn$, we have a temporal path from~$u$ to~$v$ in our temporal graph, namely, the path $(u,\vout[u],\vin[v],v)$.
Similar to~\Cref{fig gadget extensions}, the rectangular areas indicate the connector paths~$\mP[{\vout[u]}]$ and~$\mP[{\vin[v]}]$, where the respective vertex~$8$ of each such path denotes the docking point of that connector path, that is, the vertices~$\vout[u]$ and~$\vin[v]$, respectively.
Note that in comparison to the reduction behind~\Cref{hardness undir happy} (see~\Cref{fig connection happy}), the vertices~$\vout[u]$ and~$\vin[v]$ are adjacent and there is no connector path with an intermediate vertex anymore. 
}
\label{fig connection strict}
\end{figure}

Next, we define the in-vertices and out-vertices identical to the reduction behind~\Cref{hardness undir happy}.
We initialize the two sets~$\Vin$ and $\Vout$ as empty sets.
For each~$v\in V \setminus \{\omega,\omega'\}$, we add a vertex~$\vin[v]$ to~$\Vin$, and for each~$v\in V \setminus \{\alpha,\alpha'\}$, we add a vertex~$\vout[v]$ to~$\Vout$.

For each~$c\in \Vin \cup \Vout$, we now add a copy of a connector path~$\mP[c]$ to~$G'$ (see~\Cref{def con path}), where we rename the vertex~$8$ of~$\mP[c]$ to~$c$, to which we may refer to as the~\emph{docking point} of~$\mP[c]$.
That is, now~$\Vin \cup \Vout$ is a subset of the vertices of~$G'$.
We may call vertex~$1$ and~$20$ of~$\mP[c]$ the~\emph{source} and~\emph{sink} of~$\mP[c]$ respectively.

We complete the construction of the \foot by adding edges mostly identical to the reduction behind~\Cref{hardness undir happy}.
The main difference is that we directly add edges between vertices of~$\Vout$ and~$\Vin$.
These edges will be added only for a specific subset of vertex pairs of~$V$, namely for:
\begin{align*}
A_\conn =~ & \{(x,y),(y,x)\mid \{x,y\}\in E\} ~ \setminus & \hspace{-9pt} \Big(\{(\alpha,x), (\alpha',x), (x,\omega'), (x,\omega)\mid x\in V\} \cup \{(v_5^*,v_3^*)\}  \\
&&\cup \{(v_{i}^*,v_{i-2}^*)\mid i\in [8,k-8], i\mod 6 = 2\}\Big).
\end{align*}
We now start by defining the remaining edges of~$\mg$.
\begin{enumerate}
\item For each~$c\in \Vin \cup \Vout$, we add the edge~$\{t^c,\alpha\}$ to~$G'$, where~$t^c$ is the sink of~$\mP[c]$.
\item For each vertex~$v\in V \setminus \{\alpha,\alpha'\}$, we add an edge~$\{v,\vout[v]\}$ to~$G'$ (called the~\emph{out-edge of~$v$}).
\item For each~$(x,y)\in A_\conn$, we add the edge~$\{\vout[x],\vin[y]\}$ to~$G'$.
\item For each vertex~$v\in V \setminus \{\omega,\omega'\}$, we add an edge~$\{\vin[v],v\}$ to~$G'$ (called the~\emph{in-edge of~$v$}).
\item For each~$c\in \Vin \cup \Vout$, we add the edge $\{\omega,s^c\}$ to~$G'$, where~$s^c$ is the source of~$\mP[c]$.
\end{enumerate}

This completes the definition of the \foot~$G'$.

\begin{table}
\caption{The different intervals of time labels and the edges that receive that respective labels from these intervals.
The intervals `out-time', `in-time', `$\alpha$-time', `$\omega$-time', and `out2in-time' are intervals of length one.
If an interval is below another interval in this table, then each label of that interval is strictly larger than the largest label of the interval above.
For example, all labels of the selection-times are strictly smaller than the~$\omega$-time.}
\centering
\begin{tabular}{l|l|l}
name &  receiving edges & \#labels\\\hline
red-times & edges of the red paths of each connector path& 19\\
$\alpha$-time & edges between~$\alpha$ and the sink vertex of each connector path& 1\\
out-time & edges~$\{\{v,\vout[v]\}\mid v\in V \setminus \{\alpha,\alpha'\}\}$& 1\\
out2in-time &  edges~$\{\{\vout[x],\vin[y]\}\mid (x,y)\in A_\conn\}$& 1\\
selection-times & edges between the vertices of~$V$ in~$\G$& 12\\
in-time &  edges~$\{\{v,\vin[v]\}\mid v\in V \setminus \{\omega,\omega'\}\}$& 1\\
$\omega$-time & edges between~$\omega$ and the source vertex of each connector path& 1\\
blue-times  & edges of the blue paths of each connector path& 19
\end{tabular}
\label{tab times strict}
\end{table}

\begin{observation}
$G'$ does not contain a triangle.
\end{observation}

We complete the construction by defining the labels of the edges.
To this end, we define (similar to the reduction behind~\Cref{hardness undir happy}) several intervals that are pairwise disjoint and assign labels of these intervals to specific edge sets. 
See~\Cref{tab times strict} for a general overview.

\begin{itemize}
\item The first 19 time steps are the~\emph{red-times}. 
These are assigned to the red paths of all connector paths.
\item Afterwards, there is the~\emph{$\alpha$-time}.
This single label is assigned to the edges $\{\{t^c,\alpha\}\mid c\in \Vin\cup\Vout\}$.
Here, recall that~$t^c$ denotes the sink of the connector path~$\mP[c]$.
\item Afterwards, there is the~\emph{out-time}.
This single label is assigned to the edges $\{\{v,\vout[v]\}\mid v\in V \setminus \{\alpha,\alpha'\}\}$.
\item Afterwards, there is the~\emph{out2in-time}.
This label is assigned to the edges between vertices of~$\Vout$ and~$\Vin$.
\item Afterwards, there are the~\emph{selection-times}.
These are~$12$ different time labels that are assigned to the edges between the vertices of~$V$.
We will define the concrete way these labels are assigned at a later time.
\item Afterwards, there is the~\emph{in-time}.
This label is assigned to the edges $\{\{\vin[v],v\}\mid v\in V \setminus \{\omega,\omega'\}\}$.
\item Afterwards, there is the~\emph{$\omega$-time}.
This label is assigned to the edges $\{\{\omega,s^c\}\mid c\in \Vin\cup\Vout\}$.
Here, recall that~$s^c$ denotes the source of the connector path~$\mP[c]$.
\item The last 19 time steps are the~\emph{blue-times}. 
These are assigned to the blue paths of all connector paths.
\end{itemize}

It remains to define how the selection times are assigned to the edges between the vertices of~$V$ in~$\mg$ (see \Cref{fig undir strict selection}).
For simplicity, we define the selection times, that is, the labels assigned to the edges between vertices of~$V$, as the numbers of~$[1,12]$.
Note that this implies that all prior time-intervals only contain nonpositive time labels. 
We label the edges between the vertices of~$\{\alpha,\alpha',\omega,\omega'\} \cup \{v_\ell^*\mid \ell\in [3,8]\}$ as depicted in~\Cref{fig undir strict selection}.
For each~$i\in[8,k-6]$ with~$i\mod 6 = 2$,
\begin{itemize}
\item assign label~$5$ to the edge~$\{v_i^*,v_{i+1}^*\}$, label~$6$ to the edge~$\{v_{i+1}^*,v_{i+2}^*\}$, label~$9$ to the edge~$\{v_{i+4}^*,v_{i+5}^*\}$, and label~$10$ to the edge~$\{v_{i+5}^*,v_{i+6}^*\}$, 
\item assign label~$7$ to all edges of~$\{\{v_{i+2}^*,v_{i+3}\}\mid v_{i+3}\in V_{i+3}\}$ and assign label~$8$ to all edges of~$\{\{v_{i+3},v_{i+4}^*\}\mid v_{i+3}\in V_{i+3}\}$,
\item assign label~$11$ to the edges~$\{v_{i+6}^*,\omega'\}$ and~$\{v_{i+6}^*,v_{i-2}^*\}$ (this includes the edge~$\{v_6^*,v_{14}^*\}$), and
\item assign label~$5$ to the edge~$\{v_5^*,v_{i+1}^*\}$ (this includes the edge~$\{v_5^*,v_{9}^*\}$).
\end{itemize}

This completes the construction. 
Let~$\G$ be the resulting temporal graph.

Before we show the correctness of the reduction, let us observe several properties about temporal paths that traverse in-edges or out-edges.

\begin{claim}\label{claim dont go in strict}
Let~$x\in V\setminus \{\omega,\omega'\}$ and let~$c= \vin[x]$.
Each temporal path that traverses the edge~$\{x,c\}$ from~$x$ to~$c$ in~$\G$ ends in a vertex of the connector paths~$\mP[c]$ and each temporal path that traverses the edge~$\{x,c\}$ from~$c$ to~$x$ in~$\G$ ends in~$x$.
\end{claim}
The claim follows by the same arguments as~\Cref{claim dont go in}, which stated the same properties for the temporal graph of the reduction on happy temporal graphs.

\begin{claim}\label{claim go out means go in strict}
Let~$x\in V\setminus \{\alpha,\alpha'\}$ and let~$c= \vout[x]$.
Each temporal path that traverses the edge~$\{x,c\}$ from~$x$ to~$c$ in~$\G$ ends in a vertex of~$V(\mP[c]) \cup \bigcup_{y \colon (x,y)\in A_\conn} (\{y\} \cup V(\mP[{\vin[y]}]))$.
\end{claim}
\begin{claimproof}
Recall that the edge~$\{x,c\}$ receives the out-time.
The only incident edges with~$c$ that receive a larger label than the out-time, are edges towards other vertices of~$\mP[c]$ at the blue-times, and edges that receive the out2in-time towards vertices~$\vin[y]\in \Vin$ with~$(x,y)\in A_\conn$.
If a path traverses an edge of the first type, the path is guaranteed to end in~$\mP[c]$.
If a path traverses an edge towards a vertex~$\vin[y]\in \Vin$, the path can only continue with edges of two types:
the edge towards the vertex~$y$ at the in-time, or an edge towards another vertex of~$\mP[{\vin[y]}]$ at the blue-times.
By~\Cref{claim dont go in strict}, traversing the edge towards~$y$ lets the path end in~$y$.
Otherwise, as discussed above, traversing an edge that receives a label from the blue-times lets the path end in a vertex of~$\mP[{\vin[y]}]$.
Hence, each temporal path that traverses the edge~$\{x,c\}$ from~$x$ to~$c$ in~$\G$ ends in a vertex of~$V(\mP[c]) \cup \bigcup_{y \colon (x,y)\in A_\conn} (\{y\} \cup V(\mP[{\vin[y]}]))$.
\end{claimproof}

This also implies the following.
\begin{claim}\label{claim paths inside V strict}
Let~$x$ and~$y$ be distinct vertices of~$V$ with~$(x,y)\notin A_\conn$.
Each temporal path from~$x$ to~$y$ in~$\mg$ visits only vertices of~$V$.
\end{claim}
\begin{claimproof}
Assume that there is a temporal path~$P$ from~$x$ to~$y$.
We show that~$P$ visits only vertices of~$V$.
Consider the first edge of~$P$.
As~$(x,y)\notin A_\conn$, $P$ cannot start with the (potential) out-edge incident with~$x$ due to~\Cref{claim go out means go in strict}.
By~\Cref{claim dont go in strict}, $P$ also cannot traverse any in-edge~$\{y,\vin[y]\}$, as otherwise, the path would end in~$\mP[{\vin[y]}]$.
The only other possible edges that~$P$ could use as first edge are (i)~edges towards the sink vertices of the connector paths, if~$x = \alpha$, (ii)~edges towards the source vertices of the connector paths, if~$x = \omega$, or (iii)~edges between the vertices of~$V$.
If the path~$P$ would traverse an edge towards the sink~$t^c$ or source~$s^c$ of any connector path~$\mP[c]$, then~$P$ would end in a vertex of~$\mP[c]$.
This is due to the fact that any edge incident with~$s^c$ ($t^c$) that receives a label larger than the label of the edge from~$\omega$ ($\alpha$) is an edge of the blue path and thus receives a label strictly larger than the largest label incident with~$y$.
In particular, this implies that the first edge traversed by~$P$ is an edge incident with another vertex of~$V$.
This edge receives a label from the selection-times, which is strictly larger than the out-time.
Hence, $P$ can traverse none of the out-edges.
As these were the only edges between vertices of~$V$ and vertices outside of~$V$ that were not excluded from~$P$ by previous arguments, we conclude that~$P$ visits only vertices of~$V$.
\end{claimproof}

\begin{claim}\label{claim paths not in V strict}
Let~$x$ and~$y$ be distinct vertices of~$V$, such that there is no temporal path from~$x$ to~$y$ in~$\mg[V]$.
If there is a temporal path~$P$ from~$x$ to~$y$ in~$\mg$, then~$(x,y)\in A_\conn$ and~$P$ traverses the edge~$\{\vout[x],\vin[y]\}$.
\end{claim}
\begin{claimproof}
By~\Cref{claim paths inside V strict}, we immediately get that~$(x,y)\in A_\conn$, as otherwise this would contradict the assumption that~$P$ is not in~$\mg[V]$.
It remains to show that~$P$ traverses the edge~$\{\vout[x],\vin[y]\}$.
By the arguments used in~\Cref{claim paths inside V strict}, $P$ cannot traverse an edge from~$\alpha$ or~$\omega$ to the sink or the source of any connector path, as otherwise~$P$ would end in a vertex of that connector path.
Additionally, $P$ cannot traverse an in-edge~$\{z,\vin[z]\}$ from~$z$ to~$\vin[z]$ due to~\Cref{claim dont go in strict}.
As~$P$ is not entirely in~$\mg[V]$, the first edge of~$P$ is~$\{x,\vout[x]\}$.
The next edge then has to be~$\{\vout[x],\vin[y]\}$, as otherwise, $P$ would end in a vertex of~$V(\mP[{\vout[x]}]) \cup \bigcup_{z \colon (x,z)\in A_\conn, z \neq y} (\{z\} \cup V(\mP[{\vin[z]}]))$.
\end{claimproof}

\begin{observation}\label{claim dont go in2out strict}
Let~$\{c,d\}$ be an edge of~$\mg$ with~$c\in \Vin$ and~$d\in \Vout$.
Each temporal path that traverses the edge~$\{c,d\}$ from~$c$ to~$d$ in~$\G$ ends in a vertex of~$V(\mP[d])$.
\end{observation}

We show that there is a nontrivial closed tcc in~$\G$ if and only if~$G$ has a clique of size~$k$.
In the proof, we will highly rely on the relative order of the time intervals (see~\Cref{tab times strict}) and follow the same ideas as the proof behind~\Cref{hardness undir happy}.

\subsection{A size-$k$ clique in~$G$ implies a nontrivial closed tcc in~$\G$}
In this subsection, we show the one direction of the correctness of the reduction.

\begin{lemma}\label{undirected direction forward strict}
If there is a clique of size~$k$ in~$G$, then there is a nontrivial closed tcc of size~$\Theta(k)$ in~$\G$.
\end{lemma}
For the remainder of the subsection assume that~$G$ contains a clique of size~$k$.
Let~$X$ be such a clique.
Thus, for each~$i\in [1,k]$, there is a unique vertex~$x_i$ in~$X\cap V_i$.
Clearly~$x_i = v_i^*$ for each~$i\in [1,k]$, where~$V_i$ has size~$1$.
Note that this implies~$x_1 = \alpha$, $x_2 = \alpha'$, $x_{k-1} = \omega'$, and~$x_k = \omega$.

Recall that for each~$x_i\in X \setminus \{\omega,\omega'\}$, there is a vertex~$\vin[x_i]$ in~$\Vin$.
Similarly, for each~$x_i\in X \setminus \{\alpha,\alpha'\}$, there is a vertex~$\vout[x_i]$ in~$\Vout$.
Furthermore, for each~$(x,y)\in A_\conn$, there is an edge~$\{\vout[x],\vin[y]\}$ that received the out2in-time in~$\mg$.

Let~$C_X$ denote these vertices from~$\Vin\cup\Vout$.
That is, $$C_X:= \{\vin[x_i]\mid 1 \leq i \leq k-2\}\cup \{\vout[x_i]\mid 3\leq i \leq k\}.$$

We show that~$S:= X \cup \bigcup_{c\in C_X} V(\mP[c])$ is a closed tcc.
Here, recall that~$\mP[c]$ denotes the connector path that contains~$c$ and that we denote the source of~$\mP[c]$ by~$s^c$ and the sink of~$\mP[c]$ by~$t^c$.

First, we show that each vertex of~$S\setminus X$ can reach each vertex of~$S$ in~$\G[S]$ and vice versa.
\begin{lemma}\label{reachability outside V is trivial strict}
In~$\G[S]$, each vertex of~$S\setminus X$ can reach each vertex of~$S$, and each vertex of~$S$ can reach each vertex of~$S\setminus X$. 
\end{lemma}
\begin{proof}
For each~$c\in C_X$, each vertex of~$V(\mP[c])$ can reach the sink~$t^c$ by following the red path in~$\mP[c]$, thus arriving at~$t^c$ prior to the~$\alpha$-time.
Afterwards, the edge~$\{t^c,\alpha\}$ can be traversed during the~$\alpha$-time, thus reaching~$\alpha$ prior to the selection-times.
Similarly, for each~$d\in C_X$, $\omega$ can reach each vertex of~$V(\mP[d])$ by traversing the edge~$\{\omega,s^d\}$ during the~$\omega$-time and afterwards traversing the blue path of~$\mP[d]$ entirely.
Note that these paths starting from~$\omega$ start with the~$\omega$-time, which is strictly larger than the selection times.
We will show that in~$\mg[X]$, (i)~each vertex of~$V$ can be reached from~$\alpha$ and (ii)~each vertex of~$V$ can reach~$\omega$.
Recall that~$\mg[V]$ (and thus~$\mg[X]$) has only edges that receive labels from the selection-times.
This then proves the statement, as each vertex of~$\mP[c]$ can reach~$\alpha$ prior to the selection-times, and each vertex of~$\mP[d]$ can be reached from~$\omega$ by using only edges with time labels larger than the selection-times.

For the vertices of~$\{\alpha,\alpha',\omega,\omega'\} \cup \{v_\ell^*\mid \ell\in [3,8]\}$, both statements follow by the fact that $$(\alpha,\alpha',v_3^* , v_4^* , v_5^* , v_6^*,v_7^*,v_8^*,\omega',\omega)$$ is a temporal path in~$\mg$ and thus in~$\mg[S]$ (see~\Cref{fig undir strict selection}).
It remains to show the statement for each vertex~$x_\ell$ with~$\ell\in [9,k-2]$.
This implies that there is some~$i\in [8,k-8]$ with~$i\mod 6 = 2$, such that~$\ell\in [i+1,i+6]$.
Then, by construction, $(\alpha,\alpha',v_3^*,v_4^*,v_5^*,v_{i+1}^*,v_{i+2}^*,x_{i+3},v_{i+4}^*,v_{i+5}^*,v_{i+6}^*,\omega',\omega)$, is a temporal path.
Consequently, both statements also hold for vertex~$x_\ell$, which by the previous discussion implies that each vertex of~$S\setminus X$ can reach each vertex of~$X$ in~$\mg[S]$ and vice versa.
\end{proof}

It thus remains to show that the vertices of~$X$ can reach each other in~$\G[S]$.

\begin{lemma}
In~$\G[S]$, each vertex of~$X$ can reach each vertex of~$X$. 
\end{lemma}
\begin{proof}
By the proof of~\Cref{reachability outside V is trivial strict}, we immediately get that~$\alpha$ and~$\alpha'$ can reach each vertex of~$X$ in~$\mg[S]$, and~$\omega$ and~$\omega'$ can be reached by each vertex of~$X$ in~$\mg[S]$.
As we assume that these four vertices are universal in~$G$, $A_\conn$ contains the pairs~$(x,\alpha)$ and~$(x,\alpha')$ for each~$x\in V\setminus \{\alpha,\alpha'\}$.
Thus, $(x,\vout[x],\vin[\alpha],\alpha)$ and~$(x,\vout[x],\vin[\alpha'],\alpha')$ are temporal paths in~$\mg[S]$.
This shows that each vertex of~$X$ can reach both~$\alpha$ and~$\alpha'$ in~$\mg[S]$, and vice versa.
Similarly, for each~$y\in V\setminus \{\omega,\omega'\}$, $(\omega,\vout[\omega],\vin[y],y)$ and~$(\omega',\vout[\omega'],\vin[y],y)$ are temporal paths in~$\mg[S]$.
This shows that both~$\omega$ and~$\omega'$ can reach each vertex of~$X$ in~$\mg[S]$, and vice versa by the previous arguments.

For the remainder of the proof, focus on distinct vertices~$x_p$ and~$x_q$ of~$X$ with~$3 \leq p < q  \leq k-2$.
Note that if~$(x_p,x_q)\in A_\conn$, then~$(x_p,\vout[x_p],\vin[x_q],x_q)$ is a temporal path in~$\mg[S]$, and if~$(x_q,x_p)\in A_\conn$, then~$(x_q,\vout[x_q],\vin[x_p],x_p)$ is a temporal path in~$\mg[S]$.
As~$X$ is a clique in~$G$, $\{x_p,x_q\}$ is an edge of~$G$, which implies that~$(x_p,x_q)\in A_\conn$.
That is, $x_p$ can reach~$x_q$ in~$\mg[S]$.
Moreover, $\{x_p,x_q\}$ is an edge of~$G$, $A_\conn$ also contains~$(x_q,x_p)$ except for the two cases:
\begin{itemize}
\item $p=3$ and~$q=5$ and
\item $p=i-2$ and~$q=i$ for some~$i\in [8,k-8]$ with~$i\mod 6 = 2$.
\end{itemize}
For the first case, $(v_5^*,v_6^*,v_7^*,v_8^*,v_3^*)$ is a temporal path in~$\mg[S]$, and for the second case,
 $$(v_i^*,v_{i+1}^*,v_{i+2}^*,x_{i+3},v_{i+4}^*,v_{i+5}^*,v_{i+6}^*,v_{i-2}^*)$$ is a temporal path in~$\mg[S]$ (see~\Cref{fig undir strict selection}).
Thus, $x_p$ and~$x_q$ can reach each other. 
Hence, each vertex of~$X$ can reach each other vertex of~$X$.
\end{proof}

As a consequence, each vertex of~$S$ can reach each other vertex of~$S$ in~$\G[S]$.
That is, $\G[S]$ is temporally connected and thus~$S$ is a closed tcc in~$\G$.
As~$S$ is a closed tcc in~$\G$ of size~$k + 20 (2\cdot (k-2)) = 41k-80$ and~$k$ is larger than~$3$, this implies that~$\G$ contains a nontrivial closed tcc, which proves~\Cref{undirected direction forward strict}.

With the same arguments above, one can also show that~$S \cup (V'\setminus V)$ is a closed tcc, that is, we can add the vertices of the whole connector gadget (the vertices of all connector paths) to each nontrivial closed tcc of~$\mg$ without violating the property of being a closed tcc.
We will show this in~\Cref{imp 4 strict}.

\subsection{A nontrivial closed tcc in~$\G$ implies a size-$k$ clique in~$G$}
In this subsection, we show the other direction of the correctness of the reduction.

\begin{lemma}\label{undirected direction backward strict}
If there is a nontrivial closed tcc in~$\G$, then~$G$ contains a clique of size~$k$.
\end{lemma}

For the remainder of this subsection, assume that~$\G$ contains a nontrivial closed tcc.
Let~$S$ be a nontrivial closed tcc in~$\G$ of minimal size.
We will carefully analyze the structure of~$S$ and eventually show that~$S\cap V$ is a clique of size~$k$ in~$G$.
As the \foot~$G'$ of~$\G$ contains no triangle, we immediately get that~$S$ has size at least~$4$.
Due to~\Cref{intermediate}, each vertex~$v$ of~$S$ is an intermediate vertex of at least one temporal path in~$\mg[S]$, as otherwise~$\mg[S\setminus\{v\}]$ would be still temporally connected.
Since~$S$ has size at least~$4$, this would then contradict the minimality of~$S$.

Again, we follow the same ideas as the proof behind~\Cref{hardness undir happy}:
\begin{itemize}
\item Firstly, we show that~$S$ contains at least one vertex that is not from~$V$.
\item Secondly, we derive from this that~$S$ contains~$\alpha$ and~$\omega$.
\item Finally, we show that this implies that~$S$ contains one vertex of each color class, and that~$S\cap V$ is a clique in~$G$.
\end{itemize}
The proofs for these statements differ slightly from the ones for the happy case, as there are more edges than before and the labels are different between the vertices of~$V$.

\begin{lemma}\label{undir not only from V strict}
$S$ contains at least one vertex of~$V'\setminus V$.
\end{lemma}
\begin{proof}
Assume towards a contradiction that~$S\subseteq V$.
First, we show that~$S$ contains no vertex of~$\{\alpha,\alpha',\omega,\omega',v_3^*,v_4^*,v_5^*\}$.
Clearly, $S$ cannot contain~$\alpha$ or~$\omega$, as both have only one neighbor in~$V$ each.
Therefore, $S$ cannot contain~$\alpha'$, as~$\alpha'$ has only two neighbors in~$V$, for which one is~$\alpha$. 
Similarly, since each edge incident with~$\omega'$ that does not end in~$\omega$, has label~$11$, no temporal path in~$\mg[S]$ could have~$\omega'$ as an intermediate vertex.
This is due to the minimality of~$S$ and~\Cref{intermediate}.

Now observe that~$v_4^*$ has only two incident edges, one labeled with~$3$ and one labeled with~$4$.
As~$S$ has size at least~$4$, there thus is a vertex~$w$ in~$S$ that is not incident with any of these edges.
By the above, $w$ is neither~$\alpha$ nor~$\alpha'$.
Due to the labeling, each edge incident with~$w$ has label strictly larger than~$4$, which would contradict the assumption that~$\mg[S]$ is temporally connected.
Thus, $v_4^*$ is not in~$S$, which implies that~$v_3^*$ would have at most one neighbor in~$S$, and that each edge between~$v_5^*$ and any vertex of~$S$ would have label~$5$.
Due to~\Cref{intermediate}, this thus also excludes both~$v_3^*$ and~$v_5^*$ from~$S$.
Hence, $S$ contains no vertex of~$\{\alpha,\alpha',\omega,\omega',v_3^*,v_4^*,v_5^*\}$.

Now let~$\ell$ be the smallest index of~$[1,k]$ with~$S\cap V_\ell \neq \emptyset$. 
By the above argumentation, $\ell \geq 6$ and~$\ell \leq k-5$.
The latter is due to the fact that~$S$ has size at least~$4$, $S$ contains neither~$\omega$ nor~$\omega'$, and~$V_{k-3}$ and~$V_{k-2}$ contain only one vertex each.
This implies that there is some~$i\in [8,k-8]$ with~$i\mod 6 = 2$, such that~$\ell \in [i-2,i+3]$. 
Note that each of the vertices of~$\{v_{i-2}^*,v_{i-1}^*,v_{i}^*,v_{i+1}^*\}$ has at most two neighbors in~$\mg[S]$, as all their other neighbors are not in~$S$ (by the previous discussion and the fact that~$\ell \geq i-2$), namely~$v_5^*$, $\omega'$, the vertices of~$V_{i-3}$, and (i)~$v_{3}^*$ in the case of~$i = 8$, or (ii)~$v_{i-8}^*$ in the case of~$i>8$.
By the fact that~$S$ is~$2$-vertex-connected (see~\Cref{connectivity}), this implies that~$S$ contains either all of~$\{v_{i-2}^*,v_{i-1}^*,v_{i}^*,v_{i+1}^*\}$ or none such vertex.
However, the label on the edges incident with~$v_{i+1}^*$ are labeled with~$5$ and~$6$, whereas the edges incident with~$v_{i-1}^*$ are labeled with~$9$ and~$10$.
Thus, there cannot be a temporal path in~$\mg[S]$ from~$v_{i-1}^*$ to~$v_{i+1}^*$.
This implies that~$S$ contains no vertex from~$\{v_{i-2}^*,v_{i-1}^*,v_{i}^*,v_{i+1}^*\}$.
For vertex~$v_{i+2}^*$, this implies that all incident edges towards vertices that are in~$S$ (namely the vertices of~$V_{i+3}$) all receive label~$7$.
This implies that~$v_{i+2}^*$ cannot be in~$S$, as otherwise, there would be no temporal path in~$\mg[S]$ that has~$v_{i+2}^*$ as intermediate vertex.
Finally, $S$ cannot contain a vertex of~$V_{i+3}$, as each such vertex could have at most one neighbor in~$\mg[S]$, which would contradict the fact that~$S$ is~$2$-vertex-connected.
This contradicts the assumption that~$S$ contains a vertex from~$V_\ell$ with~$\ell\in [i-2,i+3]$.
We thus conclude that~$S$ is not a subset of~$V$.
\end{proof}

\begin{lemma}\label{undir both al and om strict}
$S$ contains~$\alpha$ and~$\omega$.
\end{lemma}
\begin{proof}
Due to~\Cref{undir not only from V strict}, $S$ contains at least one vertex outside of~$V$.
Since~$V'\setminus V$ only consists of the union of the vertex sets of the connector paths, this implies that there is some~$c\in \Vin\cup \Vout$, such that~$S$ contains at least one vertex of~$\mP[c]$.

By the same arguments used to show~\Cref{if s contains neither al nor om} in the proof of~\Cref{hardness undir happy}, we obtain the following.

\begin{claim}\label{if s contains neither al nor om strict}
Assume that~$S$ contains a vertex of some connector path~$\mP[c']$.
Then~$S$ contains at least two neighbors of~$\mP[c']$.
Moreover, if~$S$ does not contain both~$\alpha$ and~$\omega$, then~$S$ contains~$c'$.
\end{claim}

This implies that~$S$ contains~$c$ and at least two neighbors of~$\mP[c]$.
Based on this claim, we now show that~$S$ contains at least one of~$\alpha$ and~$\omega$, if~$S$ contains a vertex of~$V$ and a vertex of~$\Vin\cup \Vout$ that are not adjacent.

\begin{claim}\label{claim implies al or om}
It holds that:
\begin{enumerate}
\item\label{item claim implies om} If~$S$ contains a vertex~$x\in V$ and a vertex~$\vout[y]\in \Vout$ with~$x\neq y$, then~$S$ contains~$\omega$.
\item\label{item claim implies al} If~$S$ contains a vertex~$x\in V$ and a vertex~$\vin[y]\in \Vin$ with~$x\neq y$, then~$S$ contains~$\alpha$.
\end{enumerate}
\end{claim}
\begin{claimproof}
We start by showing the first statement.
Since~$\mg[S]$ is temporally connected, there is a temporal path~$P$ from~$x$ to~$\vout[y]$.
Moreover, $x$ and~$\vout[y]$ are not adjacent, as~$x\neq y$.
If~$x=\omega$, the statement is proven.
So, assume that~$x\neq \omega$.
First, we show that~$P$ has to start with an edge towards a different vertex of~$V$.
Consider the incident edges of~$x$ that do not have their other endpoint in~$V$.
These are (i)~the edges towards the sinks~$t^c$ of each connector path if~$x=\alpha$, (ii)~the edge~$\{x,\vout[x]\}$ if~$x\in V \setminus \{\alpha,\alpha'\}$, and (iii)~the edge~$\{x,\vin[x]\}$ if~$x\in V \setminus \{\omega,\omega'\}$.
Edges of the first type can only be traversed at the~$\alpha$-time and~$t^c$ has only one incident edge with a larger label.
However, this edge receives the largest assigned label among all edges and does end in a vertex of~$\mP[c]$ distinct from~$c$.
Hence, $P$ cannot traverse any edge of the first type.
Moreover, $P$ cannot traverse any of the edges of the second or the third type due to~\Cref{claim go out means go in strict} and~\Cref{claim dont go in strict} since~$x\neq y$.
This implies that~$P$ traverses one edge towards a vertex of~$V$ at the selection times.
Consequently, $P$ cannot traverse any out-edge after that, as these edges receive a strictly smaller label.
Moreover, $P$ cannot traverse any in-edge due to~\Cref{claim dont go in strict}.
The only other edges that connect a vertex from~$V$ to a vertex of~$V'\setminus V$ are incident with~$\omega$.
This implies that~$S$ contains~$\omega$, as~$P$ has to visit~$\omega$ to reach~$\vout[y]$.

We now show the second statement.
Since~$\mg[S]$ is temporally connected, there is a temporal path~$P$ from~$\vin[y]$ to~$x$.
Let~$c = \vin[y]$.
Recall that the connector path~$\mP[c]$ contains only three vertices that have neighbors outside of the connector path: the source~$s^c$, the sink~$t^c$, and the docking point~$c$.
By definition of a connector path, $c$ cannot reach~$s^c$ via only edges in~$\mP[c]$.
Thus, $P$ leaves~$\mP[c]$ via an edge incident with~$c$ or~$t^c$.
First, we show that~$P$ cannot leave~$\mP[c]$ without an edge incident with~$c$.
The only such edges are (i)~the edge~$\{c=\vin[y],y\}$ and (ii)~edges that receive the out2in-time towards vertices~$\vout[z]$ of~$\Vout$.  
By~\Cref{claim dont go in strict}, taking edge~$\{c=\vin[y],y\}$ would guarantee~$P$ to end in~$y$ and not in~$x$.
Similarly, traversing an edge towards a vertex~$\vout[z]\in \Vout$ would guarantee~$P$ to end in a vertex of~$\mP[{\vout[z]}]$ and not in~$x$ due to~\Cref{claim dont go in2out strict}.
Thus, $P$ leaves~$\mP[c]$ via an edge incident with~$t^c$.
As the only edge incident with~$t^c$ that leads to a vertex outside of~$\mP[c]$ is the edge~$\{t^c,\alpha\}$, we conclude that~$P$ visits~$\alpha$.
Hence, $\alpha$ is in~$S$.
\end{claimproof}

Based on this claim, we now distinguish different cases on how~$S$ intersects with~$\Vin\cup \Vout$.
First, observe that if~$S$ contains two nonadjacent vertices of~$\Vin\cup\Vout$, then~$S$ contains a vertex of~$V$, as no temporal path between nonadjacent vertices of~$\Vin\cup \Vout$ avoids all vertices of~$V$.
This is due to the fact that all edges between distinct connector paths all receive the same label, which prevents nonadjacent docking points from reaching each other in~$\mg[V'\setminus V]$.

\begin{itemize}
\item \textbf{Case 1:} $S$ contains two distinct vertices~$\vout[x]$ and~$\vout[y]$ of~$\Vout$\textbf{.}
By the above and the fact that vertices of~$\Vout$ are pairwise nonadjacent, there is some vertex~$v\in V\cap S$.
As~$x\neq y$, $v\neq x$ or~$v\neq y$.
\Cref{item claim implies om} of \Cref{claim implies al or om} then implies that~$\omega\in S$.
Now, if~$S$ contains a vertex~$\vin[z] \in \Vin$, then~\Cref{item claim implies al} of \Cref{claim implies al or om} implies that~$\alpha$ is in~$S$ too, as~$z\neq \omega$ by the fact that~$\vin[z]$ only exists for~$z\in V\setminus \{\omega,\omega'\}$.
So assume in the following that~$S$ contains no vertex of~$\Vin$.
Recall that~$S$ contains at least two neighbors of each of the connector paths~$\mP[{\vout[x]}]$ and~$\mP[{\vout[y]}]$ (see \Cref{if s contains neither al nor om strict}).
By construction, for each~$z\in V\setminus \{\alpha,\alpha'\}$, the neighbors of~$\mP[{\vout[z]}]$ are from~$\Vin \cup \{\alpha,\omega,z\}$, which implies that~$S$ contains~$\alpha$ and~$\omega$, or~$\omega$ and both~$x$ and~$y$.
So assume the latter is the case.
We show that this would imply that~$S$ contains a vertex of~$\Vin$, which we assumed is not the case.
The latter case implies that~$\omega$ is distinct from both~$x$ and~$y$, as otherwise, the respective connector path would have~$\omega$ as its only neighbor in~$S$.
This implies that~$\{x,y\}$ contains a vertex~$z$ of~$V\setminus \{\omega,\omega'\}$.
Hence, there must be a temporal path~$P'$ from~$\omega$ to~$z$ in~$\mg[S]$.
The path~$P'$ cannot be completely in~$\mg[V]$, as both vertices are not adjacent (since~$z \neq \omega'$), the smallest label incident with~$\omega$ is~$12$, and the largest label incident with~$z$ is at most~$11$.
By~\Cref{claim paths not in V strict}, this implies that~$P'$ visits~$\vin[z]$, which would contradict our assumption that~$S$ avoids all vertices from~$\Vin$.
Hence, $\alpha\in S$.

\item \textbf{Case 2:} $S$ contains two distinct vertices~$\vin[x]$ and~$\vin[y]$ of~$\Vin$\textbf{.}
The idea is essentially the same as in the previous case.
The fact that vertices of~$\Vin$ are pairwise nonadjacent implies that there is some vertex~$v\in V\cap S$.
As~$x\neq y$, $v\neq x$ or~$v\neq y$.
\Cref{item claim implies al} of \Cref{claim implies al or om} then implies that~$\alpha\in S$.
Now, if~$S$ contains a vertex~$\vout[z] \in \Vout$, then~\Cref{item claim implies om} of \Cref{claim implies al or om} implies that~$\omega$ is in~$S$ too, as~$z\neq \alpha$ by the fact that~$\vout[z]$ only exists for~$z\in V\setminus \{\alpha,\alpha'\}$.
So assume in the following that~$S$ contains no vertex of~$\Vout$.
Recall that~$S$ contains at least two neighbors of each of the connector paths~$\mP[{\vin[x]}]$ and~$\mP[{\vin[y]}]$ (see \Cref{if s contains neither al nor om strict}).
By construction, for each~$z\in V\setminus \{\omega,\omega'\}$, the neighbors of~$\mP[{\vin[z]}]$ are from~$\Vout \cup \{\alpha,\omega,z\}$, which implies that~$S$ contains~$\alpha$ and~$\omega$, or~$\alpha$ and both~$x$ and~$y$.
So assume the latter is the case.
We show that this would imply that~$S$ contains a vertex of~$\Vout$, which we assumed is not the case.
The latter case implies that~$\alpha$ is distinct from both~$x$ and~$y$, as otherwise, the respective connector path would have~$\alpha$ as its only neighbor in~$S$.
This implies that~$\{x,y\}$ contains a vertex~$z$ of~$V\setminus \{\alpha,\alpha'\}$.
Hence, there must be a temporal path~$P'$ from~$z$ to~$\alpha$ in~$\mg[S]$.
The path~$P'$ cannot be completely in~$\mg[V]$, as both vertices are not adjacent (since~$z \neq \alpha'$), the largest label incident with~$\alpha$ is~$1$, and the smallest label incident with~$z$ is at least~$2$.
By~\Cref{claim paths not in V strict}, this implies that~$P'$ visits~$\vout[z]$, which would contradict our assumption that~$S$ avoids all vertices from~$\Vout$.
Hence, $\omega\in S$.

\item \textbf{Case 3:} $S$ contains exactly one vertex~$\vout[x]$ of~$\Vout$ and exactly one vertex~$\vin[y]$ of~$\Vin$\textbf{.}
First consider~$x=y$.
This implies that~$x\in V\setminus \{\alpha,\alpha',\omega,\omega'\}$ and that~$\{\vout[x],\vin[y]\}$ is not an edge of~$\mg$.
Hence, $S$ has to contain a neighbor~$z$ of~$\mP[{\vin[x]}]$ distinct from~$x$.
The only neighbors of~$\mP[{\vin[x]}]$ besides~$x$ are from~$\Vout\cup \{\alpha,\omega\}$.\\
As~$\vout[x]$ is the only vertex of~$S\cap \Vout$ by assumption, this implies that~$S$ contains~$\alpha$ or~$\omega$.
By~\Cref{claim implies al or om} we then derive that~$S$ contains~$\alpha$ and~$\omega$, as~$y\notin \{\alpha,\omega\}$.

Now consider~$x\neq y$.
Again, $S$ has at least two neighbors of~$\mP[{\vout[x]}]$.
As~$\vout[x]$ and~$\vin[y]$ are the only vertices of~$S\cap (\Vin\cup \Vout)$, one neighbor of~$\mP[{\vout[x]}]$ in~$S$ is from~$\{\alpha,\omega,x\}$.
In all cases, \Cref{claim implies al or om} implies that~$S$ contains~$\alpha$ and~$\omega$: If~$S$ contains~$\alpha$ then~$\vout[x]\in S$ implies~$\omega\in S$; if~$S$ contains~$\omega$, then $\vin[y]\in S$ implies~$\alpha\in S$; if~$S$ contains~$x$, then by~$x\neq y$, $\vin[y]\in S$ implies that~$\alpha\in S$, which by the previous argument shows~$\omega\in S$.

\item \textbf{Case 4:} $S$ contains exactly one vertex~$c$ of~$\Vin \cup \Vout$\textbf{.} 
Since~$S$ has at least two neighbors of~$\mP[c]$ and~$c$ is the only vertex of~$\Vin\cup\Vout$ in~$S$, $S$ contains at least two vertices from~$\{\alpha,\omega,x\}$, where~$x$ is the vertex of~$V$ fulfilling~$c\in \{\vin[x],\vout[x]\}$.
In particular, $S$ contains at least one of~$\alpha$ and~$\omega$.
If~$S$ contains both, we are done. 
So assume in the following that this is not the case.
We will show that this leads to a contradiction.

Firstly, assume that~$S$ contains~$\omega$.
If~$S$ contains at least one vertex~$v$ of~$V\setminus \{\omega,\omega'\}$, then every temporal path (as shown in Case 1) from~$\omega$ to~$v$ contains at least two vertices from~$\Vin\cup \Vout$, which is not possible by assumption that~$S$ contains one vertex of that set.
So assume in the following that~$S \cap V \subseteq \{\omega, \omega'\}$.
This implies that~$x\in \{\omega, \omega'\}$, which further implies that~$c\in \{\vout[\omega],\vout[\omega']\}$.
If~$c = \vout[\omega]$, then we get a contradiction to the assumption that~$\alpha\notin S$, as~$\alpha$ is the only neighbor of the connector path~$\mP[c]$ outside of~$\Vin\cup\Vout \cup \{\omega\}$ and~$S$ contains at least two neighbors of~$\mP[c]$.
For~$c=\vout[\omega']$, this would imply that~$S \subseteq Q_4:=V(\mP[c]) \cup \{\omega,\omega'\}$.
However, $\G[Q_4]$ fulfills the requirement of the temporal graph~$\G_4$ described in Item~\ref
{gadget plus omega plus xk} of~\Cref{gadget extensions no tcc}.
As~$\G_4$ does not contain a nontrivial closed tcc, this contradicts the assumption that~$S$ is a closed tcc in~$\mg$.

Secondly, assume that~$S$ contains~$\alpha$.
The arguments are essentially the same as for the case that~$S$ contains~$\omega$.
If~$S$ contains at least one vertex~$v$ of~$V\setminus \{\alpha,\alpha'\}$, then every temporal path (as shown in Case 2) from~$v$ to~$\alpha$ contains at least two vertices from~$\Vin\cup \Vout$, which is not possible by assumption that~$S$ contains one vertex of that set.
So assume in the following that~$S \cap V \subseteq \{\alpha, \alpha'\}$.
This implies that~$x\in \{\alpha, \alpha'\}$, which further implies that~$c\in \{\vin[\alpha],\vin[\alpha']\}$.
If~$c = \vin[\alpha]$, then we get a contradiction to the assumption that~$\omega\notin S$, as~$\omega$ is the only neighbor of the connector path~$\mP[c]$ outside of~$\Vin\cup\Vout \cup \{\alpha\}$ and~$S$ contains at least two neighbors of~$\mP[c]$.
For~$c=\vin[\alpha']$, this would imply that~$S \subseteq Q_2:=V(\mP[c]) \cup \{\alpha,\alpha'\}$.
However, $\G[Q_2]$ fulfills the requirement of the temporal graph~$\G_2$ described in Item~\ref
{gadget plus alpha plus x1} of~\Cref{gadget extensions no tcc}.
As~$\G_2$ does not contain a nontrivial closed tcc, this contradicts the assumption that~$S$ is a closed tcc in~$\mg$.

In both cases, we get a contradiction.
We conclude that~$S$ contains~$\alpha$ and~$\omega$, if~$S$ contains only a single vertex of~$\Vin\cup\Vout$. 
\end{itemize}
Note that these cases are exhaustive and all show that~$S$ contains~$\alpha$ and~$\omega$.
\end{proof}

\begin{lemma}\label{undir path from al to om strict}
$S$ contains a vertex~$x_i\in V_i$ for each~$i\in [1,k]$.
\end{lemma}
\begin{proof}
Due to~\Cref{undir both al and om strict}, we know that~$S$ contains~$\alpha$ and~$\omega$.
As~$\G[S]$ is temporally connected, this implies that there is a temporal path~$P$ from~$\alpha$ to~$\omega$ in~$\G[S]$.
We show that~$P$ contains the vertices~$\{\alpha',v_3^*,v_4^*,v_5^*\}$.
To this end, we show that~$P$ only visits vertices of~$V$.
This is due to~\Cref{claim paths inside V strict} and the fact that~$(\alpha,\omega)\notin A_\conn$.

This implies that the first edge of~$P$ is~$\{\alpha,\alpha'\}$.
Observe that~$\omega$ has only one incident edge towards another vertex of~$V$, namely the edge~$\{\omega',\omega\}$ with label~$12$.
Hence, $P$ has to traverse this edge at that time step, which implies that~$P$ has to reach~$\omega'$ prior to time~$12$.
As each other incident edge of~$\omega'$ towards a vertex of~$V\setminus \{\omega\}$ receives label~$11$, $P$ traverses such an edge at time~$11$, and thus does not traverse any other edge with label~$11$ prior.
In particular, the edge~$\{v_3^*,v_8^*\}$ is not part of~$P$.
Hence, $P$ has to start with the subpath~$(\alpha,\alpha',v_3^*,v_4^*,v_5^*)$, as each such vertex has only these incident edges (see~\Cref{fig undir strict selection}).
This implies that~$S$ contains all of~$\{\alpha,\alpha',v_3^*,v_4^*,v_5^*\}$.

With this property, we show next that~$S$ contains also the vertices from~$\{v_6^*,v_7^*,v_8^*\}$.
As~$\mg[S]$ is temporally connected, there is a temporal path~$P_5$ from~$v_5^*$ to~$v_3^*$ in~$\mg[S]$.
We show that~$P_5=(v_5^*,v_6^*,v_7^*,v_8^*,v_3^*)$.
By definition, $(v_5^*,v_3^*)\notin A_\conn$.
Thus, \Cref{claim paths inside V strict} implies that~$P_5$ visits only vertices of~$V$.
Now focus on the edges of~$\mg[V]$ incident with~$v_3^*$ and~$v_5^*$.
Each such edge incident with~$v_5^*$ has a label of at least~$4$ and~$\{v_3^*,v_8^*\}$ is the only such edge incident with~$v_3^*$ that receives a larger label.
As~$v_3^*$ and~$v_5^*$ are not adjacent, $P_5$ traverses~$\{v_3^*,v_8^*\}$ at time~$11$.
For~$P_5$ to arrive at~$v_8^*$ prior to time~$11$, $P_5$ has to start with~$(v_5^*,v_6^*,v_7^*,v_8^*)$ (see~\Cref{fig undir strict selection}).

This implies that~$S$ contains all vertices from~$\{\alpha,\alpha',\omega,\omega'\} \cup \{v_\ell^*\mid \ell\in [3,8]\}$.
We now conclude the proof via an inductive argument.
More precisely, we will show that if for some~$i\in [8,k-8]$ with~$i\mod 6 = 2$, $S$ contains the vertices~$v_{i-2}^*$ and~$v_{i}^*$, then~$S$ contains for each~$\ell\in [i+1,i+6]$ at least one vertex of~$V_\ell$.
This includes the vertices~$v_{i+6-2}^*$ and~$v_{i+6}^*$, for which the argument then inductively shows that~$S$ contains at least one vertex per color class.
The base case follows by the fact that for~$i=8$, we already showed that~$S$ contains~$v_{6}^*$ and~$v_{8}^*$.
For the inductive case, consider the following claim.

\begin{claim}
Let~$i\in [8,k-8]$ with~$i\mod 6 = 2$, such that~$S$ contains~$v_{i-2}^*$ and~$v_{i}^*$.
Then, $S$ contains at least one vertex of~$V_\ell$ for each~$\ell\in[i+1,i+6]$.
\end{claim}
\begin{claimproof}
As~$\mg[S]$ is temporally connected, there is a temporal path~$P_i$ from~$v_i^*$ to~$v_{i-2}^*$ in~$\mg[S]$.
We show that~$P_i=(v_i^*,v_{i+1}^*,v_{i+2}^*,x_{i+3},v_{i+4}^*,v_{i+5}^*,v_{i+6}^*,v_{i-2}^*)$ for some vertex~$x_{i+3}\in V_{i+3}$.
By definition, $(v_{i}^*,v_{i-2}^*)\notin A_\conn$.
Thus, \Cref{claim paths inside V strict} implies that~$P_i$ visits only vertices of~$V$.
Now focus on the edges of~$\mg[V]$ incident with~$v_{i-2}^*$ and~$v_{i}^*$.
Each such edge incident with~$v_{i-2}^*$ has a label of at most~$11$, which implies that~$P_i$ arrives at~$v_{i-2}^*$ at time at most~$11$.
The only incident edges of~$v_i^*$ that have a smaller label than~$11$ are (i)~the edge~$\{v_i^*,v_{i-1}^*\}$ with label~$10$ and (ii)~edge~$\{v_i^*,v_{i+1}^*\}$ with label~$5$.
The path~$P_i$ cannot traverse the edge~$\{v_i^*,v_{i-1}^*\}$, as~$v_{i-1}^*$ has no incident edge with label larger than~$10$.
Thus, $P_i$ starts with~$\{v_i^*,v_{i+1}^*\}$ at time~$5$.
Afterwards, $P_i$ has to traverse the edge~$\{v_{i+1}^*,v_{i+2}^*\}$, as this is the only edge incident with~$v_{i+1}^*$ with a label larger than~$5$, namely label~$6$.
All other edges incident with~$v_{i+2}^*$ lead to vertices of~$V_{i+3}$.
Hence, $P_i$ traverses the edge~$\{v_{i+2}^*,x_{i+3}\}$ at time~$7$ for some~$x_{i+3}\in V_{i+3}$, and then follows the only other edge incident with~$x_{i+3}$, namely~$\{x_{i+3},v_{i+4}^*\}$ at time~$8$.
Now, $v_{i+4}^*$ has two incident edges with labels larger than~$8$, namely~$\{v_{i+4}^*,v_{i+12}^*\}$ with label~$11$ and~$\{v_{i+4}^*,v_{i+5}^*\}$ with label~$9$.\footnote{The first edge only exists if~$i \leq k-12$.}
As~$P_i$ arrives at~$v_{i-2}^*$ at time at most~$11$, the only way for~$P_i$ to continue is to use the edge~$\{v_{i+4}^*,v_{i+5}^*\}$ at time~$9$.
Finally, $P_i$ has to continue to the only other neighbor of~$v_{i+5}^*$, namely~$v_{i+6}^*$.
As a consequence, $P_i$ visits all the vertices of~$\{v_\ell^*\mid \ell\in\{i-1\} \cup [i,i+2]\cup [i+4,i+6]\}$ and some vertex~$x_{i+3}$ of~$V_{i+3}$.
This proves the statement, as~$P_i$ is a temporal path in~$\mg[S]$.
\end{claimproof}

Consequently, $S$ contains at least one vertex of each color class.
\end{proof}

Let~$X:= \{x_i\mid 1\leq i \leq k\}$.
In the remainder of the proof we show that~$X$ is a clique in~$G$.
This then implies that~$(G,k)$ is a yes-instance of~MCC.

\begin{lemma}
$X$ is a clique of size~$k$ in~$G$.
\end{lemma}
\begin{proof}
Let~$1\leq p < q \leq k$.
We show that~$\{x_p,x_q\}$ is an edge of~$G$.
Recall that we assume that each vertex of a size-1 color class is adjacent to each other vertex of~$V$ in~$G$.
As~$V_\ell$ has size one for each~$\ell \in [1,k]$ with~$\ell < 11$ or~$\ell \mod 6 \neq 5$, we thus only need to consider~$p,q \geq 11$, $p\mod 6=5$, and~$q\mod 6 = 5$.
This implies that~$q \geq p + 6$.
As~$\mg[S]$ is temporally connected, there is a temporal path~$P$ from~$x_p$ to~$x_q$ in~$\mg[S]$.
Note that~$P$ does not contain only vertices of~$V$.
This is due to the facts that (i)~each edge incident with~$x_p$ in~$\mg[V]$ receives a label of at least~$7$, (ii)~each edge incident with~$x_q$ in~$\mg[V]$ receives a label of at most~$8$, and (iii)~$x_p$ and~$x_q$ have no common neighbor, as~$q \geq p + 6$. 
Hence, \Cref{claim paths not in V strict} implies that~$(x_p,x_q)\in A_\conn$.
By definition of~$A_\conn$, this then implies that~$\{x_p,x_q\}$ is an edge of~$G$.
Consequently, $X$ is a clique in~$G$.
\end{proof}

This thus implies that~$G$ has a clique of size~$k$, which proves~\Cref{undirected direction backward strict}.
The correctness of the reduction now follows from~\Cref{undirected direction forward strict} and~\Cref{undirected direction backward strict}.

\subsection{Implications}

We argue the following, which is important for Implication~\ref{imp Counting}.

\begin{lemma}\label{imp 4 strict}
Let~$X$ be a clique of size~$k$ in~$G$.
Then $S^* = X \cup (V' \setminus V)$ is a maximal closed tcc in~$\G$.
\end{lemma}
\begin{proof}
First, we argue that~$S^*$ is a closed tcc.
This essentially follows by the same arguments that were used to prove~\Cref{undirected direction forward strict}.
Let~$S$ be the closed tcc for which the existence is shown in~\Cref{undirected direction forward strict} based on the existence of~$X$.
By definition of~$S$, $S^* = S \cup (V'\setminus V)$, that is, we can add the vertices of the whole connector gadget (the vertices of all connector paths) to~$S$ without violating the property of being a closed tcc.
This is due to the following: The reachability between vertices of~$S$ is done inside of~$\G[S]$, and each other vertex is part of some connector path and can reach all other vertices (and can be reached by them) via the same arguments used in the proof of~\Cref{reachability outside V is trivial strict}.

Next, we show that~$S^*$ is maximal.
Assume towards a contradiction that there is a closed tcc~$S' \supsetneq S^*$.
This set contains at least one vertex~$y_i$ of~$V\setminus X$, as $S^*$ contains all of~$V'\setminus V$.
Let~$V_i$ be the color class of~$y_i$.
We show that there is no temporal path from~$y_i$ to~$x_i$ in~$\G$.
This then contradicts the assumption that~$\G[S']$ is temporally connected.
By construction and the fact that $V_i$ has size more than~$1$, both~$x_i$ and~$y_i$ have exactly two neighbors in the selection gadget, namely~$v^*_{i-1}$ and~$v^*_{i+1}$, and the edges towards these two vertices receive labels~$7$ and~$8$, respectively (see~\Cref{fig undir strict selection}).
As we assume that~$\G[S']$ is temporally connected, there has to be a temporal path~$P$ from~$y_i$ to~$x_i$.
Since~$\{x_i,y_i\}$ is not an edge of~$G$ and $(y_i, x_i) \notin A_{\mathrm{con}}$, \Cref{claim paths inside V strict} implies that~$P$ visits only vertices of~$V$.
Since the labels along~$P$ strictly increase, the first edge of~$P$ has label~$7$ and the last edge of~$P$ has label~$8$.
No label lies strictly between~$7$ and~$8$, so $P = (y_i, a, x_i)$ for some vertex~$a$.
The first edge forces~$a = v^*_{i-1}$ and the last edge forces~$a = v^*_{i+1}$, a contradiction.
Consequently, there is no temporal path from~$y_i$ to~$x_i$ in~$\G$, a contradiction to $S'$ being a closed tcc.
\end{proof}

\end{document}